\documentclass[10pt,prd,aps,amsfonts,eqsecnum,superscriptaddress,nofootinbib,notitlepage,longbibliography,final]{revtex4-2}

\usepackage[a4paper, left=2cm, right=2cm, top=2cm, bottom=2cm]{geometry}
\usepackage{setspace}
\usepackage[english]{babel}
\addto\extrasenglish{%
	\def\appendixautorefname{Appendix}%
}
\usepackage[utf8]{inputenc}
\usepackage[T1]{fontenc}
\usepackage{lmodern}
\usepackage{amsmath,amsthm,mathdots,amssymb,bm,bbm,mathrsfs,mathtools}
\usepackage[final]{graphicx} 
\usepackage{tikz,pgfplots}\pgfplotsset{compat=1.16}
\usetikzlibrary{positioning,fit,calc}

\usepackage{thm-restate}
\newtheorem{thm}{Theorem}[section]
\newtheorem{cor}{Corollary}[section]

\newtheorem{lem}{Lemma}[section]

\newtheorem{prop}{Proposition}[section]

\newtheorem{rmk}{Remark}[section]
\newtheorem{defn}{Definition}[section]
\newtheorem{assumption}{Assumption}[section]

\theoremstyle{definition}

\newtheorem{eg}{Example}[section]

\usepackage{comment}

\usepackage{mleftright}\mleftright 

\usepackage{enumitem}

\usepackage[colorlinks=true, urlcolor=violet, linkcolor=blue, citecolor=red, hyperindex=true, linktocpage=true, pagebackref=true, draft=false]{hyperref}
\usepackage{microtype,color} 
\usepackage{xcolor}

\renewcommand{\vec}{\bm}

\newcommand{\CB}{\mathcal{B}}
\newcommand{\CC}{\mathcal{C}}
\newcommand{\BC}{\mathbb{C}}

\newcommand{\CE}{\mathcal{E}}
\newcommand{\BE}{\mathbb{E}}
\newcommand{\tBE}{\tilde{\mathbb{E}}}

\newcommand{\CH}{\mathcal{H}}

\newcommand{\CK}{\mathcal{K}}
\newcommand{\CL}{\mathcal{L}}

\newcommand{\CN}{\mathcal{N}}
\newcommand{\BN}{\mathbb{N}}

\newcommand{\CP}{\mathcal{P}}
\newcommand{\BP}{\mathbb{P}}

\newcommand{\BR}{\mathbb{R}}
\newcommand{\CS}{\mathcal{S}}

\newcommand{\CT}{\mathcal{T}}

\newcommand{\CV}{\mathcal{V}}

\newcommand{\BZ}{\mathbb{Z}}

\newcommand{\sA}{\mathsf{A}}
\newcommand{\sB}{\mathsf{B}}
\newcommand{\sC}{\mathsf{C}}
\newcommand{\sR}{\mathsf{R}}
\newcommand{\sL}{\mathsf{L}}
\newcommand{\sM}{\mathsf{M}}
\newcommand{\sI}{\mathsf{I}}
\newcommand{\sX}{\mathsf{X}}
\newcommand{\sY}{\mathsf{Y}}

\newcommand{\vA}{\bm{A}}

\newcommand{\vB}{\bm{B}}

\newcommand{\vC}{\bm{C}}

\newcommand{\vH}{\bm{H}}
\newcommand{\vh}{\bm{h}}
\newcommand{\vI}{\bm{I}}

\newcommand{\vL}{\bm{L}}

\newcommand{\vN}{\bm{N}}
\newcommand{\vO}{\bm{O}}
\newcommand{\vP}{\bm{P}}

\newcommand{\vR}{\bm{R}}
\newcommand{\vS}{\bm{S}}

\newcommand{\vV}{\bm{V}}
\newcommand{\vW}{\bm{W}}
\newcommand{\vX}{\bm{X}}
\newcommand{\vY }{\bm{Y }}
\newcommand{\vZ}{\bm{Z}}
\newcommand{\vsigma}{\bm{ \sigma}}

\newcommand{\vrho}{\bm{ \rho}}
\renewcommand{\L}{\left}
\newcommand{\R}{\right}

\newcommand{\vcH}{\vec{\CH}}
\newcommand{\vcS}{\vec{\CS}}

\newcommand{\dagg}{\dagger}

\newcommand{\vertiii}[1]{{\left\vert\kern-0.25ex\left\vert\kern-0.25ex\left\vert #1 \right\vert\kern-0.25ex\right\vert\kern-0.25ex\right\vert}}
\newcommand{\norm}[1]{\Vert {#1} \Vert}

\newcommand{\lnormp}[2]{\lnorm{#1}_{#2}}

\newcommand{\labs}[1]{\left\vert {#1} \right\vert}
\newcommand{\lnorm}[1]{\left\Vert {#1} \right\Vert}
\newcommand{\e}{\mathrm{e}}

\newcommand{\ri}{\mathrm{i}}
\newcommand{\rd}{\mathrm{d}}
\newcommand*{\tr}{\mathrm{Tr}}

\newcommand{\indicator}{\mathbbm{1}}
\newcommand{\Lword}[1]{\text{Lindbladian}}

\newcommand{\ceil}[1]{\left\lceil{#1}\right\rceil}
\newcommand{\floor}[1]{\left\lfloor{#1}\right\rfloor}

\newcommand{\undersetbrace}[2]{ \underset{#1}{\underbrace{#2}}}

\DeclarePairedDelimiterX{\braket}[1]{\langle}{\rangle}{#1}
\DeclarePairedDelimiterX\ketbra[2]{| }{|}{#1 \delimsize\rangle\!\delimsize\langle #2}	
\DeclarePairedDelimiterX\dotp[2]{\langle}{\rangle}{#1, #2}

\DeclareMathAlphabet{\dutchcal}{U}{dutchcal}{m}{n}
\SetMathAlphabet{\dutchcal}{bold}{U}{dutchcal}{b}{n}
\DeclareMathAlphabet{\dutchbcal} {U}{dutchcal}{b}{n}

\newcommand{\under}[2]{\underbrace{#1}_{\substack{#2}}}

\makeatletter
\DeclareRobustCommand*{\pmzerodot}{%
	\nfss@text{%
		\sbox0{$\vcenter{}$}
		\sbox2{0}%
		\sbox4{0\/}%
		\ooalign{%
			0\cr
			\hidewidth
			\kern\dimexpr\wd4-\wd2\relax 
			\raise\dimexpr(\ht2-\dp2)/2-\ht0\relax\hbox{%
				\if b\expandafter\@car\f@series\@nil\relax
				\mathversion{bold}%
				\fi
				$\cdot\m@th$%
			}%
			\hidewidth
			\cr
			\vphantom{0}
		}%
	}%
}

\makeatletter
\def\l@subsubsection#1#2{}
\makeatother

\usetikzlibrary{quantikz2}

\begin{document}
\renewcommand{\appendixautorefname}{Appendix}
\renewcommand{\chapterautorefname}{Chapter}
\renewcommand{\sectionautorefname}{Section}
\renewcommand{\subsubsectionautorefname}{Section}

\title{Fast Mixing of Quantum Spin Chains at All Temperatures}
\author{Thiago Bergamaschi} 
\affiliation{University of California, Berkeley, CA, USA}
\author{Chi-Fang Chen}
\email{achifchen@gmail.com}

\affiliation{University of California, Berkeley, CA, USA}
\affiliation{Massachusetts Institute of Technology, Cambridge, USA}

\begin{abstract}
It is shown that every one-dimensional Hamiltonian with short-range interactions admits a quantum Gibbs sampler~\cite{chen2023efficient} with a system-size independent spectral gap at all finite temperatures. Consequently, their Gibbs states can be prepared in polylogarithmic depth, and satisfy exponential clustering of correlations, generalizing~\cite{araki1969gibbs}. 
\end{abstract}
\maketitle
\vspace{-1cm}

\begingroup
\small  
\tableofcontents
\endgroup

\section{Introduction}

All evidence suggests that many-body physics in one spatial dimension (1D) is thermodynamically trivial. In 1969, Araki proved that every translation-invariant infinite spin chain satisfies an exponential decay of correlations, and thus admits no finite-temperature phase transitions~\cite{araki1969gibbs}. Ever since, other such manifestations -- the convergence of complex-time evolution~\cite{araki1969gibbs, bluhm2022exponential,perez2023locality}, area laws and the existence of a tensor-network description \cite{VGC04,hastings2006solving, anshu2022area, wolf2008area}, Markov properties~\cite{Kato2019,kuwahara2024clustering}, efficient classical algorithms for computing partition functions and local observables~\cite{fawzi2023subpolynomial, harrow2020classical} -- have gradually provided a comprehensive picture of 1D thermal physics and its complexity.

Nevertheless, one-dimensional many-body \textit{dynamics} can be extremely complex. Random circuits, structured or unstructured (e.g., \cite{brandao_local_tdesign,chen2024incompressibility,schuster2025random}), can produce exceedingly complex quantum states as time progresses; under proper initial states, the Feynman Hamiltonian can perform universal quantum computation \cite{kitaev02}. 
One could argue that such rich quantum behavior partly stems from the exact absence of noise. Yet, quantum fault tolerance schemes can encode any quantum computation into a 1D time-dependent circuit with reset gates, resilient even against extensive noise \cite{aharonov1999faulttolerantquantumcomputationconstant, GB25}. The above complexity, universality, and robustness of the dynamics seemingly rule out a simple theory of one-dimensional physics.
 
To reconcile this conceptual discrepancy between statics and dynamics, we must recognize that Nature's thermalization dynamics - coupling the system to a thermal bath - is noisy, by no means fine-tuned, and must respect the principles of thermal detailed balance. Nevertheless, under this thermal noise model, we know that macroscopic quantum phenomena can still exist at low enough temperatures due to the existence of self-correcting quantum memories (e.g., the 4D toric code \cite{DKLP01}). Even though thermal self-correction in one dimension would be rather surprising, there have been surprisingly no universal theorems on the rate of thermalization, beyond the intuition suggested by static properties. 

In classical statistical physics, a long-standing paradigm is that the rate of thermalization -- in terms of the mixing time of Glauber dynamics -- should be intimately tied to the structure of static correlations in thermal equilibrium \cite{Martinelli1999}. Static correlations often provide the essential starting point for proving mixing times, and a slowdown of Glauber dynamics also hints at thermodynamic phase transitions. This interplay between statics and dynamics has served as a fruitful guiding principle in exceedingly complex many-body systems, and has remained productive even in recent advances in Markov chain mixing times (e.g., \cite{anari2021spectral}).

For quantum systems, this paradigm has been extended, but restricted to commuting Hamiltonians \cite{kastoryano2016commuting}. The corresponding continuous-time quantum Markov chain that models the system-bath interaction is often referred to as a master equation or a Lindbladian, most notably the Davies generator \cite{davies74}. The main conceptual challenge is that one fundamentally cannot \textit{condition} on entangled quantum states like one normally does in classical distributions \cite{fannes1995boundaryconditionsquantumlattice}, thus requiring a substantial reformulation. Even though the search for an ultimate, physically motivated static property remains active, several notions of correlation decay have been defined (e.g. \cite{capel2021modified,capel2024quasi, kochanowski2025rapid}) and shown to imply fast mixing in a recursive manner akin to the classical spirits.

For noncommuting Hamiltonians, frameworks for proving mixing time bounds remain in their infancy. Only recently have workable quantum analogs of Glauber dynamics, \textit{quantum Gibbs samplers}, been proposed~\cite{temme2011quantum,yung2012quantum, Shtanko2021AlgorithmsforGibbs, chen2021fast, rall2023thermal, wocjan2023szegedy,chen2023quantum,jiang2024quantum,ding2024single}, particularly Lindbladian dynamics that satisfy KMS-detailed-balance and respect the locality of the Hamiltonian~\cite{chen2023efficient, ding2024efficient, SA24}. Regardless of whether one treats these Lindbladians as toy models of thermalization or as a thermal simulation algorithm, their explicit formalism allows one to pose sharp mathematical problems that can be precisely compared with the predictions of statistical mechanics. We ask:
\begin{align}
    \textit{Do all one-dimensional quantum systems thermalize?}
\end{align}
In this work, we show the affirmative by proving a \textit{system-size} independent spectral gap for the \cite{chen2023efficient} family of Lindbladians, for any one-dimensional Hamiltonian with short-range interactions, at any finite temperature. Together with existing static properties of 1D Gibbs states, our work extends the interplay between dynamics and statics to non-commuting Hamiltonians. In fact, a crucial input to our proof is the sub-exponential clustering of correlations of 1D systems \cite{Kimura_2025}; \textit{posteriori}, the constant spectral gap unconditionally strengthens~\cite{Kimura_2025} to a genuine \textit{exponential} decay, generalizing \cite{araki1969gibbs} to finite, non-translationally invariant systems. We believe our resolution of 1D spectral gaps serves as a stepping stone towards sharp mixing times of other non-commutative Hamiltonians, in non-perturbative regimes.

\subsection{Main Results}

We consider the following family of one-dimensional nearest-neighbor Hamiltonians defined on $n$ qudits of local dimension $2^q$, with open boundary conditions (\autoref{fig:1Dfamily}):

\begin{align}
    \vH =\sum_{b=1}^{n-1} \vH_{b,b+1}\quad \text{where}\quad \norm{\vH_{b,b+1}} \le 1. \label{eq:def1D}
\end{align} 
All finite-range 1D Hamiltonians can be embedded into the above by choosing a system-sized independent constant $q$ and rescaling the strength. In particular, the terms need not commute, nor be translationally invariant.

Our main theme features the interplay between thermal equilibrium properties and the thermalization dynamics in one dimension. Central to this discussion is the quantum \emph{Gibbs state}, defined by the Hamiltonian $\vH$ and an inverse-temperature $\beta\geq 0$ as:  
\begin{equation}
    \vrho := \frac{e^{-\beta \vH}}{\tr[e^{-\beta \vH}]} ,
\end{equation}
which, according to statistical mechanics, models the quantum system in true thermal equilibrium.

\begin{figure}[t]
\includegraphics[width=0.6\textwidth]{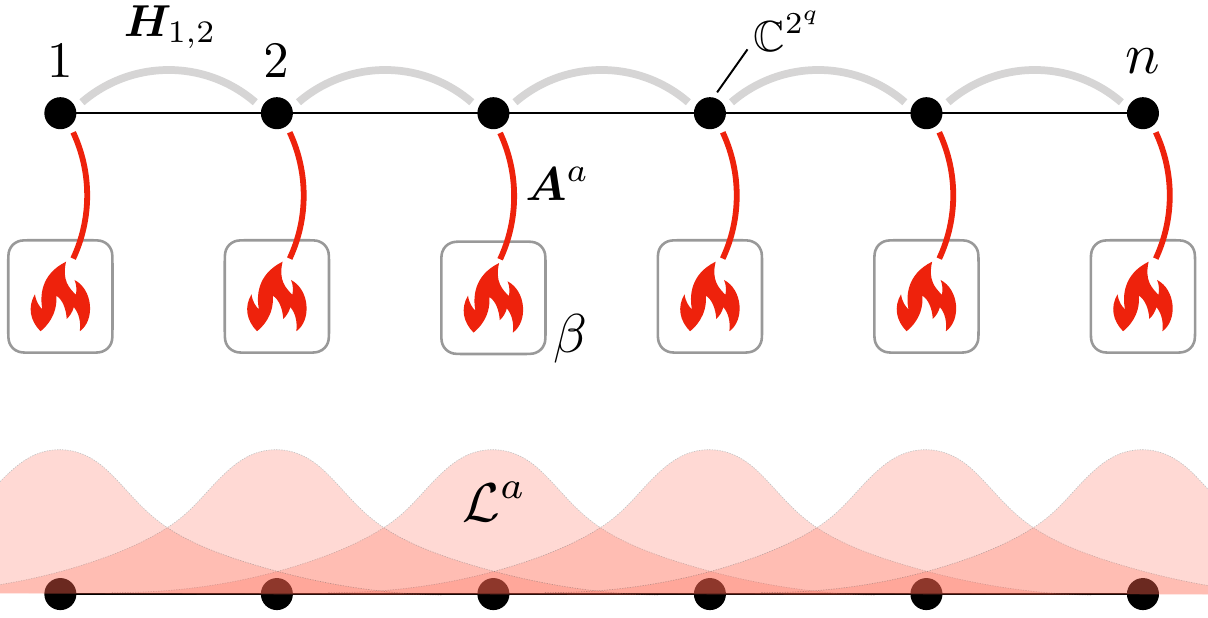}
\caption{
We consider a 1D Hamiltonian with nearest neighbour interaction; by taking large enough $q$ and rescaling, this captures all finite-range 1D models. We are interested in its rate of thermalization, when weakly coupled to a constant temperature bath on every site. Mathematically, we model the thermalization dynamics by an exactly detailed-balanced Lindbladian~\cite{chen2023efficient}, where each Lindbladian term corresponding to each system-bath coupling $\vA^a$ is quasi-local.
}\label{fig:1Dfamily}
\end{figure}

Our goal is to study the rate at which the spin chain approaches equilibrium when weakly coupled to a heat bath at every site (\autoref{fig:1Dfamily}). Formally, we model the dynamics by the family of quasi-local Lindbladians of \cite{chen2023efficient} (see Section~\ref{sec:Lindbladian}), where the ``jump'' operators are the set of all single-site Pauli-like operators $\CS^1_{[n]}$ 
\begin{align}
    \CL = \sum_{a\in \CS^1_{[n]}} \CL_a \quad \text{such that}\quad \CL_a[\vrho] = 0,\label{eq:maintext_L}
\end{align}
and each $\norm{\CL_a}_{1-1}$ is independent of the system size.
Since the \cite{chen2023efficient} family satisfies Kubo-Martin-Schwinger (KMS) $\vrho_{\beta}-$detailed-balance exactly, the rate of convergence can be captured by an estimate of its spectral gap, which is the central theme of our work. Specifically, we are interested in the thermodynamic scaling where the system size $n$ grows independently of the inverse temperature $\beta$ and local dimension $2^q$.

\begin{thm}[System-size independent spectral gap]
\label{thm:main}
The Lindbladian~\eqref{eq:maintext_L} defined by a one-dimensional spin chain Hamiltonian as in \eqref{eq:def1D} has a system-size independent spectral gap, which depends only on the inverse temperature $\beta$ and the local dimension $2^q$.
\end{thm}

Since the Lindbladian~\eqref{eq:maintext_L} acts extensively, a system-size independent gap is the best possible.
So far, an unconditional system-size independent gap is only known in commuting or perturbative regimes. Most notably, the Davies generator of 1D \textit{commuting} Hamiltonians \cite{kastoryano2016commuting}, and the Lindbladian family \eqref{eq:maintext_L} of non-commuting, geometrically-local Hamiltonians \textit{at sufficiently high temperatures}~\cite{rouze2024efficient} -- see Section~\ref{section:prior} for related work. Even though our paper focuses on the~\cite{chen2023efficient} family, we believe the argument extends to other KMS-detailed-balance, quasi-local Lindbladians~\cite{ding2024efficient, SA24} with reasonable filter functions (see~\autoref{rmk:other_L}).

As an immediate corollary, our lower bound on the spectral gap $\lambda$ of $\CL$ entails a bound on its \textit{mixing time}: the evolution of $\CL$ converges to error $\epsilon$ in trace distance to the Gibbs state in time $t_{\mathsf{mix}}(\epsilon) = \frac{1}{\lambda} O\big(n + \log\frac{1}{\epsilon}\big)$, regardless of the initial state. Thus, ~\autoref{thm:main} can be regarded as an upper bound on the rate of thermalization, or as an efficiency guarantee for a dissipative state preparation algorithm.\\ 

\noindent \textit{Low-depth state preparation.} A caveat of the algorithmic aspect of spectral gaps is that, a priori, the size of the quantum circuit which implements the desired Lindbladian evolution can be superlinear in the system size (like $n^2$). Fortunately, a nearly optimal state-preparation cost is available by employing an adiabatic algorithm for the \textit{purified Gibbs state} (or, the thermal field double state), defined on two copies of the system
\begin{align}
    \ket{\sqrt{\vrho}} := \frac{1}{\sqrt{\tr[e^{-\beta\vH}]}} \sum_{i} e^{-\beta E_i/2} \ket{\psi_i}\otimes \ket{\psi_i^*},
\end{align}
where we denote the eigendecomposition by $\vH = \sum_i E_i\ketbra{\psi_i}{\psi_i},$ and $\ket{\psi_i^*}$ denotes the entrywise conjugation.
\begin{cor}[$\mathsf{polylog}$ depth adiabatic algorithms]\label{cor:adiabatic}
The purified Gibbs state of any 1D Hamiltonian~\eqref{eq:def1D} can be prepared up to error $\epsilon$ in trace distance, in circuit depth $d_{\beta, q}\cdot \mathsf{polylog}(n/\epsilon)$ and circuit size $s_{\beta, q}\cdot n\cdot\mathsf{polylog}(n/\epsilon)$. In particular, the circuit is comprised of 2-qubit nearest-neighbor gates with 1D connectivity (with ancillas). 
\end{cor}
Up to polylogarithmic factors, the gate complexity and depth are both optimal in the system-size and error dependence, providing a conceptually clean correlation structure for the purified Gibbs states. Previously, the best provable quantum algorithms for 1D Gibbs states required superlinear depth, either using an MPO description \cite{Kliesch19, Molnar_2015, KAA21}, or a ``patching'' argument~\cite{brandao2019finite_prepare, Swingle_2016, Kato2019,kuwahara2024clustering}, see \autoref{section:discussion}. 

The proof largely imports the adiabatic simulation framework for gapped ground states \cite{has04}. Due to detailed balance, the Lindbladian \eqref{eq:maintext_L} under similarity transformation and vectorization defines a Hermitian, frustration-free ``parent Hamiltonian'' whose zero-eigenstate is the purified Gibbs state, and whose spectrum coincides with that of our Lindbladian $\CL$. Lowering the temperature defines a natural adiabatic path amenable to off-the-shelf quasi-adiabatic evolution for gapped ground state families~\cite{Bachmann_2011, chen2023efficient}. In contrast, optimal classical Gibbs sampling algorithms often go through a log-Sobolev inequality, but the quasi-adiabatic algorithm only requires a spectral gap throughout the temperature range. Of course, the preparation circuit may have nothing to say about the convergence rate of the physical thermalization dynamics, where log-Sobolev inequalities remain open. A self-contained, rigorous derivation of the gate complexity is presented\footnote{In and only in this appendix, we consider $D$-dimensional Hamiltonians with exponentially decaying interactions.} in Section~\ref{sec:adiabatic}.\\

\noindent \textit{Back to decay-of-correlations.} In spirit, our argument is an extension of the interplay between statics and dynamics to non-commutative settings. As mentioned, our proof of the spectral gap in \autoref{thm:main} starts from a black-box statement on the sub-exponential decay of correlations in arbitrary 1D Hamiltonians, at arbitrary temperatures~\cite{Kimura_2025}. Remarkably, as a corollary of the derived constant spectral gap for $\CL$, we may return to our starting point, to \textit{unconditionally} strengthen the imported clustering statement. Indeed, the constant spectral gap implies that the parent Hamiltonian has a unique, gapped ground state (the purified Gibbs state), and invoking~\cite{hastings2006spectral} then implies the exponential decay of correlations.

\begin{cor}[Exponential clustering of correlations] \label{cor:clustering} The Gibbs state of any 1D Hamiltonian  \eqref{eq:def1D} admits exponential decay of correlations, in that there exist constants $\mu, \alpha$ depending only on $\beta,q$, such that for any pair of operators $\vX$,$\vY$ supported on intervals $\sA, \sC\subseteq [n]$:
\begin{align}
         \bigg|\tr[\vX^\dagger\vY\vrho] - \tr[\vrho \vX^\dagger] \cdot \tr[\vrho \vY]\bigg| \leq \alpha\cdot\|\vX\|\cdot  \|\vY\| \cdot   e^{-\mu\cdot \mathsf{dist}(\sA, \sC)}.
\end{align}
\end{cor}
This exponential decay form is optimal, already in classical spin chains~\cite{Ruelle1999StatisticalMR}. \autoref{cor:clustering} generalizes the seminal work of \cite{araki1969gibbs} for infinite, translationally-invariant systems to the finite, non-translationally invariant setting. Curiously, the proof attaining the true exponential decay is, fundamentally, dynamical.

\subsection{Main Challenges}
\label{section:prior}

A long-standing paradigm in \textit{classical} Gibbs sampling is that the \textit{static} correlations in the Gibbs measure, should be intimately tied to the \textit{dynamic} mixing time of local, detailed-balanced Markov chains. To control the spectral gap of the Markov chain, one often starts with a careful definition of decay-of-correlations and seeks a local-to-global recursion, where the spectral gap of the dynamics on a region of spins is related to that of smaller subregions in the chain (see, e.g., \cite{Martinelli1999}).

Of course, even classically, not all notions of correlation decay are equivalent. For instance, the decay of two-point correlation functions (weak-clustering) is far too loose to achieve fast-mixing in general. Intuitively, weak clustering is a property of the equilibrium measure that may not necessarily capture the non-equilibrium states encountered throughout the mixing process. Instead, the notion that implies rapid mixing (and is equivalent to) is \textit{strong spatial mixing} (SSM). SSM is meant to capture how the marginals of regions of spins are robust to changes to their boundary conditions. Remarkably, the only general exception to this distinction is in one dimension: two-point correlation decay already implies SSM, so the two notions coincide. \\

Adapting this paradigm to \textit{quantum} Gibbs samplers faces two fundamental conceptual challenges:
\begin{enumerate}
    \item One cannot decompose a quantum state by ``conditioning on'', or ``pinning'' boundary conditions.
    \item detailed-balanced Lindbladian dynamics for non-commuting Hamiltonians are generally non-local.
\end{enumerate}

So far, the existing body of static-to-dynamic arguments has been largely focused on commuting Hamiltonians (e.g.,~\cite{kastoryano2016commuting,capel2021modified, kochanowski2025rapid}), which mainly confronts the first point. Already, one must admit that liberal conditioning, paramount to the very definition of decay-of-correlations, is not available and requires a substantial change of language.\footnote{This connects closely to the breakdown of the Dobrushin-Lanford-Ruelle (DLR) theory \cite{Dobrushin1968TheDO, Lanford1969ObservablesAI} of boundary conditions in the quantum setting \cite{fannes1995boundaryconditionsquantumlattice}.} Instead, one relies on the following non-commutative version of a \textit{conditional expectation}:

\begin{defn}[Conditional Expectation]\label{defn:CE_intro}
We say a superoperator $\BE$ (a map from operators to operators) is a {conditional expectation} with respect to the Gibbs state $\vrho$ if it satisfies:
    \begin{itemize}
        \item \emph{KMS-detailed-balance:} for any pair of operators $\vX,\vY,$ $\braket{\BE[\vY],\vX }_{\vrho}=\braket{\vY,\BE[\vX]}_{\vrho}.$
        \item \emph{Unitality:} $\BE[\vI]=\vI.$ Consequently, for any operator $\vX$, $\tr[ \BE[\vX]\vrho] = \tr[\vX\vrho].$ 
        \item \emph{Monotonicity:} The spectrum of $\BE$ is real and contained in $[0,1].$ 
        \item \emph{Complete Positivity}.
    \end{itemize}
\end{defn}

The simplest and most natural concrete example is the time evolution of a (local) Lindbladian. 
Intuitively, the following \autoref{eg:Levol} is a Heisenberg-picture interpretation of ``resampling'' the region of qubits according to its current boundary, as the quantum counterpart to running Glauber dynamics restricted to a local set of spins.
\begin{eg}\label{eg:Levol}
    For any region of qubits $\sR\subseteq [n]$, consider a sum over (individually detailed-balanced) Lindbladian terms $\CL_\sR=\sum_{a\in \sR}\CL_a$. Then, the infinite-time evolution of the generator $\BE_\sR = \lim_{t\rightarrow \infty} e^{t\CL_\sR^\dagger}$ is a conditional expectation map.
\end{eg}

In this language, \cite{kastoryano2016commuting} introduced the \textit{strong clustering} condition, a quantum version of decay-of-correlations which explicitly depends on the choice of conditional expectations (such as that resulting from Davies generator \cite{davies74}). Informally, their definition is essentially the following ``gluing identity''. A family of (projective) conditional expectation maps $\{\BE_\sR\}_{\sR\subseteq [n]}$ satisfies strong-clustering if, for disjoint regions $\sA, \sB, \sC\subseteq [n]$, 
\begin{equation}
    \BE_{\sA\sB}\BE_{\sB\sC}\approx \BE_{\sA\sB\sC}, \label{eq:sc-intro}
\end{equation}
in a suitable weighted norm that respects the Gibbs measure. Colloquially, ``resampling'' the region $\sB\sC$, and then subsequently $\sA\sB$, is always approximately equivalent to directly resampling $\sA\sB\sC$ (regardless of their current boundary configurations). Strong clustering is meant to capture the ``insensitivity to boundary conditions'' present in the classical definition of SSM, in that the expressions in \eqref{eq:sc-intro} can only differ if the configuration at the boundary of the region $\mathsf{C}$ still influences $\sA$. 

\cite{kastoryano2016commuting} shows that if \eqref{eq:sc-intro} holds with sufficiently fast decay in the overlap region $\sB$, then one is able to relate the conditional variances in the two sub-regions, ultimately resulting in a recursion over the \textit{local gaps} of said regions.\footnote{The conditional variance of an operator $\vX$ on a region $\sR$ is $\mathsf{Var}_\sR(\vX)=\|\vX-\BE_\sR[\vX]\|_{\vrho}^2$. The \textit{local gap} of the restricted generator $\CL_\sR$ quantifies the rate of decay of this conditional variance. Its classical analog corresponds to the rate of convergence of, e.g., Glauber dynamics, when restricted to a patch of spins, over a configuration of the neighboring spins. } Their approach to commuting quantum systems serves as the main inspiration for our work. Still, with the exception of 1D commuting systems,  the ultimate form of correlation decay remains an active area of debate, in part due to the difficulty in relating strong clustering to more physical quantities or explicitly showing strong clustering in specific systems.\\

For \textit{non-commuting} Hamiltonians, frameworks for studying mixing times are almost nonexistent, and proofs of fast mixing have been restricted to the perturbative, high-temperature regimes, without explicit reference to the underlying static correlation \cite{rouze2024efficient, rouze2024optimal, tong2025fastmixingweaklyinteracting,vsmid2025polynomial,vsmid2025rapid}. 
To extend the commuting arguments~\cite{kastoryano2016commuting} to our non-commuting setting, one must confront the second point above head-on and address the quasi-locality of the dynamics \cite{chen2023quantum, chen2023efficient, ding2024efficient, SA24}. Here, we highlight two technical challenges:\\

\noindent \textit{Quasi-locality of the conditional expectation, and its local spectral gap}. The natural candidate for a family of conditional expectations is the time-evolution of the generator $\CL$. However, in contrast to the classical and commuting settings where a local update effectively acts only on a finite-dimensional space, the evolution of a single ``local'' generator $\CL_a$ \eqref{eq:maintext_L} can act on the whole system. On the other hand, truncating the tails of the dynamics can break the underlying detailed balance structure.
As a consequence, we do not even know how to establish the base case for the recursion -- i.e., a local spectral gap independent of the system size.\\

\noindent \textit{Proving strong clustering in 1D systems.} Even if we were to assume that the evolution of $\CL$ is (quasi-)local, a proof of strong clustering in 1D is by no means immediate. Indeed, it likely requires a detailed understanding of the specific conditional expectation, in conjunction with extensive use of the locality properties of 1D systems (see \autoref{section:1D-locality} for further background). Even though the Gibbs states of one-dimensional non-commuting Hamiltonians have been studied since Araki~\cite{araki1969gibbs}, proofs of correlation decay, even for two-point functions, are surprisingly recent~\cite{Kimura_2025}.

\subsection{Proof Ideas}
\label{sec:overview}

\begin{figure}[b]
\centering
\begin{tikzpicture}[
  box/.style={rounded corners, minimum width=1.5cm, minimum height=1.3cm, text width=4.2cm, align=center, font=\small, draw=black, fill=gray!2},
  arrow/.style={-{Latex}, thick},
  node distance=1cm and 1.6cm,
  box2/.style={rounded corners, minimum width=1.1cm, minimum height=1.3cm, text width=3.8cm, align=center, font=\small, draw=black, fill=gray!2},
  box3/.style={rounded corners, minimum width=1.3cm, minimum height=1.6cm, text width=3.5cm, align=center, font=\small, draw=black, fill=gray!2},
  box4/.style={rounded corners, minimum width=1.0cm, minimum height=1.0cm, text width=2.6cm, align=center, font=\small, draw=black, fill=gray!2},
    boundingbox/.style={draw=black, thick, rounded corners=5pt, inner sep=10pt, fit=#1}
]

\node[box] (T1) {
\cite{Kimura_2025} \textbf{Sub-Exponential} \\$\infty-$\textbf{Weak Clustering} 
};

\node[box, right=of T1] (T2) {
 \textbf{Sub-Exponential} \\\textbf{(KMS) Weak Clustering} 
};

\node[box2, above=of T1] (T3) {
  \textbf{$\CK$ is Locally Gapped} \\ [3pt]
  (Lemma \ref{lem:local_gap_K})
};

\node[box2, above=of T2] (T4) {
  \textbf{Quasi-Locality of $\tBE$} \\[3pt]
  $\tBE_\sA\tBE_\sC\approx \tBE_{\sA\sC}$
};

\node[box3](T5) at ($ (T2)!0.5!(T4) + (6.0cm,0) $) {
  \textbf{Strong Clustering}\\[3pt]
  $\tBE_{\sA\sB}\tBE_{\sB\sC}\approx \tBE_{\sA\sB\sC}$ \\ [3pt]
  (for suitable $\sA,\sB,\sC$)
};

\node[box4] (T6) at ($(T5.south) + (0,-2.2cm)$) {
  \textbf{$\CK$ is Gapped}
};

\node[box4] (T7) at ($(T6) + (-6.0cm,0)$) {
  \textbf{$\CL$ is Gapped}
};

\node[box2] (T8) at ($(T7) + (-6.0cm, 0)$) {
  \textbf{Exponential \\$\infty$-Weak Clustering} \\ [3pt]
};

\coordinate (mergepoint) at ($ (T2)!0.5!(T4) + (2.7cm,0) $);

\draw[arrow] (mergepoint) -- (T5.west);

\draw[arrow] (T2.east) -- ($(T2.east)+(0.4cm,0)$) -- (mergepoint);
\draw[arrow] (T4.east) -- ($(T4.east)+(0.6cm,0)$) -- (mergepoint) node[pos=0.65, above right=22pt and -15pt] {\footnotesize Thm \ref{thm:strong_from_weak}};

\node[draw,circle,inner sep=1.5pt,thick,fill=white] at (mergepoint) {\small\textbf{+}};

\node[boundingbox=(T1)(T2)(T3)(T4)(T5)(T6)(T7)(T8)] {};

\draw[arrow] (T3) -- (T4) node[midway, above] {\footnotesize Lem~\ref{lem:localized_E}};
\draw[arrow] (T1) -- (T2) node[midway, above] {\footnotesize Thm~\ref{thm:kms_from_weak}};
\draw[arrow] (T5.south) -- (T6.north) node[midway, right] {\footnotesize Thm~\ref{thm:K-is-gapped}};
\draw[arrow] (T6) -- (T7) node[midway, above] {\footnotesize Thm~\ref{thm:ckg_gap}};
\draw[arrow] (T7) -- (T8) node[midway, above] {\footnotesize Cor~\ref{cor:clustering}};

\end{tikzpicture}

\caption{ \label{fig:1D-argument-summary}
An outline of the proof of \autoref{thm:main}, starting from the sub-exponential decay-of-correlations in the Gibbs state of arbitrary 1D Hamiltonians, at all finite temperatures, established by \cite{Kimura_2025}. }
\end{figure}
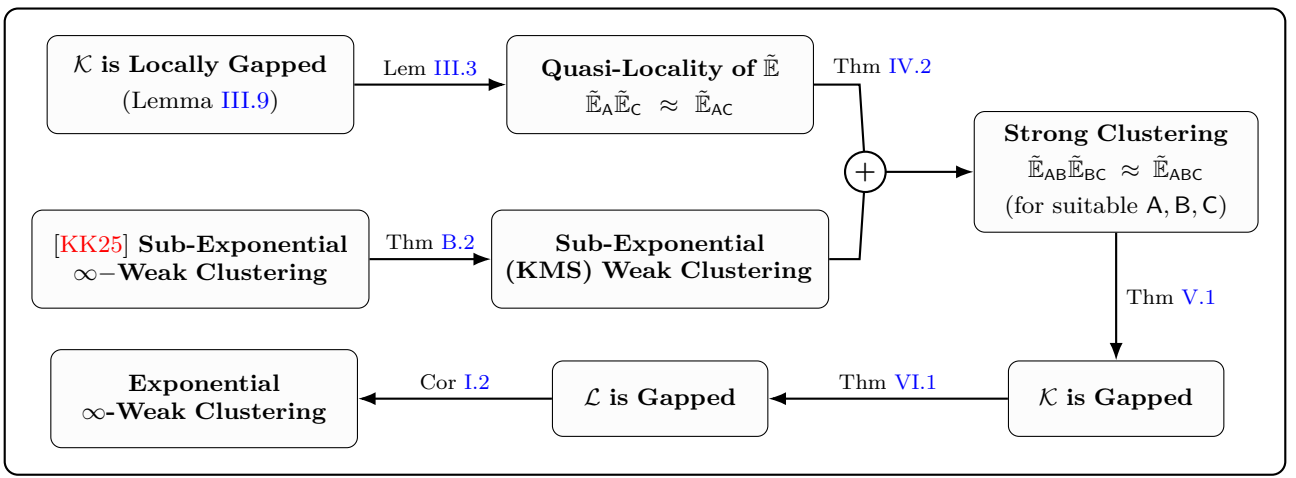

We reconcile the challenges by constructing a relaxation of conditional expectation maps $\tBE$, that simultaneously
\begin{enumerate}
    \item enjoys sufficient locality and admits a local spectral gap, and 
    \item defines a version of strong clustering that can be derived from a more physical, off-the-shelf version of correlation decay from the literature \cite{Kimura_2025}.
\end{enumerate}

Put together, these maps fit into the recursive scheme of \cite{kastoryano2016commuting}, and will be used to derive the spectral gap. However, $\tBE$ is not exactly generated by the Lindbladian $\CL$ as in \autoref{eg:Levol}, but instead is the time-evolution of an auxiliary dynamics $\CK$. At the end of the day, we perform a spectral gap comparison to relate the two. A summary of the argument is presented in \autoref{fig:1D-argument-summary}, and we expand on the key ingredients as follows.\\

\noindent \textbf{A Quasi-Local Spectral Conditional Expectation (\autoref{sec:CE}).} For any subset $\sA \subseteq [n],$ we define the map $\tBE_\sA$ to be the infinite-time limit of an auxiliary generator $\CK$ in the Heisenberg picture:
\begin{align}
   \tBE_{\sA}[\cdot ]:= \lim_{t \rightarrow \infty}e^{t\CK_\sA}[\cdot] \quad \text{where}\quad \CK_\sA =\sum_{a\in\sA} \CK_a.  \label{eq:intro_ER}
\end{align}
For each jump operator $\vA^a$ (concretely, a single-site Pauli operator), the local generator is defined by:
    \begin{align}
    \CK_a[\vX] &:= [\vA^a,\vX] \big(\vrho^{1/2} \vA^{a\dagger}\vrho^{-1/2}\big)-  \big(\vrho^{-1/2} \vA^{a\dagger}\vrho^{1/2}\big)[\vA^a,\vX]. 
\label{eq:intro-CK}
\end{align}
 In many ways, this new generator $\CK_a$ resembles the Lindbladian $\CL^{\dagger}_a$\footnote{The Lindbladian $\CL$ is defined in the Schrodinger picture. We define $\CK$ in the Heisenberg picture because we will mostly study its action on operators.}: it is (KMS) detailed-balanced, fixes the identity, and has nonpositive spectra. However, it is \textit{not} completely positive. Consequently, the infinite-time limit $\tBE_{\sA}$ is not a physical operation, and we refer to it as a \textit{spectral conditional expectation} instead (\autoref{defn:spectral_expectation}). 
 
The main virtue of introducing $\CK$ (instead of using the original Lindbladian $\CL^{\dagger}$) is its highly explicit \textit{locality} properties, in both time and space:

\begin{enumerate}
    \item \textit{The kernel of $\CK_\sA$} admits an explicit characterization: $\CK_a[\vX] = 0\iff [\vA^a, \vX] = 0$. Intuitively, $\CK_{\sA}$ wants to trivialize operators on $\sA$, and stops sharply whenever an operator becomes trivial on $\sA$. Consequently, the limiting map $\tBE_{\sA}$ is an orthogonal projection onto the linear space spanned by operators trivial on $\sA.$\footnote{That is, the image of $\tBE_{\sA}$ is explicit $\tBE_{\sA}[\vX] = \vX\iff \vX = \vI_{\sA} \otimes \vX'_{\sA^c}$.} 
    \item[2.] \textit{A local gap}. In the $\beta=0$ limit, $\CK$ recovers the generator of the depolarizing semi-group. At any finite $\beta$, the spectral gap of $\CK_\sA$ relative to its kernel can be shown to be at least inverse-exponential in $|\sA|$ (but independent of $n$), by a comparison to the $\beta=0$ limit.
    \item[3.] \textit{Quasi-locality}. In one dimension, at every finite $\beta$, the \textit{complex-time evolution} $\vrho \vA \vrho^{-1}$ of local operators is quasi-local\footnote{In higher dimensions, $\CK$ as literally written can have diverging strength.} (see Appendix~\ref{section:1D-locality}). Therefore, the generator $\CK_\sA$ can be truncated to a radius $\ell$ around $\sA$ up to an error exponentially small in $\ell\log\ell$. 
\end{enumerate}

In contrast, it is unclear which (if any) of these properties are shared by $\CL^{\dagger}$ \cite{chen2023efficient}.  Put together, evolving $\CK$ must yield a map $\tBE$ that inherits this quasi-locality: $\tBE_\sA$ is well-approximated by that with a truncated Hamiltonian (defined through \eqref{eq:intro-CK}) to a region $\sR$ around $\sA$ (\autoref{lem:localized_E}): 
    \begin{align}
    \|\tBE_\sA-\tBE_\sA^{(\sR)}\|_{\vrho} \sim \exp\big(c_1\cdot |\sA| - c_2\cdot \ell\log \ell\big),\quad  \text{where}\quad \ell = \mathsf{dist}(\sA, [n]\setminus \sR) \label{eq:main_EAERA}.
\end{align}

\begin{figure}[t]
\includegraphics[width=0.8\textwidth]{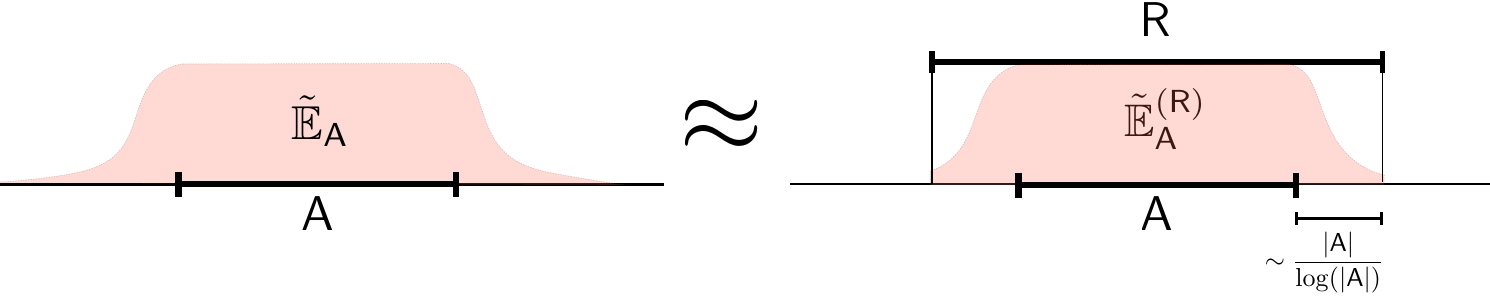}
\caption{
The quasi-locality of the spectral conditional expectation map. For any set $\sA$, the spectral conditional expectation $\tBE_\sA$ is well-approximated by that with a truncated Hamiltonian nearby $\sA$. Quantitatively, the length of the ``buffer zone'' needs to be at least $\frac{\labs{\mathsf{\sA}}}{\log\labs{\mathsf{\sA}}}.$
}\label{fig:quasilocal_E}

\end{figure}

Quantitatively, the buffer region needs to contain at least $\sim |\sA|/\log|\sA|$ neighboring sites (see \autoref{fig:quasilocal_E}) for the truncation error to be non-trivial. As a corollary, the conditional expectations over pairs of intervals factorize $\tBE_\sA\tBE_\sC\approx \tBE_{\sA\sC},$ so long as $\sA$ and $\sC$ are sufficiently far from each other, relative to $|\sA\sC|/\log |\sA\sC|$. These fundamental quasi-locality properties serve as a recurring tool throughout the entire argument.\\

\noindent \textbf{Strong Clustering from Weak Clustering (\autoref{section:clustering}).} Equipped with the family of spectral conditional expectations $\{\tilde{\BE}_{\sR}\}_{\sR\subseteq [n]}$, we can now revisit the definition of strong clustering. For consecutive, disjoint intervals $\sA, \sB, \sC\subseteq [n]$ (see~\autoref{fig:EAB_BC_ABC}), 
\begin{equation}
    \tBE_{\sA\sB}\tBE_{\sB\sC}\approx \tBE_{\sA\sB\sC} \label{eq:gluing_intro_TBE}.
\end{equation}

Since the intersection of the images $\tBE_{\sA\sB}$, $\tBE_{\sB\sC}$ is exactly $\tBE_{\sA\sB\sC}$, we can already deduce that $\lim_{m\rightarrow \infty}(\tBE_{\sA\sB}\tBE_{\sB\sC})^{m} = \tBE_{\sA\sB\sC}.$ However, strong clustering requires that the convergence to $\tBE_{\sA\sB\sC}$ take place in just one step. 
Our proof (\autoref{thm:strong_from_weak}) of strong clustering \eqref{eq:gluing_intro_TBE} for 1D non-commuting Hamiltonians is a reduction to the decay of two-point correlation functions in the KMS inner product -- often referred to as (KMS) \textit{weak clustering}. We then show that (KMS) weak clustering holds unconditionally in 1D Hamiltonians, by a reduction to the recent results of \cite{Kimura_2025}, see \autoref{thm:kms_from_weak}. Although this final step is somewhat laborious, it mostly relies on established techniques, and so we relegate the details to Appendix \ref{app:Weakclustering}.\\

\noindent \textit{The key step in the proof of \eqref{eq:gluing_intro_TBE}.} The proof of \eqref{eq:gluing_intro_TBE} is one of our main technical contributions, and we attempt to sketch a clean conceptual component, ``trivializing an island.'' Consider an operator $\vX_\sB$ supported on an interval $\sB\subseteq [n]$, and centered w.r.t. the Gibbs measure $\tr[\vrho\vX_\sB]=0$. We claim that, running the conditional expectation map $\tBE_{\sA\sB\sC}[\vX_\sB]$ on a region $\sA\sB\sC$ surrounding $\sB$, approximately trivializes the operator ``island''
\begin{align}
\tBE_{\sA\sB\sC}[\vX_\sB] \approx 0\label{eq:kills_island}, 
\end{align}
for large enough $\sA,\sC.$ In general, the map $\tBE_\sR$ always trivializes operators on $\sR$ but may produce new components supported outside $\sR$. 
However, the claim in~\eqref{eq:kills_island} is that, since $\vX_\sB$ is supported ``quite far'' from the boundary of $\sA\sB\sC$, the resulting operator must simply vanish.
In fact, this can be related to the decay of two-point correlation functions (i.e., weak clustering):
\begin{align}
    \big\|\tBE_{\sA\sB\sC}[\vX_\sB]\big\|^2_{\vrho}  &= \braket{ \under{\vX_{\sB}}{\text{on }\sB}, \quad  \under{\tBE_{\sA\sB\sC}[\vX_\sB]}{\text{on }[n]\setminus \sA, \sB, \sC}}_{\vrho}, 
\end{align}
using that $\tBE_{\sA\sB\sC}^2=\tBE_{\sA\sB\sC}.$ The RHS is exactly a (KMS) correlation function! 

To fully prove~\eqref{eq:gluing_intro_TBE}, we must also handle the case where $\vX$ is supported on $\sB$ and also outside of $\sA\sB\sC$. In order to do so, we fundamentally need to leverage the quasi-locality of $\tBE$. Unfortunately, this comes with a serious caveat: our version of strong clustering only holds over ``well-separated tripartitions'', where the size of $\sB$ is at least near-linear in the full interval length, $\sim |\sA\sB\sC|/\log |\sA\sB\sC|$.\\

\begin{figure}[t]
\includegraphics[width=0.8\textwidth]{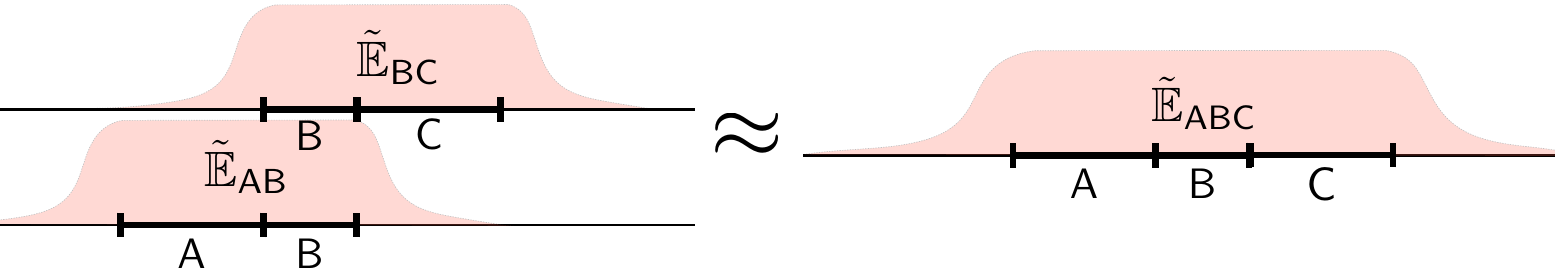}
\caption{
{The Definition of Strong Clustering. The notion of strong clustering for our spectral conditional expectations is best described by a ``gluing'' property, see equation \eqref{eq:gluing_intro_TBE}. Due to quasi-locality, we must take the size of the region $\sB$ to scale at least with $\labs{\sA\sC}/\log(\labs{\sA\sC})$.
}\label{fig:EAB_BC_ABC}
}
\end{figure}

\noindent \textbf{Spectral Gap Recursions for the Generator $\CK$ (\autoref{sec:conditionalgapK})}. Despite its limitations, our version of strong clustering still gives rise to a recursion over local gaps, following the \cite{kastoryano2016commuting} scheme. A proof overview is deferred to \autoref{sec:conditionalgapK}. However, due to the constraints on the size of the middle region $\sB$, the result is an inverse$-\mathsf{polylog}(|\sR|)$ local gap for the restricted generator $\CK_\sR$, over regions of size $|\sR| \geq \ell_0(\beta, q)$.

Fortunately, we are able to \textit{bootstrap} this improvement to the local gap for the generator $\CK$ all the way up to a constant, which is optimal. The claim is that, at larger length scales, the conditional expectation $\tBE_\sR$ becomes more and more localized, significantly improving our original a priori estimates~\eqref{eq:main_EAERA}. That is, the improvement to the local gap, at a larger length-scale $\ell\geq \ell_0$, allows us to iterate on the quasi-locality properties and, in turn, on the strong-clustering guarantees. 

There are many routes to such a bootstrapping argument, and in \autoref{sec:clustering-from-gap} we implement a Lieb-Robinson approach for the time evolution of the generator $\CK$. Alternatively, one can likely redo the analyses of \autoref{sec:CE}+\autoref{section:clustering}. Either way, we show that above the length-scale $\ell_0$, strong clustering holds over every ``well-separated'' tripartition $\sA, \sB, \sC\subseteq [n]$ where $|\sB|\geq |\sA\sB\sC|^{1-\delta}$, for an appropriate constant $\delta>0$. This feeds into another round of spectral gap recursion, which finally results in a constant gap for $\CK$. \\

\noindent \textbf{Compare to the Spectral Gap of $\CL$ (\autoref{sec:CKGgap}).}  The last step in our proof of \autoref{thm:main} is to relate the hard-earned spectral gap of $\CK$ (which is not completely positive) to that of the full Lindbladian $\CL^\dagger$ \eqref{eq:maintext_L} (which is physical). Since both the generators are ergodic -- thus sharing the same fixed points -- it suffices to lower bound the \textit{Dirichlet form} of $\CL^\dagger$ (see \autoref{lem:Dirichlet}), by that of $\CK$: 
\begin{align}
  \sum_{a\in \CS_{[n]}^1}  \|[\vA^a, \vO]\|_{\vrho}^2 =     \braket{\vO, -\CK[\vO]}_{\vrho} \le c_{\beta,q}\cdot  \braket{\vO, -\CL^\dagger[\vO]}_{\vrho} ,
\end{align}
up to a suitable $\beta, q$ dependent $c_{\beta,q}$. Here, we compare the gap for the dynamics only at the global level, and it is unclear whether a comparison at earlier stages of the proof is feasible, due to the lack of local gaps and the explicit form for the kernel of $\CL$.

This comparison strategy can be traced back to tricks from \cite{chen2025quantumMarkov}, with the additional challenge that establishing a PSD order of the Dirichlet form forbids any other norm other than the KMS-inner product. Recall the Lindblad operators of $\CL$ \eqref{eq:maintext_L} consists of the \textit{operator Fourier transforms} $\hat{\vA}^a(\omega)$ of each Pauli jump operator $\vA^a$ (see~\autoref{section:oft}). We proceed by decomposing each such $\vA^a$ as a linear combination of the operator Fourier transforms $\hat{\vA}^a(\omega)$ (\autoref{lem:sumoverenergies}), and then introducing a cutoff frequency $\Omega\in \BR^+$. The ``low-frequency components'' ($|\omega|\leq \Omega$) can be directly related to the Dirichlet form of $\CL$; while ``the high-frequency components'' ($|\omega|\geq \Omega$) decay rapidly, and are small relative to the Dirichlet form of $\CK.$

\subsection{Discussion}
\label{section:discussion}
In this paper, we showed that quantum spin chains thermalize at every finite temperature by proving an optimal, system-size-independent spectral gap for the Lindbladian~\cite{chen2023efficient}. Our dynamical result is a concrete step towards a comprehensive picture of thermal physics in one dimension, extending the interplay between static (correlation decay) and dynamics (mixing times) to non-commutative settings. Consequently, we also obtain an optimal exponential decay of correlations, which generalizes the translational invariant case~\cite{araki1969gibbs}, and a polylog-depth circuit for adiabatically preparing the purified Gibbs state. Our proof is a non-commutative extension of the spectral gap recursion scheme of~\cite{kastoryano2016commuting}, and employs a non-physical conditional expectation to reconcile challenges rooted in the non-locality of the dynamics.\\

\noindent \textit{Decay-of-correlations in higher-dimensions.} Throughout the proof, however, we heavily exploited the mature literature of 1D results, most notably the convergence of complex-time evolution \cite{araki1969gibbs,  perez2023locality} and various forms of correlation decay~\cite{bluhm2022exponential, Kimura_2025}. In higher dimensions, complex-time evolution generically diverges below a finite temperature; thus, the auxiliary generator $\CK$ may have a divergent norm. What is more, even given another candidate conditional expectation map, proving strong clustering is likely nontrivial, since even classically strong clustering need not follow from more tractible notions of (weak) correlation decay. At least, our 1D results clarify that proving an optimal spectral gap in non-perturbative regimes is possible even with non-local dynamics.\\

\noindent \textit{Mixing times beyond spectral gaps.} In commuting Gibbs samplers, several works have studied more fine-grained approaches to mixing time bounds beyond spectral gaps, e.g., log-Sobolev inequalities \cite{Bardet_2021, capel2021modified, Bardet2024, BCL24}, optimal transport metrics \cite{capel2024quasi}, coupling arguments \cite{BGL25, páezvelasco2025efficientsimplegibbsstate, Hwang24}, etc. Most relevant to our work are the results of \cite{Bardet2024,kochanowski2024rapid}, who established that if the Davies generator of a 1D commuting Hamiltonian is gapped, then it also admits a (modified) log-Sobolev inequality, and thus is \textit{rapid-mixing.} An intriguing open question is to identify if the \cite{chen2023efficient} Lindbladian family similarly admits such an implication. Unfortunately, it remains unclear to what extent the techniques of \cite{Bardet2024,kochanowski2024rapid} are applicable in our context, since our spectral conditional expectation is not completely positive. \\

\noindent \textit{Low depth adiabatic algorithms.} Even without a log-Sobolev inequality, spectral gaps can still lead to near-optimal state preparation circuits through quasi-adiabatic evolution. To establish the gate complexity guarantee in \autoref{cor:adiabatic}, we write down a detailed analysis of an approach put forth by \cite{chen2023efficient}.

The general statement goes as follows: Let $\vH$ be a $d$-dimensional lattice Hamiltonian with exponentially decaying interactions. If the associated family of Lindbladians $\{\CL_{s}\}_{s\in [0, \beta]}$, admits a uniform spectral gap lower bound from infinite temperature $\beta =0$ to a fixed $\beta^{-1}$, then, the Gibbs state $\vrho_\beta$ admits a $\mathsf{poly}\log$-depth state-preparation circuit. Therefore, the absence of a dynamical phase transition in the spectral gap of the thermalization process implies $\ket{\sqrt{\vrho_\beta}}$ has very limited entanglement. This resembles the ``sudden death of entanglement'' \cite{bakshi2024high}, i.e., that Gibbs states are completely unentangled at sufficiently high temperatures, but here we allow for quasi-local entangling circuits, which hold all the way to the dynamical phase transition where the Lindbladian spectral gap closes.\\

\noindent \textit{Other Gibbs state preparation algorithms.} More broadly, while mixing time bounds are sufficient for Gibbs sampling, there are several other routes to Gibbs state preparation algorithms. Most notably, the patching argument~\cite {brandao2019finite_prepare, Swingle_2016, Kato2019,kuwahara2024clustering}, which is a quantum algorithm to prepare Gibbs states that admit certain decay-of-correlations. These algorithms express the Gibbs state as a constant-depth circuit of quasi-local ``gates'', which act on $\log n$ or $\mathsf{poly}\log  n$ qubits at a time. When implemented using 2-qubit gates, the resulting circuit depth is polynomial in 1D, and super-polynomial in higher dimensions.

Another route lies in a tensor network (MPO) description \cite{Kliesch19, Molnar_2015, KAA21} for the quantum Gibbs state. Notably, \cite{KAA21} devised a quasi-linear time classical algorithm to represent the Gibbs state of 1D non-commuting Hamiltonians at arbitrary finite temperatures, but it remains open to give a tight conversion to quantum circuits, especially in provably achieving a depth sublinear in the system-size.\footnote{The generic conversion scheme of~\cite{lin2021real} gives a quantum circuit depth at least linear in the system size.}\\

\textbf{Acknowledgements.}
We thank Anurag Anshu, Jo\~ao Basso, Fernando Brand\~ao, Angela Capel, Soonwon Choi, Jonas Haferkamp, Michael Kastoryano, Vedika Khemani, Tomotaka Kuwahara, Lin Lin, Yunchao Liu, Tony Metger, David Perez Garcia, Nikhil Srivastava, Matteo Scandi, Daniel Stilck Fran\c{c}a, Ewin Tang, Umesh Vazirani, and Yongtao Zhan for encouraging discussions. CFC thanks Cambyse Rouzé for hosting him at Institut Polytechnique de Paris in summer 2024, where the questions about 1D spectral gaps and a related work~\cite{chen2025quantumMarkov} were initiated. CFC is supported by a Simons-CIQC postdoctoral fellowship through NSF QLCI Grant No. 2016245. 

\textbf{Contribution statement.}
CF.C. conceived and developed the framework for a spectral recursion via the nonphysical generator $\CK$ and guided the collaboration with T.B. T.B. later joined the project towards obtaining a tight spectral gap by adapting the framework~\cite{kastoryano2016commuting}. All authors contributed to writing and revising the manuscript.

\section{Preliminaries}

\subsection{Notation}

We consider a one-dimensional spin chain of finite system size $n$, where the spins have constant local dimension $2^q$. Each such spin $i\in [n]:=\{1, 2, \cdots, n\}$ will be referred to as a \textit{site}. The associated spin chain Hilbert space is denoted as $\CH = \otimes_{i=1}^{n} (\BC^{2^q})$.  The space of linear operators on $\CH$ is referred to as $\CB(\CH)$; we use boldface letters to refer to operators. A superoperator $\CN:\CB(\CH)\rightarrow \CB(\CH)$ is a map between linear operators. 

We denote the set of single-site $q$-weight Pauli strings, supported on the $i$th $2^q$ dimensional qudit as:
\begin{align}
    \CS_i &:= ( \{\vI_2,\vX,\vY,\vZ\}^{\otimes q}\} )\otimes \vI_{2^q}^{n-1},
\end{align}
and we will never care how the Pauli strings are actually embedded within a site. Sans serif is used to refer to subsets of sites $\sA\subseteq [n]$. For any subset $\sA\subseteq [n]$, we define the set of single-site Pauli strings, and multi-site Pauli strings:
\begin{align}
    \CS^1_\sA &:= \cup_{i\in \sA} \CS_i,\\
    \CS_\sA &:= ( \{\vI_2,\vX,\vY,\vZ\}^{\otimes q\labs{\sA}})\otimes \vI_{2^q}^{n-\labs{\sA}}.
\end{align}

Throughout the paper, we focus on the scaling with the system size $n$; the local dimension $2^q$ and the inverse temperature $\beta$ will be arbitrary constants that are independent of the system size. For self-contained lemmas, we either display the explicit dependence on $\beta,q$, or use notation
\begin{align}
    b_{\beta,q}, a_{\beta} 
\end{align}
to omit numerical factors that may change from line to line but depend only on the subscripts.

\subsection{Hamiltonians on a One-Dimensional Lattice}

In this paper, we consider a finite one-dimensional spin chain, where the spins interact under a short-range, 2-local, nearest-neighbor Hamiltonian.

\begin{defn}[1D Hamiltonians]\label{defn:1D}
    Consider the spin chain Hilbert space $\CH=\otimes_{i=1}^{n} (\BC^{2^q})$ on $n$ sites, where each site $b \in [n]$ corresponds to a qudit of local dimension $2^q$. We consider 2-local Hamiltonians defined on a one-dimensional lattice of the form
\begin{align}
    \vH = \sum_{\gamma \in \Gamma} \vH_{\gamma}=\sum_{b=1}^{n-1} \vH_{b,b+1}\quad \text{where}\quad \norm{\vH_{b,b+1}} \le 1,
\end{align}
which acts on nearest-neighbour sites, with nonperiodic boundary conditions. 
\end{defn}

We refer to the Hamiltonian terms $\vH_{b, b+1}$ between sites as \textit{links}. For any two subsets of sites $\sA, \sB\subseteq [n]$, we denote by $\operatorname{dist}(\sA,\sB)$ the minimal number of links on any path between any site of $\sA$ to any site of $\sB$. In particular, for two sites $i, j\in [n]$, $\operatorname{dist}(i,j) = |i-j|$; if $\sA, \sB$ are disjoint and adjacent on the lattice, their distance is 1. Further, for any link $\vH_\gamma$, we refer to $\operatorname{dist}(\gamma,\sA)$ as the distance between the support of $\vH_\gamma$ to $\sA\subseteq [n]$. 

An \textit{interval} $[a, b]:=\{a, a+1, \cdots, b\}\subseteq [n]$ of the 1D lattice is a subset where all the sites are adjacent. For an integer $\ell,$ we define the $\ell-$neighborhood of an interval $\sA = [a,b]$ to be the interval $\sA_{\ell} : = [a-\ell,b+\ell]$. Namely, it includes all sites within distance $\ell.$ A set of intervals is said to be contiguous if they are pairwise adjacent. 

Given two disjoint subsets $\sA, \sC\subseteq [n]$, we refer to their support as a $k$-\textit{sequence} if the entire chain can be partitioned into an alternating sequence of $(k+1)$ intervals, containing either $\sA$-sites or $\sC$-sites (but not both), separated by $k$ buffer regions $\sB$. If $\sA, \sC$ are already intervals, then the pair is a $1$-sequence. In general, for even $k$:
\begin{equation}\label{eq:many-intervals}
\sA_1\cup\sB_{1}\cup\sC_2\cup\sB_{2}\cup\sA_3\cdots \sC_{k}\cup\sB_{k}\cup \sA_{k+1} = [n]
\end{equation}
where $\sA\subseteq\cup_{i \text{ odd}} \sA_i, \sC\subseteq \cup_{i \text{ even}} \sC_i$, and $\sA_i\cap\sC=\emptyset$ (and vice-versa). The case of odd $k$ is analogous.

Ubiquitous in our argument is the restriction of the Hamiltonian terms to a subset of the chain. For any subset $\sR\subseteq [n]$, we define the restricted Hamiltonian and the associated induced Gibbs state:
\begin{align}
     \vH^{\sR} = \sum_{\gamma \subset \sR} \vH_{\gamma} \quad \text{and}\quad\vrho^{\sR} := \frac{\e^{-\beta \vH^\sR}}{\tr[\e^{-\beta \vH^\sR}]}. 
\end{align}
For any two disjoint intervals $\sA, \sB\subseteq[n]$, we refer to the Hamiltonian link between the two as
\begin{align}
    \vH^{\mathsf{A}:\mathsf{B}} :=  \vH^{\mathsf{A}\mathsf{B}} - \vH^{\mathsf{A}} -\vH^{\mathsf{B}}.
\end{align}
We refer the reader to \autoref{section:1D-locality} for further background on known locality properties of Gibbs states of 1D Hamiltonians.

\subsection{The KMS-Inner Product}
The Kubo-Martin-Schwinger (KMS) inner product of two operators $\vX, \vY$, weighted by a full rank density matrix $\vrho$, is defined by
\begin{align}
\langle \vX,\vY\rangle_{\vrho}:=\tr[\vX^\dagger \vrho^{\frac{1}{2}}\vY\vrho^{\frac{1}{2}}]\,.
\end{align}
We further denote by $\norm{\vX}_{\vrho} := \sqrt{\langle \vX,\vX\rangle_{\vrho}}$ the $\vrho$-weighted norm induced by the KMS inner product. Unconditionally, the operator norm upper bounds the KMS norm.

\begin{lem}\label{lem:operatornorm}
    For any state $\vrho$ and operator $\vX$, we have that $\norm{\vX}_{\vrho}\le \norm{\vX}$. 
\end{lem}

However, this norm conversion is almost certainly sub-optimal, especially in 1D systems. Instead, we would prefer a Holder-like inequality that is homogeneous in the KMS-norm. The following is a direct application of Holder's, once we put the right fractions of $\vrho$ in the right places.
\begin{lem}[Holder in KMS Norm, e.g., {~\cite[Lemma IX.4]{chen2025quantumMarkov}}]\label{lem:loose_AO}
For any pair of operators $\vA,\vO$, and full rank state $\vrho,$
\begin{align}
\lnormp{\vA\vO}{\vrho} &\le \norm{\vrho^{1/4}\vA\vrho^{-1/4}} \cdot\lnormp{\vO}{\vrho},\\
\lnormp{\vO\vA}{\vrho} &\le \norm{\vrho^{-1/4}\vA\vrho^{1/4}} \cdot\lnormp{\vO}{\vrho}.
\end{align}     
\end{lem}
\begin{proof} Rewrite the weighted norm in the Frobenius norm 
    \begin{align}
        \norm{\vA\vO}_{\vrho} = \norm{ \vrho^{1/4}\vA\vO\vrho^{1/4}}_2 =\norm{ \vrho^{1/4}\vA\vrho^{-1/4}\cdot\vrho^{1/4}\vO\vrho^{1/4}}_2&\le  \norm{ \vrho^{1/4}\vA\vrho^{-1/4}}\cdot\norm{\vrho^{1/4}\vO\vrho^{1/4}}_2 \\
        &= \norm{ \vrho^{1/4}\vA\vrho^{-1/4}}\norm{\vO}_{\vrho}.
    \end{align}
Repeat for $\norm{\vO\vA}_{\vrho}$ to conclude the proof. 
\end{proof}
Of course, the efficacy of the above is contingent on the convergence of complex-time dynamics of $\vA$ in 1D Hamiltonians. See \autoref{section:1D-locality} for a discussion. See~\cite[Lemma IX.5]{chen2025quantumMarkov} for a regularized variant of the above in generic systems. We also consider the following induced norm on superoperators:

\begin{defn}
    [KMS-Induced Superoperator Norm]\label{defn:induced_norm} For any superoperator $\CN:\CB(\CH)\rightarrow \CB(\CH)$ and full-rank state $\vrho$,
    \begin{align}
    \norm{\CN}_{\vrho}:= \sup_{\vO}\frac{\norm{\CN[\vO]}_{\vrho}}{\norm{\vO}_{\vrho}},
\end{align}

\end{defn}

\autoref{defn:induced_norm} may not be equal to the $1-1$ superoperator norm. For detailed-balanced superoperators (see \autoref{thm:ckg-db}), this is equal to the spectral radius.

\subsection{Operator Fourier Transforms}
\label{section:oft}

The definition of the Lindbladian dynamics we study \cite{chen2023efficient} hinges on the operator Fourier transform. We recall that the operator Fourier transform of an operator $\vA$ associated with the Hamiltonian $\vH$ can be written as 
\begin{align}\label{eq:OFT}
{\hat{\vA}_\sigma}(\omega)=  \frac{1}{\sqrt{2\pi}}\int_{-\infty}^{\infty} \e^{\ri \vH t} \vA \e^{-\ri \vH t} \e^{-\ri \omega t} \cdot f_\sigma(t)\rd t =\sum_{\nu\in B(\vH)} \vA_{\nu}\cdot \hat{f}_\sigma(\omega - \nu).
    \end{align}

If $\vH$ has the spectral decomposition $\vH=\sum_{i}E_i\vP_{E_i}$, then the Bohr frequencies $B(\vH)$ are the set of energy differences, and for $\nu\in B(\vH)$, $\vA_\nu:=\sum_{E_2-E_1=\nu}\vP_{E_2}\vA \vP_{E_1}$ denotes a component of the Bohr frequency decomposition of $\vA$. Finally, $f_\sigma$ denotes a Gaussian weight or filter function with an energy uncertainty parameter $\sigma>0$:
    \begin{align}
    f_\sigma(t) = \e^{-\sigma^2t^2}\sqrt{\sigma\sqrt{2/\pi}} \quad \text{and}\quad  
        \hat{f}_\sigma(\omega)=\frac{1}{\sqrt{2\pi}}\int_{-\infty}^{\infty}\e^{-\ri\omega t} f(t)\mathrm{d}t = \frac{1}{\sqrt{{\sigma}\sqrt{2\pi}}} \exp\L(- \frac{\omega^2}{4\sigma^2}\R).
        \label{eq:fwft}
    \end{align}
    Whenever implicit, we omit the subscript $\sigma$. The parameter $\omega$ will allow us to quantitatively and individually track the different parts of the operator $\vA.$

\begin{lem}[Decomposing an Operator by the Energy Change]\label{lem:sumoverenergies}
For any (not necessarily Hermitian) operator $\vA$, we have that
\begin{align}
\vA =  \frac{1}{\sqrt{2\sigma\sqrt{2\pi}}} \int_{-\infty}^{\infty} \hat{\vA}(\omega)\rd \omega.
\end{align}
\end{lem}

We further require the following bounds on the operator norm of operator Fourier transforms. In some sense, they quantify that the operator FT decays in norm in the large frequency $\omega$ regime. 

\begin{lem}[A priori Norm bounds on Operator Fourier Transforms, Corollary IX.2 \cite{chen2025quantumMarkov}]\label{lem:bounds_imaginary_conjugation}
For any $\beta ,\omega\in \BR$ and operator $\vA$ with norm $\norm{\vA}\le 1$, the operator Fourier transform $\hat{\vA}(\omega)$ with uncertainty $\sigma$~\eqref{eq:OFT},~\eqref{eq:fwft} satisfies 
\begin{align}
        \norm{\hat{\vA}(\omega)} \le \frac{\e^{-\beta \omega + \sigma^2 \beta^2}}{\sqrt{{\sigma}\sqrt{2\pi}}} \norm{\e^{\beta \vH}\vA\e^{-\beta \vH}}.
\end{align}
\end{lem}

We note the strength of $\hat{\vA}(\omega)$ is closely tied to the convergence of complex-time evolution; if for large $\beta$ the complex-time evolution of the local operator $\vA$ remains small, then $\hat{\vA}(\omega)$ decays faster.

\subsection{The Lindbladian with Exact Detailed Balance}
\label{sec:Lindbladian}

Consider a Hamiltonian $\vH$ on $n$ qudits, an inverse temperature $\beta>0$, an energy uncertainty $0<\sigma\leq \beta^{-1}$, and a single self-adjoint jump $\vA^a = \vA^{a\dagger}$. The (quasi-local) Lindbladian~\cite{chen2023efficient} is written as
\begin{align}
		\CL_a[\cdot] = \underset{\text{``coherent''}}{\underbrace{ -\ri [\vC^a, \cdot]}} + 
		\int_{-\infty}^{\infty} \gamma(\omega) \bigg(\underset{\text{``transition''}}{\underbrace{\hat{\vA}_\sigma^a(\omega)(\cdot)\hat{\vA}_\sigma^{a}(\omega)^\dagg}} - \underset{\text{``decay''}}{\underbrace{\frac{1}{2}\{\hat{\vA}_\sigma^{a}(\omega)^\dagg\hat{\vA}_\sigma^a(\omega),\cdot\}}}\bigg)\rd\omega,\label{eq:exact_DB_L}
\end{align}
\noindent where we either use the shifted-Metropolis or Gaussian weight function:
\begin{align}
    \gamma_\mathsf{M}(\omega) := \exp\L(-\beta\max\left(\omega +\frac{\beta \sigma^2}{2},0\right)\R), \quad \text{or}\quad \gamma_\mathsf{G}(\omega):= \exp\bigg(\frac{(\beta\omega+1)^2}{2(2-\beta^2\sigma^2)}\bigg).\label{eq:Metropolis}
\end{align} 
 \noindent The ``coherent part'' $\vC^{a}$ is a Hermitian operator:
\begin{align}
    &\vC^a = \iint_{-\infty}^{\infty} \gamma(\omega) \cdot c(t) \cdot \hat{\vA}_\sigma^a(\omega,t)^{\dagger}\hat{\vA}_\sigma^a(\omega,t)  \rd t\rd \omega, \quad  \\\text{with}\quad  &c(t) := \frac{1}{\beta\sinh(2\pi t/\beta)} \quad \text{and}\quad  \hat{\vA}_\sigma^a(\omega,t):=\e^{i\vH t}\hat{\vA}_\sigma^a(\omega)\e^{-i\vH t}.
\end{align}

When $\norm{\vA^a}=1$, we have system-size independent norm bounds $\norm{\CL_a}_{1-1}\le O(1)$ for the Gaussian weight, and $\norm{\CL_a}_{1-1}\le O(\beta)$ for the Metropolis weight~\cite[Corollary A.1]{chen2025quantumMarkov}. In our actual algorithm, we will make the explicit choice of energy width $\sigma=\frac{1}{\beta}$, and the collection of jump operators $\{\vA^a\}_a = \CS_{[n]}^1$ contains all single-site Pauli operators on the chain. For transparency, however, we will often keep $\sigma$ and $\{\vA^a\}_a$ as tunable parameters in our lemmas.

The defining feature of the above Lindbladian is achieving detailed balance while preserving locality. Indeed, the time-evolved operators respect the spatial locality of the Hamiltonian (up to some exponentially decaying tail  with distance) 
due to Lieb-Robinson bounds~\cite{Lieb1972,hastings2006spectral,chen2023speed}. Meanwhile, the Lindbladian satisfies exact KMS-detailed-balance, which plays a subtle yet crucial role in enabling exact mathematical identities. We recall that, given a full-rank state $\vrho$, a Lindbladian $\CL$ is said to satisfy the Kubo-Martin-Schwinger (KMS)-$\vrho$-detailed-balance if it is symmetric with respect to the KMS inner product associated with $\vrho.$

\begin{thm}[Lindbladian with Exact Detailed Balance{~\cite{chen2023efficient}}]\label{thm:ckg-db}
The Lindbladian $\CL=\sum_a \CL_a$ defined in Equation \eqref{eq:exact_DB_L} satisfies KMS-$\vrho$-detailed-balance 
\begin{align}
    \braket{\vX,\CL^{\dagger}[\vY]}_{\vrho} = \braket{\CL^{\dagger}[\vX],\vY}_{\vrho},
\end{align}
and hence fixes the Gibbs state exactly $\CL[\vrho] =0$ for $\vrho\propto \e^{-\beta \vH}.$
\end{thm}

\begin{rmk}
Other $\vrho$-weighted norms are possible (e.g., GNS~\cite{kastoryano2016commuting}), but in this paper, we will only consider the KMS inner product since the Lindbladian we consider is only guaranteed to be KMS-detailed-balanced.    
\end{rmk}

Our goal in this paper is to establish a lower bound on the spectral gap of $\CL$ when defined on a 1D Hamiltonian, at any fixed finite inverse-temperature $\beta>0.$ We operate under the following definition of the spectral gap of a primitive Lindbladian, see \cite{temme2010chi, Kastoryano_2013} for an extended discussion.

\begin{defn}[The Spectral Gap]
For any KMS-$\vrho$-detailed-balanced Lindbladian $\CL$, the \emph{spectral gap} of $\CL$ is defined by the minimal ratio of the Dirichlet form to the variance:
\begin{align}
        \lambda_{\mathsf{gap}}(\CL) = \inf_{\vO} \frac{\braket{\vO, -\CL^\dagger[\vO]}_{\vrho}}{\norm{\vO - \vI \cdot \tr[\vO\vrho]}_{\vrho}^2}.   
\end{align}
\end{defn}

\section{Spectral Conditional Expectations}
\label{sec:CE}

The study of mixing times often hinges on the interplay between locality and detailed balance. 
In the classical setting, a key tool is the systematic use of boundary conditioning, which naturally gives rise to the conditional Gibbs distribution and, in turn, to the notion of a conditional expectation over a local patch or region of spins. 
In quantum systems, non-commutative versions of conditional expectations can, in principle, be defined. 

\begin{defn}[Conditional Expectation, restatement of \autoref{defn:CE_intro}]\label{defn:expectation}
We say a superoperator $\BE:\CB(\CH)\rightarrow \CB(\CH)$ in the Heisenberg picture is a \emph{conditional expectation} with respect to a full rank mixed state $\vrho$ if it satisfies:
    \begin{itemize}
        \item \emph{KMS-detailed-balance:} for any $\vX,\vY,$ $\braket{\BE[\vY],\vX }_{\vrho}=\braket{\vY,\BE[\vX]}_{\vrho}.$
        \item \emph{Complete Positivity}.
        \item \emph{Unitality:} $\BE[\vI]=\vI.$ 
        \item \emph{Monotonicity:} The spectrum of $\BE$ is real and contained in $[0,1].$ 
    \end{itemize}
\end{defn}
In other words, a conditional expectation is simply the adjoint of a $\vrho-$detailed-balanced quantum channel, with additional constraints on the spectrum. Of course, implicitly, any workable conditional expectation must also respect the locality of the system. For commuting Hamiltonians, explicit conditional expectations have been defined --such as the time evolution of the Davies' or heat-bath generators-- which has enabled a systematic approach to commuting Gibbs sampling, see \cite[Section 3]{kastoryano2016commuting}. 

A central challenge in the non-commuting case is the lack of an exact Markov property \cite{chen2025quantumMarkov}, which rules out the possibility of a genuinely local conditional expectation. Unfortunately, the naive means of defining a quasi-local conditional expectation --such as truncating a quasi-local Lindbladian evolution in space-- oftentimes breaks the underlying detailed balance condition, which is equally indispensable. The central ingredient behind our results is an explicit, workable version of a ``pseudo'' conditional expectation map in one dimension that respects both locality and detailed balance. Intriguingly, our proposal has all relevant features of standard conditional expectations, except for the \textit{complete positivity} condition. Accordingly, we refer to these maps as \textit{spectral conditional expectations}, as they serve all the purposes for proving a spectral gap. 

This section defines and explores such spectral conditional expectations. We begin in \autoref{sec:explicit-E} by presenting the relevant definitions and main results of this section. Next, in \autoref{sec:generator-K} we introduce an auxiliary generator $\CK$, whose limiting time-evolution gives rise to the spectral conditional expectation $\tBE$. In \autoref{sec:apriori-gap-K}, we prove the generator admits a system-size-independent \textit{local gap}, and subsequently in \autoref{sec:quasi-local-E} we prove this local gap entails $\tBE$ is quasi-local.

\subsection{An Explicit Spectral Conditional Expectation}
\label{sec:explicit-E}
In this paper, the following variant of the condition expectation maps will be used extensively. 
\begin{defn}[Spectral Conditional Expectation]\label{defn:spectral_expectation}
We say a superoperator $\tBE:\CB(\CH)\rightarrow \CB(\CH)$ in the Heisenberg picture is a \emph{spectral conditional expectation} with respect to a full rank mixed state $\vrho$ if it satisfies:
    \begin{itemize}
        \item \emph{KMS-Detailed-Balance:} for any $\vX,\vY,$ $\braket{\tBE[\vY],\vX }_{\vrho}=\braket{\vY,\tBE[\vX]}_{\vrho}.$
        \item \emph{Unitality:} $\tBE[\vI]=\vI.$ 
        \item \emph{Monotonicity:} The spectrum of $\tBE$ is real and contained in $[0,1].$
    \end{itemize}
\end{defn}
Even though a spectral conditional expectation does not constitute a physical, implementable quantum map, it remains a valid self-adjoint operator with the correct top eigenvector. Sacrificing complete positivity allows us to construct an explicit map that will be the backbone of the proof of the spectral gap.

\begin{lem}[An Explicit Spectral Conditional Expectation]\label{lem:explicit_CE}
    For any operator $\vY$ and subset $\sR\subset [n]$, the following map defines a spectral conditional expectation (\autoref{defn:spectral_expectation}):
\begin{align}
        \tilde{\BE}_\sR[\vY] := \sum_i \braket{\vB^i_{\bar{\sR}},\vY}_{\vrho} \cdot \vB^i_{\bar{\sR}}\label{eq:ER},
    \end{align}
    where the set of operators $\{\vB^i_{\bar{\sR}}\}_i$ can be any KMS-orthonormal basis for $\CB(\CH_{\bar{\sR}})$ such that $\braket{\vB^i_{\bar{\sR}},\vB^j_{\bar{\sR}}}_{\vrho} = \delta_{ij}.$ In particular, we pick identity to be the first basis vector $\vI=\vB^1_{\bar{\sR}}.$ 
\end{lem}

\begin{proof}
    $\tBE_\sR$ is the orthogonal projection onto the operators supported on the complement region $\CB(\CH_{\bar{\sR}})$, w.r.t. the KMS inner product. Therefore, all eigenvalues are either 0 or 1. Detailed balance follows from the identity:
    \begin{align}
        \braket{\tBE[\vY],\vX }_{\vrho} = \sum_i \braket{\vB^i_{\bar{\sR}},\vY}_{\vrho}^* \cdot \braket{\vB^i_{\bar{\sR}},\vX}_{\vrho} = \sum_i \braket{\vY, \vB^i_{\bar{\sR}}}_{\vrho} \cdot \braket{\vB^i_{\bar{\sR}},\vX}_{\vrho} =  \braket{\vY,\tBE[\vX]}_{\vrho}.
    \end{align}
    Unitality follows from $\braket{\vI,\vI}_{\vrho} = 1$. 
\end{proof}

Henceforth, $\tBE$ will always refer to the family of spectral conditional expectations of \autoref{lem:explicit_CE}. In this linear algebraic view, it follows that the above conditional expectation is consistent with the inclusion of subsets, since $\sR_1 \subseteq \sR_2$ implies that $\CB(\CH_{\bar{\sR}_2}) \subseteq \CB(\CH_{\bar{\sR}_1}).$

\begin{lem}[Consistency]
    For any two subsets such that $\sR_1\subseteq \sR_2 \subseteq [n],$ the family of spectral conditional expectations $\{\tBE_\sR\}_{\sR\subseteq [n]}$ defined in \autoref{lem:explicit_CE} satisfies:
\begin{equation}
   \tBE_{\sR_1} \tBE_{\sR_2} =  \tBE_{\sR_2}. 
\end{equation}
In particular, this spectral conditional expectation is \emph{projective} $\tBE_{\sR}^2=\tBE_{\sR}.$
\end{lem}
 For the proof of spectral gap, what makes the above map useful is its quasi-locality (recall \autoref{fig:quasilocal_E}), which may not be a priori obvious from the definition. The proof is deferred to \autoref{sec:quasi-local-E}.
\begin{lem}
    [The Quasi-Locality of the Spectral Conditional Expectation]\label{lem:localized_E} Fix two subsets $\sA\subseteq \sR\subseteq [n]$ of the sites of a 1D Hamiltonian (\autoref{defn:1D}). The spectral conditional expectation $\tBE_\sA$ defined w.r.t $\vH$ in \autoref{lem:explicit_CE}, can be well-approximated by that w.r.t. the restriction $\vH^{\sR}$ to $\sR$: 
    \begin{equation}
        \|\tBE_\sA - \tBE_{\sA}^{(\sR)}\|_{\vrho} \leq c_1^{|\partial \sR|}\cdot \exp\bigg(c_2\cdot  |\sA| - c_3\cdot \ell \log \ell\bigg),
    \end{equation}
    where $\ell = \mathsf{dist}(\sA, [n]\setminus \sR)$,  the number of end-points of the subset $\sR$ is written as $|\partial\sR|$, and $c_1, c_2, c_3$ are explicit constants which depend only on $\beta, q$.
\end{lem}

The truncated spectral conditional expectation above is no longer detailed-balanced w.r.t. $\vrho,$ however, it serves as a handy intermediate object that makes the locality of $\tBE_{\sA}$ explicit. In particular, for any distant two intervals $\sA, \sC$ such that $\mathsf{dist}(\sA, \sC) \geq \alpha\cdot |\sA\sC|/\log |\sA\sC|$, the conditional expectations on $\sA, \sC$ approximately factorize (\autoref{fig:EAEC}):
\begin{equation}
    \tBE_\sA \tBE_\sC \approx \tBE_{\sA\sC}.
\end{equation}

Curiously, the best way to understand the above is by introducing a \textit{dynamical} map that converges to the above spectral conditional expectations. In our later use, we will actually go back and forth between the spectral conditional expectations and the dynamical generator.

\begin{figure}[t]
\includegraphics[width=0.8\textwidth]{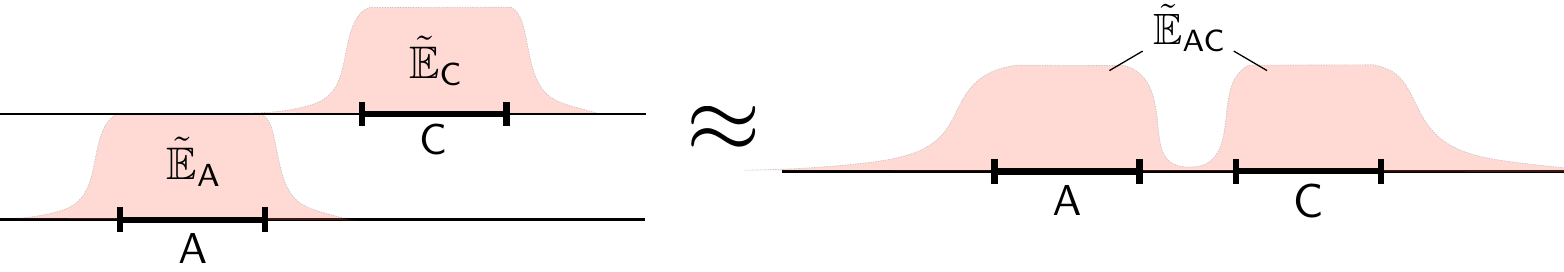}
\caption{
For any set $\sA$, $\sC$, the spectral conditional expectation is well-approximated by the composition of individual ones.}\label{fig:EAEC}

\end{figure}

\subsection{A Generator for our Spectral Conditional Expectations}
\label{sec:generator-K}
Our spectral conditional expectation originates from exponentiating a specific auxiliary generator $\CK.$ 

\begin{defn}[The Auxiliary Generator]\label{defn:CK} 
For any jump operator $\vA^a\in \CB(\CH)$ and full-rank mixed state $\vrho$, define the superoperator $\CK_a:\CB(\CH)\rightarrow \CB(\CH)$ in the Heisenberg picture:
\begin{align}
    \CK_a[\vX] &:= [\vA^a,\vX] \bigg(\vrho^{1/2} \vA^{a\dagger}\vrho^{-1/2}\bigg)-  \bigg(\vrho^{-1/2} \vA^{a\dagger}\vrho^{1/2}\bigg)[\vA^a,\vX] . \label{eq:CK}
\end{align}
\end{defn}

As $\beta = 0$, we recover the generator for the depolarizing semi-group. To shorthand notation, we often write the complex-time evolution as $\vA(i\beta/2) = e^{-\beta\vH/2}\vA e^{\beta\vH/2}$. As we discuss shortly, the role of the factors of $\vrho$ is to ensure KMS-detailed-balance. The motivation and concrete meaning behind this generator $\CK$ is best captured by the following formula for its Dirichlet form.

\begin{lem}
    [The Dirichlet Form of $\CK$]\label{lem:dirichlet-ck} For any (not-necessary Hermitian) jump operator $\vA^a\in \CB(\CH)$, \emph{Dirichlet form} of the local generator $\CK_a$ from \autoref{defn:CK} can be written as the KMS-norm of a commutator. For any operator $\vX\in \CB(\CH)$:
    \begin{align}
        \braket{\vX,-\CK_a[\vX]}_{\vrho} =\|[\vA^a, \vX]\|_{\vrho}^2. \label{eq:K_dirichlet}
    \end{align}
\end{lem}

In words, the time-derivative of an operator $\vX$ w.r.t. $\CK_a$ is measured in terms of a norm of commutators of $\vX$ with $\vA^a$; i.e., how non-trivial $\vX$ is on the support of $\CK_a$.

\begin{proof}
    Omit the superscript $\vA^a=\vA$. We have that for any pair of operators $\vX, \vY$,
    \begin{align}
        \braket{\vY,-\CK_a[\vX]}_{\vrho} &=\braket{\vY,\vA(i\beta/2)^{\dagger}[\vA,\vX]}_{\vrho}  - \braket{\vY,[\vA,\vX] \vA(-i\beta/2)^{\dagger}}_{\vrho}\\&= \tr[[\vY^\dagger,\vA^\dagger]\vrho^{1/2}[\vA, \vX]\vrho^{1/2}] \\&= \braket{[\vA, \vY],[\vA, \vX]}_{\vrho} \label{eq:db-of-K}
    \end{align}
    with $\vX = \vY$ we attain the desired expression.
\end{proof}

 The commutator-square form also implies that $\CK$ must generate a spectral conditional expectation.

\begin{lem}
    [Generators of Spectral Conditional Expectations]\label{lem:finite_time_ce} For any operator $\vA^a\in \CB(\CH),$ the generator $\CK_a$ from \autoref{defn:CK} satisfies:
    \begin{enumerate}
        \item  \emph{KMS-detailed-balance:} $\braket{\CK_a[\vY],\vX }_{\vrho}=\braket{\vY,\CK_a[\vX]}_{\vrho}$
    \item \emph{Trace-preservation:} $\CK_a[\vI] =0$.
    \item \emph{Monotonicity:} The eigenvalue spectrum of $\CK_a$ is real and non-positive.
    \end{enumerate}
    Therefore, for any $t\geq 0$, and any set of jump operators $\{\vA^a\}$, the exponential map $e^{t\CK}[\cdot]$ for $\CK = \sum_a \CK_a$ defines a spectral conditional expectation as in \autoref{defn:spectral_expectation}. 
\end{lem}

\begin{proof}
Let us first argue properties (1-3). Trace preservation follows from the definition in \eqref{eq:CK} and $\vX = \vI$. Since for every $\vX$, the Dirichlet form \eqref{eq:K_dirichlet} is (real and) non-positive, so is the spectrum of $\CK$. Finally, detailed balance follows from the symmetry derived in the computation in \eqref{eq:db-of-K}:
\begin{align}
    \braket{\vY,\CK_a[\vX]}_{\vrho} = -\braket{[\vA, \vY],[\vA, \vX]}_{\vrho}  = -\braket{[\vA, \vX],[\vA, \vY]}_{\vrho}^* =  \braket{\vX,\CK_a[\vY]}_{\vrho}^*=\braket{\CK_a[\vY], \vX}_{\vrho}.
\end{align}
Since $x\leq 0\Rightarrow e^x\in [0,1]$ and the spectrum of $\CK$ is non-positive, the spectrum of the exponential map is in $[0, 1]$. 
\end{proof}

For our purposes, we will henceforth take the jumps $\vA^a$ to be all the single-site Pauli operators $\CS^1_{\sR}$ acting on a region $\sR\subseteq [n]$:
\begin{align}
    \CK_{\sR} := \sum_{a\in \CS_\sR^1} \CK_a. \label{eq:defKR}
\end{align}

Why do we even bother introducing another super-operator, if the detailed balance Lindbladian already satisfies all the above and is also more physical?  The key merit of introducing $\CK$ is that its kernel is explicitly characterizable, and most importantly, respects locality, which $\CL$ may not.

\begin{lem}[The Kernel of $\CK_{\sR}$ \eqref{eq:defKR}] Fix any subset $\sR\subseteq [n]$ and operator $\vX\in \CB(\CH)$. Then, the generator $\CK_\sR$ \eqref{eq:defKR} with all single-site Pauli jumps on $\sR$ satisfies:
\begin{align}
    \CK_{\sR}[\vX] = 0 \quad \text{if and only if}\quad  \vX \in \CB(\CH_{\bar{\sR}}).
\end{align}    
\end{lem}
\begin{proof} Consider the Dirichlet form (\autoref{lem:dirichlet-ck}):
    \begin{align}
        \CK_{\sR}[\vX] = 0 &\iff -\braket{\vX,\CK_\sR[\vX]}_{\vrho} = 0 \\
        &\iff [\vA^a,\vX] = 0 \quad \text{for each}\quad \vA^a \in \CS_{\sR}^1\\
        &\iff \vX \in \CB(\CH_{\bar{\sR}}).
    \end{align}
    The first line leverages that $-\CK_{\sR}$ is detailed-balanced and has non-negative eigenvalues.
\end{proof}
Due to detailed balance, the limiting self-adjoint map is completely determined by the kernel, which is exactly the advertised spectral conditional expectation.

\begin{lem}[Recovering Spectral Conditional Expectations] The infinite time limit of~\eqref{eq:defKR} is exactly the spectral conditional expectation~\eqref{eq:ER},
    \begin{align}
        \tBE_\sR[\cdot] = \lim_{t\rightarrow \infty} e^{t\CK_\sR}[\cdot]. \label{eq:inf_time}
    \end{align}
\end{lem}
\begin{proof}
    Since the $\CK_\sR$ gives a finite-dimensional, self-adjoint (w.r.t. KMS-inner product) linear operator whose spectrum is real and non-positive, the effect of the infinite-time conditional expectation is exactly the orthogonal projection onto the zero-eigenspace, namely the kernel of $\CK_\sR$.
\end{proof}

This dynamical view of $\tBE_\sR[\cdot]$ through $\CK$ is what enables us to fully expose the locality of $\tBE_\sR[\cdot].$ In particular, the fact that the kernel respects a strict locality allows us to quantify a rate of local convergence to the kernel, in terms of a conditional gap that is independent of the global system size. This gap then implies the infinite-time limit, must also inherit the locality of $\CK.$ On the other hand, we do not know whether the infinite-time limit for the Lindbladian $\CL$ \eqref{eq:exact_DB_L} \cite{chen2023efficient} should generally enjoy any locality at all.

\subsubsection{Physicality of $\CK$}

To conclude the introduction to $\CK$, we briefly remark on why it does not generally give a physical Lindbladian. In the framework of~\cite{chen2023efficient}, it is most natural to rewrite the Dirichlet form in terms of Bohr-frequency decompositions $\vA_{\nu}.$ Since $\vA = \sum_{\nu} \vA_{\nu},$ we have that
    \begin{align}
        \tr[ \sqrt{\vrho}[\vA^a,\vX]^{\dagger} \sqrt{\vrho}[\vA^a,\vX]]= \sum_{\nu_1,\nu_2\in B(\vH)} \bar{\alpha}_{\nu_1\nu_2} \tr[ \sqrt{\vrho}[\vA_{\nu_2}^a,\vX]^{\dagger} \sqrt{\vrho}[\vA^a_{\nu_1},\vX]],
    \end{align}
    where
    \begin{align}
        \bar{\alpha}_{\nu_1,\nu_2} = 1\quad \text{for each}\quad \nu_1,\nu_2 \in B(\vH). 
    \end{align}
    Indeed, the all-ones matrix $\bar{\alpha}$ is PSD. However, if we try to reverse-engineer the transition part $\CT = \sum_{\nu_1,\nu_2} \alpha_{\nu_1,\nu_2} \vA^a_{\nu_1}\cdot (\vA_{\nu_2}^a)^{\dagger}$ for a hypothetical KMS detailed-balanced Lindbladian (e.g.,~\cite[Lemma C.2]{rouze2024efficient}) that reproduces the exact Dirichlet form, the matrix
    \begin{align}
        \alpha_{\nu_1,\nu_2} = \frac{2\cosh(\beta(\nu_1-\nu_2)/4)}{\e^{\beta(\nu_1+\nu_2)/4}},
    \end{align}
    is not generally PSD, therefore the $\CT$ is not always completely positive.

    For thermal state preparation, even if $\CK$ does not give a physical Lindbladian dynamics, its \textit{parent Hamiltonian} or discriminant on a duplicated Hilbert space can feed into an adiabatic algorithm to prepare the purified Gibbs state. See Section \ref{sec:results} for a discussion. 

\begin{lem}\label{lem:spectral-radius}
    For any single-site jump operator $\vA^a,$ the spectral radius of the generator $\CK_a$ is finite for 1D Hamiltonians (\autoref{defn:1D}). Indeed, 
\begin{align}
    \labs{\braket{\vX,\CK_a[\vX]}_{\vrho}} \le \big(\norm{\vrho^{\frac{1}{4}}\vA\vrho^{-\frac{1}{4}}}+\norm{\vrho^{-\frac{1}{4}}\vA\vrho^{\frac{1}{4}}}\big)^2 \cdot \norm{\vX}^2_{\vrho} \leq c_{\beta} \cdot \norm{\vX}^2_{\vrho},
\end{align}
for an appropriate constant dependent on $\beta$. 
\end{lem}
Again, for $\CK$ to be well-behaved at large system sizes, we rely on the convergence of imaginary time evolution in 1D (\autoref{lem:locality_complextime}).

\subsection{An a priori Spectral Gap for $\CK$.}
\label{sec:apriori-gap-K}

Key to the locality of the spectral conditional expectation map $\tBE$ \eqref{eq:inf_time} are the spectral properties of $\CK$. Due to its particularly structured form, one can show that $\CK_{\sR}$ admits a system-size independent \textit{conditional spectral gap}, building on a chain of definitions. 
\begin{defn}
    [The Conditional Variance]\label{defn:variance} Let $\{\tBE_\sR\}_{\sR\subseteq [n]}$ be the family of spectral conditional expectations as in \autoref{lem:explicit_CE}. For any subset $\sR\subseteq [n]$, the conditional variance of an operator $\vX\in\CB(\CH)$ is defined by:
    \begin{equation}
        \mathsf{Var}_\sR(\vX):= \|\vX - \tBE_\sR[\vX]\|^2_{\vrho}.
    \end{equation}
\end{defn}
For the full lattice, $\mathsf{Var}_{[n]}(\vX) = \norm{\vX-\tr[\vX\vrho]}_{\vrho}^2.$ Naturally, the conditional variance defines a \textit{conditional spectral gap}, associated to $\CK_{\sR}.$ 
\begin{defn}[The Conditional Spectral Gap] For any subset $\sR \subseteq [n]$, the conditional spectral gap of $\CK_\sR$~\eqref{eq:defKR} is defined by
    \begin{align}
        \lambda_{\sR}:=\inf_{\vX\in\CB(\CH)} \frac{\braket{\vX, -\CK_\sR[\vX]}}{\mathsf{Var}_\sR(\vX)}.
    \end{align}
\end{defn}
We always have that $ 0< \lambda_{\sR} \le \labs{\sR}\cdot c_{\beta,q}$, due to the bound on the spectral radius of the individual $\CK_a$ in \autoref{lem:spectral-radius}. The main lemma of this subsection, is an \textit{a priori lower bound} on the conditional spectral gap of the super-operator $\CK_{\sR}$. Even though our eventual bounds will be much improved for large subsets, these sub-optimal lower bounds on the gap are necessary for the base case of the recursion. While the following can be defined for any subset $\sR$ that may consist of an arbitrary union of intervals, our spectral gap recursion will focus on particular geometries.  
\begin{lem}[$\CK_{\sR}$ is Locally Gapped]\label{lem:local_gap_K} For any subset $\sR \subseteq [n]$, the conditional spectral gap of $\CK_\sR$ is at least:
\begin{align}
    \lambda_\sR \ge \frac{1}{(c_\beta\cdot 2^{q})^{|\sR|}},
\end{align}  
for an explicit constant $c_\beta\geq 1$ which depends only on $\beta$.
\end{lem}
Although for any finite interval $\sA\subseteq [n]$, the superoperator $\CK_\sA$ may generically be non-local (i.e., act on the entire chain), \autoref{lem:local_gap_K} implies it at least admits some locality in time. Indeed, the finite lower bound on the conditional gap above implies that for a suitably chosen finite time $t$ scaling exponentially in $|\sA|$, the finite-time spectral conditional expectation approximates its infinite-time version: $\BE_{\sA, t}[\vX]\approx \BE_\sA[\vX]$.

\subsubsection{Proof of~\autoref{lem:local_gap_K}}
The proof of \autoref{lem:local_gap_K} (below) is based on a comparison to the depolarizing semi-group (\autoref{lem:gap_product_L}), and the fact that one can change the underlying state to the maximally mixed at a controllable loss (\autoref{lem:compare_measure}).

\begin{lem}[Spectral Gap of Depolarizing Semi-group]\label{lem:gap_product_L} With $\vec{\tau}_\sR$ the maximally mixed state over $\otimes_\sR \mathbb{C}^{2^q}$, $\vsigma_{\bar{\sR}}$ any mixed state over $\otimes_{\bar{\sR}} \mathbb{C}^{2^q}$, and any operator $\vX$:
    \begin{align}
        \sum_{a\in \CS_{\sR}^1} \norm{[\vA^a,\vX]}^2_{\vec{\tau}_\sR\otimes \vsigma_{\bar{\sR}}} \ge 4^{q} \cdot \norm{\vX- \vec{\tau}_\sR\otimes\tr_\sR[\vX]}^2_{\vec{\tau}_\sR\otimes \vsigma_{\bar{\sR}}}.
    \end{align}
\end{lem}

The proof can be found in e.g. \cite{kastoryano2013quantum} and is included for completeness. 

\begin{proof}
    With $\BP_i[\cdot] = \vec{\tau}_i \otimes \tr_i[\cdot ]$ the erasure channel on site $i$, we have
    \begin{equation}
        \| \vX-\vec{\tau}_\sR\otimes\tr_\sR[\vX]\|^2_{\vec{\tau}_\sR\otimes \vsigma_{\bar{\sR}}} = \sum_{i\in \sR} \bigg\| \prod_j^{i-1} \BP_j(1-\BP_i)  [\vX]\bigg\|^2_{\vec{\tau}_\sR\otimes \vsigma_{\bar{\sR}}} \leq \sum_{i\in \sR}\bigg\| (1-\BP_i)  [\vX]\bigg\|^2_{\vec{\tau}_\sR\otimes \vsigma_{\bar{\sR}}}.
    \end{equation}
    We can further expand the variance of the single-site depolarizing channel in terms of double commutators with $q$-qubit Pauli operators:
    \begin{equation}
        (1-\BP_i)  [\vX] = \frac{1}{2}\cdot \frac{1}{4^q}\cdot \sum_{a\in \CS_i^1} [\vA^a, [\vA^a, \vX]] \Rightarrow \bigg\| (1-\BP_i)  [\vX]\bigg\|_{\vec{\tau}_\sR\otimes \vsigma_{\bar{\sR}}} \leq \frac{1}{4^q}\cdot  \sum_{a\in \CS_i^1}\norm{[\vA^a,\vX]}_{\vec{\tau}_\sR\otimes \vsigma_{\bar{\sR}}},
    \end{equation}
    where Holder's inequality was applied with $\|\vA^a\|\leq 1$. The advertised result then follows from the Cauchy-Schwarz inequality. 
\end{proof}

\begin{lem}[Comparison of Measures]\label{lem:compare_measure}
Consider a 1D Hamiltonian $\vH$ (\autoref{defn:1D}) and a bipartition of the chain into two disjoint subsets  $\sA\cup\sB= [n].$ Consider the Hamiltonian $\vH^\sA+\vH^\sB$ with the link $\vH^{\sA:\sB}$ between the intervals removed, and let $\vrho^{\sA\sB}$ and $\vrho^{\sA}\otimes \vrho^{\sB}$ be the associated Gibbs states. Then, there exists an explicit constant $\infty>c_\beta>1$ such that for any operator $\vX$,
\begin{align}
    \frac{1}{c_\beta}\cdot \norm{\vX}_{\vrho^{\sA\sB}} \le  \norm{\vX}_{\vrho^{\sA}\otimes \vrho^{\sB}} \leq c_\beta\cdot \norm{\vX}_{\vrho^{\sA\sB}}. \label{eq:measure_comparison}
\end{align}    
\end{lem}
This convenient lemma is also applied throughout the paper, e.g., often when we truncate the spectral conditional expectation. 
\begin{proof}
Let us denote $\vsigma = \vrho^{\sA}\otimes \vrho^{\sB}.$ Expanding the KMS inner product,
\begin{align}
    \tr[\sqrt{\vrho} \vX \sqrt{\vrho}\vX^{\dagger}] &= \tr\bigg[\vsigma^{1/4} \vX \vsigma^{1/4}\cdot  (\vsigma^{-1/4} \vrho^{1/2}\vsigma^{-1/4} )\cdot \vsigma^{1/4}\vX^{\dagger}\vsigma^{1/4} \cdot (\vsigma^{-1/4} \vrho^{1/2}\vsigma^{-1/4}) \bigg] \\
    &\le \norm{\vsigma^{-1/4} \vrho^{1/2}\vsigma^{-1/4}}^2 \cdot \norm{\vX}_{\vsigma}^2. 
\end{align}
To analyze the operator norm of $\vsigma^{-1/4} \vrho^{1/2}\vsigma^{-1/4} = \sqrt{\frac{Z_\sA Z_\sB}{Z_{\sA\sB}}} e^{\beta \vH_{\sA}/4} e^{\beta \vH_{\sB}/4}e^{-\beta \vH_{\sA\sB} /2} e^{\beta \vH_{\sA}/4} e^{\beta \vH_{\sB}/4} $, we apply \autoref{lem:expansionals} to the matrix exponentials, and the Golden-Thompson inequality to the ratio of partition functions:
\begin{equation}
    Z_\sA Z_\sB = \tr[e^{-\beta(\vH_\sA+\vH_\sB)}]\leq \tr[e^{-\beta\vH_{\sA\sB}}e^{\beta \vh}] = Z_{\sA\sB}\tr[\vrho e^{\beta\vh}] \leq Z_{\sA\sB} e^{\beta},
\end{equation}
where $\vh$ is the Hamiltonian term on the link between $\sA, \sB$, and $\|\vh\|\leq 1$. The opposite direction of \eqref{eq:measure_comparison} follows analogously
\begin{equation}
    Z_{\sA\sB} = \tr[e^{-\beta(\vH_\sA+\vH_\sB+\vh)}] \leq \tr[e^{-\beta\vH_{\sA}-\beta\vH_{\sA}}e^{-\beta \vh}] = Z_{\sA}Z_{\sB}\tr[\vsigma e^{-\beta\vh}] \leq Z_{\sA}Z_{\sB} e^{\beta}.
\end{equation}
\end{proof}

We are now in a position to prove \autoref{lem:local_gap_K}, on the conditional spectral gap of $\CK$.

\begin{proof}

[of \autoref{lem:local_gap_K}] 
For any operator $\vX,$ consider its Pauli string decomposition over the subsystem $\sR$:
\begin{align}
     \vX = \sum_{\vS^i \in \CS^{\labs{\sR}}_\sR} \vS^i\otimes\vX^i_{\bar{\sR}},
\end{align}
where we start with the identity $\vS^0 = \vI.$ Consider also the restricted Gibbs state $\vrho^{\bar{\sR}} \propto e^{-\beta \vH_{\bar{\sR}}}$, which is a product state between $\bar{\sR}$ and $\sR$ such that the terms $\vS^i\otimes\vX^i_{\bar{R}}$ are pairwise orthogonal under $\braket{\cdot,\cdot}_{\vrho^{\bar{\sR}}}.$ Then, 
    \begin{align}
    \sum_{a\in \CS_{\sR}^1} \braket{\vX,-\CK_a[\vX]}_{\vrho} =\sum_{a\in \CS_{\sR}^1} \norm{[\vA^a,\vX]}^2_{\vrho} &\ge c_\beta^{-2|\sR|} \sum_{a\in \CS_{\sR}^1} \norm{[\vA^a,\vX]}^2_{\vrho^{\bar{\sR}}}\tag*{(By~\autoref{lem:compare_measure})}\\
    &\ge c_\beta^{-2|\sR|} \cdot \lnorm{\vX -\vec{\tau}_\sR\otimes\tr_\sR[\vX]}^2_{\vrho^{\bar{\sR}}} \tag*{(By~\autoref{lem:gap_product_L} and $4^q \ge 1$)}\\
    &= c_\beta^{-2|\sR|} \cdot\sum_{\vS^i \in \CS_{\sR}^{\labs{\sR}} \setminus \{\vI\}} \lnorm{\vS^i\otimes \vX^{i}_{\bar{\sR}}}^2_{\vrho^{\bar{\sR}}} \\ & \ge c_\beta^{-4|\sR|} \cdot\sum_{\vS^i \in \CS_{\sR}^{\labs{\sR}} \setminus \{\vI\}} \lnorm{\vS^i\otimes \vX^{i}_{\bar{\sR}}}^2_{\vrho} \tag*{(By~\autoref{lem:compare_measure})}.
    \end{align}
    \noindent 
    Next, we express the variance in terms of the individual Pauli decomposition
    \begin{align}
        \norm{(1-\BE_{R})[\vX]}_{\vrho} &\le \sum_{\vS^i \in \CS_{\sR}^{\labs{\sR}} \setminus \{\vI\}} \lnorm{(1-\BE_{R})[\vS^i\otimes\vX^{i}_{\bar{\sR}}]}_{\vrho}\tag*{(Triangle inequality)}\\
        &\le \sum_{\vS^i \in \CS_{\sR}^{\labs{\sR}} \setminus \{\vI\}} \norm{\vS^i\otimes\vX^{i}_{\bar{\sR}}}_{\vrho} \tag*{(Since $\norm{1-\BE_{\sR}}_{\vrho}\le 1$)}\\
        &\le  2^{q|\sR|}\cdot  \sqrt{\sum_{\vS^i \in \CS_{\sR}^{\labs{\sR}} \setminus \{\vI\}} \norm{\vS^i\otimes\vX^{i}_{\bar{\sR}}}^2_{\vrho}} \tag*{(Cauchy-Schwarz)}.
    \end{align}
    Chain the inequalities to conclude the proof.   
\end{proof}

\subsection{Quasi-Locality of the Spectral Conditional Expectation}
\label{sec:quasi-local-E}
Now that we know the infinite-time conditional expectation can be well-approximated at finite times, the quasi-locality of $\tBE$ can be related to that of $\CK.$
Our first goal is to show that the Hamiltonian generating $\CK_\sA$ \eqref{eq:CK} can be truncated to a distance $\ell$ around $\sA$:
\begin{align}
    \tBE_\sA \approx \tBE_{\sA}^{(\sA_\ell)}.
\end{align}

Since we will be working with errors measured under the KMS inner product, it is natural to consider the corresponding KMS-induced superoperator norm; recall \autoref{defn:induced_norm}. As a trivial example, 
\begin{lem}\label{lem:norm_E}
    The KMS-induced superoperator norm of a spectral conditional expectation map (\autoref{defn:spectral_expectation}) which fixes $\vrho$ is $\norm{\BE}_{\vrho} = 1.$
\end{lem}

Next, we define the truncated Superoperator $\CK$ by truncating the Hamiltonian.
\begin{defn}
    [The Truncated Superoperator $\CK^{(\sR)}$]\label{defn:truncated-k} Consider a subset $\sR\subseteq [n]$ and let $\vH^\sR$ denote all the Hamiltonian terms within $\sR$. Then, for any single-site operator $a\in \CS_{\sR}^1$, the truncation of $\CK_a$ to the subset $\sR$ is the superoperator
    \begin{align}
        \CK_{a}^{(\sR)}[\vX] :=  [\vA^a,\vX] \vA_{\sR}^a(-i\beta/2)^{\dagger} - \vA_{\sR}^a(i\beta/2)^{\dagger}[\vA^a,\vX], \label{eq:CK_trunc}
    \end{align}
    \noindent with $\vA_{\sR}^a(z) = e^{i \vH^{\sR}z}\vA^a e^{-i \vH^{\sR}z}.$ For any subset $\sA\subseteq \sR $ and finite time $t\geq 0$, the truncated finite time conditional expectation is given by
    \begin{equation}
        \tBE_{\sA,t}^{(\sR)}[\cdot] := e^{t\CK_{\sA}^{(\sR)}}, \quad \text{where} \quad\CK_{\sA}^{(\sR)} :=\sum_{a\in \CS^1_\sA}\CK^{(\sR)}_{a}.  \label{eq:RAt}
    \end{equation}
    When the subscript $t$ is omitted, we refer to the infinite time limit. 
\end{defn}

The quasi-locality of $\tBE_{\sA}$ will be inherited from the quasi-locality of $\CK_\sA.$ 
\begin{lem}
    [A Priori Error Bounds on the Truncated Superoperator $\CK^{(\sR)}$]
    \label{lem:strictly_local} In the context of \autoref{defn:truncated-k}, there exists explicit constants $c_1, c_2>0$ as a function of $\beta, q$ such that for any subset $\sR\subseteq [n]$, and single-site operator $\vA^a\in \CS_{\sR}^1$ with support on a site $i\in [n]$:
    \begin{align}
 \norm{\CK_{a}^{(\sR)}-\CK_a}_{\vrho} \leq c_1\cdot e^{-c_2\cdot \ell \log \ell},
    \end{align}
    where the length-scale $\ell := \mathsf{dist}(i, [n]\setminus \sR)$.
\end{lem}
In words, the truncation error decays with the distance of the site $a$ to the boundary of the region $\sR$. The proof is a direct computation, with liberal use of locality of complex evolution (see \autoref{section:1D-locality}).

\begin{proof}
    The norm above is equivalent to the $2$-$2$ superoperator norm under a similarity transformation $\CK\rightarrow \vrho^{\frac{1}{4}}\CK[\vrho^{-\frac{1}{4}} \cdot \vrho^{-\frac{1}{4}}]\vrho^{\frac{1}{4}}$. 
    \begin{equation}
          \norm{\CK_{a}^{(\sR)}-\CK_a}_{\vrho}= \bigg\|\vrho^{\frac{1}{4}}\CK_{a}^{(\sR)}[\vrho^{-\frac{1}{4}} \cdot \vrho^{-\frac{1}{4}}]\vrho^{\frac{1}{4}}- \vrho^{\frac{1}{4}}\CK_{a}[\vrho^{-\frac{1}{4}} \cdot \vrho^{-\frac{1}{4}}]\vrho^{\frac{1}{4}}\bigg\|_{2-2}.\label{eq:2-2-norm}
    \end{equation}
    One can directly expand the commutators present in \eqref{eq:CK_trunc}:
    \begin{align}
        -\vrho^{\frac{1}{4}}\CK_{a}^{(\sR)}[\vrho^{-\frac{1}{4}} \vX \vrho^{-\frac{1}{4}}]\vrho^{\frac{1}{4}} &= \vrho^{\frac{1}{4}} \vA^a_{\sR}(i\beta/2)^{\dagger} \vA^a \vrho^{-\frac{1}{4}} \vX - \vrho^{\frac{1}{4}}  \vA^a_{\sR}(i\beta/2)^{\dagger} \vrho^{-\frac{1}{4}} \vX \vrho^{-\frac{1}{4}}\vA^a \vrho^{\frac{1}{4}}\\
        & - \vrho^{\frac{1}{4}} \vA^a \vrho^{-\frac{1}{4}}\vX \vrho^{-\frac{1}{4}} \vA^a_{\sR}(-i\beta/2)^{\dagger} \vrho^{\frac{1}{4}} + \vX \vrho^{-\frac{1}{4}} \vA^a \vA^a_{\sR}(-i\beta/2)^{\dagger} \vrho^{\frac{1}{4}}.
    \end{align}
    For conciseness, henceforth we drop the superscript $a$. We proceed by appling the triangle inequality to each term; which can be further simplified using Holder's inequality $\|\vL \vX\vR\|_2\leq \|\vL\|\cdot \|\vR\|\cdot \|\vX\|_2$:
    \begin{align}
        \eqref{eq:2-2-norm}&\le  \lnorm{ \vrho^{\frac{1}{4}} (\vA_{\sR}(i\beta/2)^{\dagger}- \vA(i\beta/2)^{\dagger}) \vrho^{-\frac{1}{4}} \cdot \vrho^{\frac{1}{4}}\vA \vrho^{-\frac{1}{4}}}  \\
        &+ \lnorm{\vrho^{\frac{1}{4}} ( \vA_{\sR}(i\beta/2)^{\dagger} - \vA(i\beta/2)^{\dagger})\vrho^{-\frac{1}{4}}}\cdot \norm{\vrho^{-\frac{1}{4}}\vA \vrho^{\frac{1}{4}}}\\
        &+\lnorm{\vrho^{\frac{1}{4}}\vA \vrho^{-\frac{1}{4}}} \cdot \norm{\vrho^{-\frac{1}{4}} (\vA_{\sR}(-i\beta/2)^{\dagger}- \vA(-i\beta/2)^{\dagger})\vrho^{\frac{1}{4}}}  \\ 
        &+ \lnorm{\vrho^{-\frac{1}{4}} \vA \vrho^{\frac{1}{4}}\cdot \vrho^{-\frac{1}{4}}(\vA_{\sR}(-i\beta/2)^{\dagger}-\vA(-i\beta/2)^{\dagger}) \vrho^{\frac{1}{4}}},\label{eq:2-2-norm-2}
    \end{align}
    and once again via Holder's inequality,
\begin{align}
   \eqref{eq:2-2-norm-2} &\le \bigg(\lnorm{\vrho^{\frac{1}{4}} ( \vA_{\sR}(i\beta/2)^{\dagger} - \vA(i\beta/2)^{\dagger})\vrho^{-\frac{1}{4}}} + \lnorm{\vrho^{-\frac{1}{4}} (\vA_{\sR}(-i\beta/2)^{\dagger}- \vA(-i\beta/2)^{\dagger})\vrho^{\frac{1}{4}}}\bigg) \\
&\quad \cdot   \bigg(\norm{\vrho^{\frac{1}{4}}\vA \vrho^{-\frac{1}{4}}} + \norm{\vrho^{-\frac{1}{4}}\vA \vrho^{\frac{1}{4}}}\bigg).
\end{align}
Next, we need to apply the 1D complex-dynamics bounds of~\autoref{lem:locality_complextime}. First, we simply remark that each $\vA$ is a single-site Pauli-operator, and thereby there exists an explicit positive constant $c_\beta>0$ (doubly exponential in $\beta$, see \eqref{eq:imaginary-time-conjugation}) s.t.
\begin{align}
        \norm{\vrho^{\frac{1}{4}}\vA \vrho^{-\frac{1}{4}}} + \norm{\vrho^{-\frac{1}{4}}\vA \vrho^{\frac{1}{4}}} \le c_\beta.
\end{align}
The first term is more subtle; we first truncate the differences to increasing ranges via an annulus decomposition, and then apply the norm bounds on imaginary time evolution \autoref{lem:locality_complextime} \eqref{eq:imaginary-time-conjugation} to dispense the conjugation by $\vrho^{1/4}$. Only then can one compute the error when comparing the $r, r+1$ ranges. To proceed, we may assume that $\sR$ is an interval of the form $[-\ell_\sL, \ell_\sR]$, with the site $i\in \sR$ placed at the origin at distance $\ell_\sL$ to the left-most boundary, and $\ell_\sR$ to the right-most. WLOG let $\ell_\sR\geq \ell_\sL$. Then, 
\begin{align}
    \lnorm{\vrho^{\frac{1}{4}} ( \vA_{\sR}(i\beta/2)^{\dagger} - \vA(i\beta/2)^{\dagger})\vrho^{-\frac{1}{4}}}& \leq \lnorm{\vrho^{\frac{1}{4}} ( \vA_{[-\ell_\sL, \ell_\sR]}(i\beta/2)^{\dagger} - \vA_{[-\ell_\sR, \ell_\sR]}(i\beta/2)^{\dagger})\vrho^{-\frac{1}{4}}} \\ &+\lnorm{\vrho^{\frac{1}{4}} ( \vA_{[-\ell_\sR, \ell_\sR]}(i\beta/2)^{\dagger} - \vA(i\beta/2)^{\dagger})\vrho^{-\frac{1}{4}}}.\label{eq:left-right-trunc}
\end{align}
By the annulus expansion, 
\begin{align}
       \lnorm{\vrho^{\frac{1}{4}} ( \vA_{[-\ell_\sR, \ell_\sR]}(i\beta/2)^{\dagger} - \vA(i\beta/2)^{\dagger})\vrho^{-\frac{1}{4}}} &\le  \sum_{r=\ell_\sR}^{\infty}\lnorm{\vrho^{\frac{1}{4}} ( \vA_{r}(i\beta/2)^{\dagger} - \vA_{r+1}(i\beta/2)^{\dagger})\vrho^{-\frac{1}{4}}}\\
        &\le c_\beta\cdot\sum_{r=\ell_\sR}^{\infty}\lnorm{ \vA_{r}(i\beta/2)^{\dagger} - \vA_{r+1}(i\beta/2)^{\dagger}} \cdot e^{2\beta (r+1)} \\
        &\leq c_\beta \cdot e^{4\beta} \cdot \sum_{r=\ell_\sR}^\infty \frac{\ln^r c_\beta }{r!} \cdot    e^{2\beta r} \\ 
        &\leq c_\beta e^{4\beta}\cdot \frac{(e^{2\beta}\ln c_\beta )^{\ell_\sR}}{\ell_\sR!}\cdot c_\beta^{ e^{2\beta}} \\
        &\leq a_\beta \cdot e^{-b_\beta \ell_\sR \log \ell_\sR},
    \end{align}
for a suitable choice of constants $a_\beta, b_\beta$. The first term in \eqref{eq:left-right-trunc} is analogous, up to super-exponential decay in $\ell_\sL$. With $\ell = \min(\ell_\sL, \ell_\sR) =  \mathsf{dist}(i, [n]\setminus \sR)-1$ then gives the desired bound. 

\end{proof}
\subsubsection{Proof of~\autoref{lem:localized_E}}
We are now in a position to prove \autoref{lem:localized_E}.
\begin{proof}

    [of \autoref{lem:localized_E}] 
The possibility of truncating the Hamiltonian hinges on the a priori spectral gap for $\CK_\sA$ and $\CK_{\sA}^{(\sR)}$ (\autoref{lem:local_gap_K}). Let us introduce a tunable cutoff time $t$. By the triangle inequality,
\begin{align}
     \|\tBE_\sA - \tBE_{\sA}^{(\sR)}\|_{\vrho} \leq  \|\tBE_\sA - \tBE_{\sA, t}\|_{\vrho} + \|\tBE_{\sA, t} - \tBE_{\sA, t}^{(\sR)}\|_{\vrho} +  \|\tBE_{\sA, t}^{(\sR)}- \tBE_{\sA}^{(\sR)}\|_{\vrho}.
\end{align}
The unconditional lower bound (say, $\lambda_{\sA}$) on the conditional spectral gap controls (both) finite-time approximation errors:
\begin{align}
    \lnorm{e^{t\CK_\sA } - \tBE_{\sA}}_{\vrho}  ,\quad  \lnorm{e^{t\CK_{\sA}^{(\sR)} } - \tBE_{\sA}^{(\sR)}}_{\vrho^{\sR}}  \leq  e^{-\lambda_{\sA}t}.
\end{align}
However, we note the error of the truncated map is weighted w.r.t. the Gibbs state of $\vH^{\sR}$. Nevertheless, by the two-sided comparison of measures \autoref{lem:compare_measure}, one can puncture $\vH$ at all the $\partial\sR$ endpoints of the subset $\sR$ to arrive at $\vH^{\sR}$:
\begin{align}
  \lnorm{e^{t\CK_{\sA}^{(\sR)} } - \tBE_{\sA}^{(\sR)}}_{\vrho}  \le c_\beta^{|\partial\sR|}\cdot \lnorm{e^{t\CK_{\sA}^{(\sR)} } - \tBE_{\sA}^{(\sR)}}_{\vrho^{\sR}}. \label{eq:comparison_super_errors}
\end{align}
for an explicit constant $c_\beta$. It remains only to compare the finite-time dynamics. For this purpose, we first require the following simple bound on the superoperator norm of $\BE_{\sA,t}^{(\sR)}$, analogous to \eqref{eq:comparison_super_errors}:
\begin{equation}
    \|\tBE_{\sA,t}^{(\sR)}\|_{\vrho} \leq c_\beta\cdot  \|\tBE_{\sA,t}^{(\sR)}\|_{\vrho^{\sR}} = c_\beta^{|\partial\sR|}.\label{eq:super-norm-trunc}
\end{equation}
Then, 
\begin{align}
    \lnorm{e^{ t\CK_\sA} - e^{ t\CK_{\sA}^{(\sR)}}}_{\vrho}&=\lnorm{ \int_{0}^{t} e^{\CK_\sA(t-s) }(\CK_{\sA}^{(\sR)}-\CK_\sA)e^{\CK_{\sA}^{(\sR)} s} \rd s}_{\vrho} \tag*{(Duhamel's principle)}\\
    &\le \int_{0}^{t} \lnorm{(\CK_{\sA}^{(\sR)}-\CK_\sA)e^{\CK_{\sA}^{(\sR)}s}}_{\vrho} \rd s \tag*{(\autoref{lem:norm_E})}\\
    &\le  \norm{\CK_{\sA}^{(\sR)}-\CK_{\sA}}_{\vrho}\cdot  \int_{0}^{t}\norm{e^{\CK_{\sA}^{(\sR)}s}}_{\vrho} \rd s \\
    &\le t\cdot  |\sA|\cdot  c_\beta^{|\partial\sR|} \cdot c_1 \cdot e^{-c_2 \ell \log \ell}, \tag*{(\autoref{lem:strictly_local}) + \eqref{eq:super-norm-trunc}}
\end{align}
for suitable constants $c_1, c_2$, and where $\ell = \mathsf{dist}(\sA, [n]\setminus \sR)$. Leveraging the exponential lower bound for $\lambda_\sA$ as guaranteed by \autoref{lem:local_gap_K}, one can make the explicit choice of  $t = \exp[\alpha_{\beta, q} |\sA|]\cdot (1+ \gamma_{\beta, q}\ell \log \ell)$ (for suitably chosen constants) to arrive at the advertised bound. 
\end{proof}

\section{Decay of Correlations}
\label{section:clustering}
Proofs of Markov chain spectral gaps often hinge on mindful definitions of correlation decay. Perhaps due to the challenges involving conditioning, even for commuting Hamiltonians, the search for suitable definitions of correlation decay has been an active area in the study of quantum Gibbs sampling. Broadly, there are two main classes of quantum correlation decay that somewhat parallel the classical definitions. The first class is \textit{weak} clustering of correlations, which intends to capture correlations involving two observables that have distinct, far-apart supports. These include and generalize the physically motivated two-point thermal correlation functions, which are common targets in thermal calculations, can be directly measured experimentally, and entail transport and thermodynamic properties of the system. 

At a technical level, there is a zoo of plausible definitions of weak clustering that one may write down, partly due to the non-commutativity of the quantum Gibbs state. For our purposes, the most natural version is \textsf{KMS Weak Clustering}, which is defined over the same KMS inner product as the Lindbladian.

\begin{defn}
    [\textsf{KMS Weak Clustering}]\label{defn:kms_clustering} We say a Gibbs state $\vrho$ of a 1D Hamiltonian satisfies \emph{\textsf{KMS Weak Clustering}} with error $\epsilon_\mathsf{KMS}(\cdot)$, if for every pair of operators $\vX, \vY$ with support on intervals $\sA, \sC\subseteq [n]:$
 \begin{equation}\label{eq:kms-weak-clustering}
    \mathsf{Cov}(\vX, \vY):= \bigg|\braket{\vX, \vY}_{\vrho} - \tr[\vrho \vX] \cdot \tr[\vrho \vY]\bigg| \leq \|\vX\|_{\vrho}\cdot  \|\vY\|_{\vrho} \cdot \epsilon_\mathsf{KMS}(\mathsf{dist}(\sA, \sC)).
\end{equation}
\end{defn}

Since we often will require modifying the topology of the observables, we extend the notion of \textsf{KMS Weak Clustering} to cases when the supports of the observables $\vX, \vY$ lie on sequences of intervals. Whenever required, we make the choice of topology explicit.  
The distinction between the above and other standard definitions of weak clustering of correlations in the literature lies in 
\begin{enumerate}
    \item how to quantify the inner product $\braket{\cdot, \cdot}$ between a pair of observables, and 
    \item the normalization of the observables in the error term.
\end{enumerate}
Recently, \cite{Kimura_2025} proved that a related version of the clustering condition above holds unconditionally in 1D systems -- but there the error was measured in operator norm, and the inner product was instead GNS ($\tr[\vX\vY\vrho]$).  In \autoref{app:Weakclustering}, we present a self-contained exposition on other such notions of clustering of correlations studied in the literature, as well as a proof that \cite{Kimura_2025}'s result implies \autoref{defn:kms_clustering} holds at all temperatures in 1D systems with sub-exponential decay. This step is rather laborious, but is primarily based on existing tools.

\begin{thm}
    [\textsf{KMS Weak Clustering} in 1D at all Finite Temperatures]\label{thm:kms-clustering-uncon}
    For any inverse-temperature $\beta\in \BR^+$ and  local dimension $2^q$, there exists constants $\alpha(\beta, q), \eta(\beta, q)>0$ such that the Gibbs state of every 1D Hamiltonian (\autoref{defn:1D}) admits \emph{\textsf{KMS Weak Clustering}} (\autoref{defn:kms_clustering}) with error function:
    \begin{equation}\label{eq:kms-subexp}
        \epsilon_{\mathsf{KMS}}(\ell) \leq \alpha\cdot \exp\big(-\ell^{\eta}\big).
    \end{equation}
\end{thm}

We note that although the theorem above is stated for observables over intervals, if every 1D Hamiltonian satisfies \textsf{KMS Weak Clustering} over intervals, then clustering also holds over finite $k-$sequences of intervals \eqref{eq:many-intervals} by ``folding'' the chain over itself.

\begin{rmk}[Observable Support Topology]\label{rmk:folding}
      Given a topology of the form $\sA_1\sB_1\sC\sB_2\sA_2$ with observables supported on $\sA = \sA_1\cup \sA_2$ and $\sC$, one can define another 1D Hamiltonian where the number of qubits and interaction strength per-site has doubled, but the observables are now correctly supported on intervals.\footnote{We thank Tomotaka Kuwahara for this observation and further correspondence related to \cite{Kimura_2025}.}
    \end{rmk}

Unfortunately, for the purpose of proving spectral gaps, weak clustering in any form appears insufficient. Instead, one requires the related but more stringent notion of \textit{strong clustering}, which involves the appropriate application of a conditional expectation map \cite{kastoryano2016commuting}. To study the spectral gap of KMS-detailed-balanced Lindbladians, the most relevant formulation is again defined under the KMS inner product.

\begin{defn}
    [Conditional Covariance]\label{defn:conditional_covariance} Let $\{\tBE_\sR\}_{\sR\subseteq [n]}$ be the family of spectral conditional expectations as in \autoref{lem:explicit_CE}, with respect to a Gibbs state $\vrho$. For any subset $\sR\subseteq [n]$, the conditional covariance of a pair of operators $\vX, \vY$ with respect to $\tBE_\sR$ is defined as:
    \begin{equation}\label{eq:cond_cov}
        \mathsf{Cov}_\sR(\vX, \vY) :=\braket{(1-\tBE_\sR)[\vX],(1-\tBE_\sR)[\vY] }_{\vrho}.
    \end{equation}
\end{defn}

We remark that when $\sR = [n]$ we recover the covariance in \eqref{eq:kms-weak-clustering}. Equipped with the above definition of conditional covariance, we can now introduce strong clustering.

\begin{defn}
    [\textsf{KMS Strong Clustering}, based on \cite{kastoryano2016commuting}]\label{defn:strong_clustering} Let $\{\tBE_\sR\}_{\sR\subseteq [n]}$ be the family of spectral conditional expectations as in \autoref{lem:explicit_CE}, with respect to a Gibbs state $\vrho$. Then, $(\vrho, \{\tBE_\sR\})$ is said to satisfy \emph{\textsf{KMS Strong Clustering}} with error $\epsilon_\mathsf{S}(\cdot)$ if, for every observable $\vO$ and every consecutive disjoint intervals $\sA, \sB, \sC\subseteq [n]:$
    \begin{equation}\label{eq:strong_clustering}
         \mathsf{Cov}_{\sA\sB\sC}(\tBE_{\sA\sB}[\vO], \tBE_{\sB\sC}[\vO]) \leq \|\vO\|^2_{\vrho}\cdot \epsilon_\mathsf{S}(\mathsf{dist}(\sA, \sC)).
    \end{equation}
\end{defn}

Unfortunately, the use of conditional expectations makes strong clustering difficult to interpret, and challenging to derive from the more physical notion of weak clustering. The only previous success story in bridging the gap from weak to strong clustering is in 1D commuting Hamiltonians, where \cite{kastoryano2016commuting} proved that (KMS) weak clustering implies strong clustering, and where the choice of conditional expectation was defined based on commuting Gibbs samplers such as Davies' generator.\\

\noindent \textbf{Strong from Weak Clustering in 1D.} The main result of this section is to build upon our explicit spectral conditional expectation to define a workable strong clustering condition, and prove a direct reduction to \textsf{KMS Weak Clustering} in 1D systems. 
\begin{thm}
    [\textsf{KMS Strong Clustering} from \textsf{KMS Weak Clustering}]\label{thm:strong_from_weak} 
    Consider the family of spectral conditional expectations $\{\tBE_\sR\}_{\sR\subseteq [n]}$ w.r.t. the Gibbs state $\vrho$ of a 1D Hamiltonian (\autoref{defn:1D}). Assume $\vrho$ satisfies \emph{\textsf{KMS Weak Clustering}} with error function $\epsilon_{\mathsf{KMS}}$, over any sequence of intervals with topology $\sA_1\sB_1\sC\sB_2\sA_2$.\footnote{The observables lie on $\sA_1\cup\sA_2$ and $\sC$, and $\mathsf{dist}(\sA, \sC)=1+\min(|\sB_1|, |\sB_2|)$. That is, a 2-sequence \eqref{eq:many-intervals}.}

    Then, there exist explicit constants $c_1, c_2, c_3, c_4>0$ as a function of $\beta, q$ such that  $(\vrho, \{\BE_\sR\}_\sR)$ admits \emph{\textsf{KMS Strong Clustering}} with error 
    \begin{equation}
        \epsilon_\mathsf{S}(\ell) \leq c_1\cdot \sqrt{\epsilon_{\mathsf{KMS}}(c_2\cdot \ell)} + e^{-c_3\cdot \ell\log \ell}
    \end{equation}
over every tripartition of consecutive disjoint intervals $\sA, \sB,\sC\subseteq [n]$ such that $\frac{1}{2}|\sA|, \frac{1}{2}|\sC| \geq |\sB|\geq c_4\cdot |\sA\sB\sC| / \log |\sA\sB\sC|$. 
\end{thm}

As an immediate corollary of the above and \autoref{thm:kms-clustering-uncon}, we conclude that some form of strong clustering holds in 1D quantum systems at all temperatures. Technically, our strong clustering statement only applies if the sub-interval $\sB$ is sufficiently large relative to the size of the other sub-intervals, due to the quasi-locality of our spectral conditional expectation. Nevertheless, as we shall see in the next section, this limited form of strong clustering still enables the proof of a conditional spectral gap for $\CK$ -- improving over the a priori bound in \autoref{lem:local_gap_K}.

While not obvious at the moment, \textit{after} we have obtained a constant spectral gap, we will be able to return to this discussion to obtain a truly exponential weak clustering of correlations in all non-translational invariant systems (see Section \ref{sec:results}).

\subsection{Strong Clustering from KMS Weak Clustering}

The remainder of the section presents the proof of \autoref{thm:strong_from_weak}. 
    Since our spectral conditional expectation family $\tBE_{\sR}$ is projective $\tBE_{\sR}=\tBE_{\sR}^2$, the covariance in \eqref{eq:cond_cov} simplifies as
    \begin{equation}\label{eq:gluing-to-clustering}
        \mathsf{Cov}_{\sA\sB\sC}(\tBE_{\sA\sB}[\vO], \tBE_{\sB\sC}[\vO]) = \braket{\vO, (\tBE_{\sA\sB}\tBE_{\sB\sC}-\tBE_{\sA\sB\sC})[\vO])}_{\vrho}.
    \end{equation}
    The proof of strong clustering, then, boils down to proving a ``gluing'' identity for spectral conditional expectations
    \begin{align}
        \tBE_{\sA\sB}\tBE_{\sB\sC} \approx \tBE_{\sA\sB\sC},\label{eq:gluing-strong-clustering}
    \end{align}
    in the KMS superoperator norm (recall \autoref{fig:EAB_BC_ABC}). The proof of \autoref{thm:strong_from_weak} consists of two lemmas. 
    \begin{enumerate}
        \item \textit{Conditional Expectations Trivialize an Island} (\autoref{lem:island}). We already know that the effect of the conditional expectation map $\tBE_{\sX}$ is to trivialize operators supported on $\sX\subseteq [n]$; but it may produce some junk on the boundary of $\sX$ (recall \autoref{lem:explicit_CE}). \autoref{lem:island} shows that if the operator is supported deep in the interior of $\sX$, then weak clustering actually implies that $\BE_{\sX}$ fully trivializes the operator ``island'', without affecting/leaving junk on the boundary (see \autoref{fig:trivialize}). 

        \item \textit{A 3-stack Gluing Lemma} (\autoref{lem:gluing_2}). As it turns out,  \autoref{lem:island} naturally implies a ``three-stack'' version of the desired gluing expression in \eqref{eq:gluing-strong-clustering}. Roughly speaking, the spectral conditional expectation on $\sA\sB\sC$ can be approximated by that on $\sA, \sC$ and subsequently a slightly enlarged region around $\sB$: 
        \begin{equation}
            \tBE_{\sA\sB\sC}  \approx \tBE_{\sA_\sR\sB\sC_\sL} \tBE_{\sA}  \tBE_{\sC}. \quad \text{(\autoref{fig:3stack})} \label{eq:tripartite-gluing}
        \end{equation}
        The derived 3-stack Gluing Lemma then implies the bipartite statement \eqref{eq:gluing-strong-clustering}, namely strong clustering \eqref{eq:gluing-to-clustering}. 
    \end{enumerate}

    Over the course of the above two steps, which we prove as follows, we consider a decomposition of the full 1D chain, partitioned into consecutive disjoint intervals as:
    \begin{equation}\label{eq:chain_decomp_island}
        \sL\cup \sA_\sL \cup \sA_\sR \cup \sB\cup \sC_\sL \cup \sC_\sR \cup \sR = [n],
    \end{equation}
    where we shorthand $\sA = \sA_\sR\cup\sA_\sL$ and analogously for $\sC$. 

\begin{figure}[ht]
\includegraphics[width=0.8\textwidth]{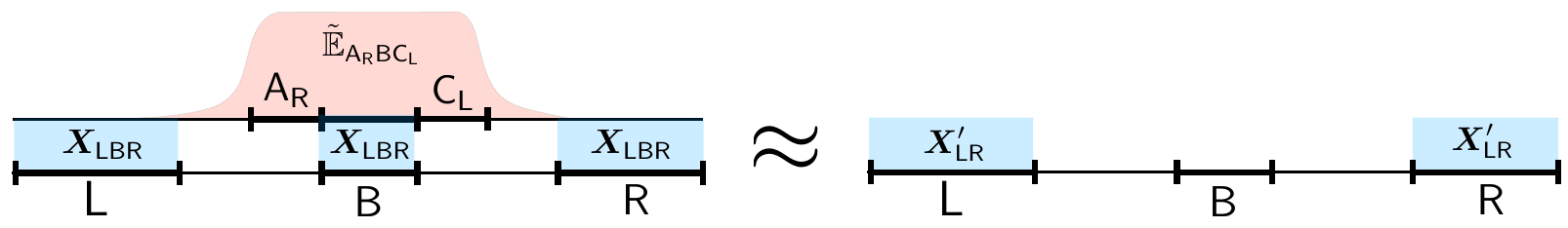}
\caption{\textit{Conditional Expectations Trivialize Islands} (\autoref{lem:island}).
 The coveted gluing identities reduce cleanly to the following property of spectral conditional expectations. If an operator has a nontrivial ``island'' sandwiched by identities (but can be arbitrary afar), then acting by a (slightly fattened) spectral conditional expectation fully trivializes the island.}
\label{fig:trivialize}
\end{figure}

\subsubsection{Proof of \autoref{lem:island}: trivializing an island}
    \begin{lem}[The Conditional Expectation Trivializes an Island]\label{lem:island} 
    Under the assumptions of \autoref{thm:strong_from_weak}, for any operator $\vX_{\sL\sB\sR}$ trivial on $\sA,\sC,$ there is an  operator $\vX'_{\sL\sR}$ supported only on $\sL, \sR$ such that
    \begin{align}
        \lnorm{\tBE_{\sA_\sR\sB\sC_\sL}\vX_{\sL\sB\sR} -   \vX'_{\sL\sR}}_{\vrho} \le \bigg(c_1\cdot \epsilon_{\mathsf{KMS}}(\min(|\sA_\sR|, |\sC_\sL|))^{1/2} + \exp\big( c_2\cdot |\sA\sB\sC| - c_3\cdot \ell\log \ell \big)\bigg) \|\vX_{\sL\sB\sR}\|_{\vrho},
    \end{align}
    for constants $c_1,c_2,c_3$ that depend only on $\beta, q$, and $\ell = \min(|\sA_{\sL}|, |\sC_{\sR}|)$.
\end{lem}

The argument is partially captured in the explanation~\eqref{eq:kills_island} when the outer regions $\mathsf{L}$ and $\mathsf{R}$ are empty; here, the main technical step is to handle possible correlation with them. Consider the factorized (or \textit{punctured}) Hamiltonian, arising from dropping the Hamiltonian links connecting $\sL,\sA$ and the term connecting $\sC,\sR$:
\begin{align}
    \vH_\mathsf{punct} = \vH_{\sL}+ \vH_{\sA\sB\sC} + \vH_{\sR},
\end{align}
and the resulting factorized Gibbs state
\begin{align}
    \vrho_\mathsf{punct} = \vrho_\sL \otimes \vrho_{\sA\sB\sC}\otimes \vrho_\sR\quad \text{and let}\quad \vrho_{\sL\sR} := \vrho_\sL\otimes \vrho_\sR.
\end{align}
This allows us to consider a factorized orthonormal basis:
\begin{align}
\text{ for } \sI\subseteq \sA\sB\sC, \quad \braket{\vB^i_{\sI},\vB^j_{\sI}}_{\vrho_{\sA\sB\sC}} = \delta_{ij},\quad \braket{\vB^i_{\mathsf{out}},\vB^j_{\mathsf{out}}}_{\vrho_{\sL\sR}} = \delta_{ij}.
\end{align}
In our use, it will be convenient to extend the inner product $\braket{\vB^i_{\sI},\vY}_{\vrho_{\sA\sB\sC}}$ for operators $\vY$ defined on the full system via a partial trace by linearity.

\begin{proof}

[of \autoref{lem:island}] Let us assume that $\norm{\vX_{\sL\sB\sR}}_{\vrho}=1$. \\

\noindent \textbf{Step 1: Truncate $\tBE$.} To begin, by \autoref{lem:localized_E} we may keep only the Hamiltonian terms close to $\sA_\sR\sB\sC_\sL$, $\vH\rightarrow \vH_{\mathsf{ABC}}$, for the spectral conditional expectation. For any operator $\vY,$ 
\begin{align}
    \lnorm{ \tBE_{\sA_\sR\sB\sC_\sL}[\vY] -  \tBE^{(\mathsf{ABC})}_{\sA_\sR\sB\sC_\sL}[\vY]}_{\vrho} &\leq \exp\L(c_2\cdot  |\sA\sB\sC| - c_3\ell\log \ell\R)\cdot \|\vY\|_{\vrho}.\label{eq:quasi-local-gluing}
\end{align}
where $\ell = \min(|\sA_{\sL}|, |\sC_{\sR}|)$ since $\sA_{\sL}\sA_\sR\sB\sC_\sL\sC_{\sR} =\sA\sB\sC$; $c_2, c_3$ are suitable constants dependent only on $\beta, q$.\\

\noindent \textbf{Step 2: Decouple $\sL, \sR$.} Next, we consider a decomposition for $\vX_{\sL\sB\sR}$ with respect to the tensored KMS-inner product spaces
\begin{align}
    \braket{,}_{\vrho_{\sL\sR}}\times \braket{,}_{\vrho_{\sA\sB\sC}}.
\end{align}
We isolate the identity component in $\sB$, which is the $\vX'_{\sL\sR}$ term, and perform a Schmidt decomposition for the remainder, with respect to $\braket{,}_{\vrho_{\sL\sR}}$ and the zero-mean subspace of $\braket{,}_{\vrho_{\sA\sB\sC}}$.
\begin{align}
    &\vX_{\sL\sB\sR} = \under{\tr_{\sA\sB\sC}[\vrho_{\sA\sB\sC}\vX_{\sL\sB\sR}] \vI_{\sA\sB\sC}}{ =:\vX'_{\sL\sR}}  +\sum_{j=2} \alpha_{j} \vW_{\sL\sR}^{j}\otimes \undersetbrace{\text{zero mean}}{\vW_{\sB,\sA\sB\sC}^{j}}, 
\end{align}
where for each $i, j\geq 2$, the components satisfy the following normalization conditions:
\begin{equation}
    \braket{\vW_{\sB,\sA\sB\sC}^{i}, \vW_{\sB,\sA\sB\sC}^{j}}_{\vrho_{\sA\sB\sC}}= \braket{\vW_{\sL\sR}^{i}, \vW_{\sL\sR}^{j}}_{\vrho_{\sL\sR}}=\delta_{i, j} \quad \text{ and } \quad \tr[\vW_{\sB,\sA\sB\sC}^{j}\vrho_{\sA\sB\sC}]=0\label{eq:orthonormalW}
\end{equation}
and the coefficients are normalized s.t. $\sum_j |\alpha_j|^2 \leq \norm{\vX_{\sL\sB\sR}}_{\vrho_\mathsf{punct}}^2 \leq c_{\beta, q}\norm{\vX_{\sL\sB\sR}}_{\vrho}^2$, with $c_{\beta, q}$ the constant from the comparison of measures (\autoref{lem:compare_measure}). Applying conditional expectation truncated to $\sA\sB\sC$ leaves the $\sL\sR$ parts invariant:
    \begin{align}
       \tBE^{(\mathsf{ABC})}_{\sA_\sR\sB\sC_\sL}[\vX_{\sL\sB\sR}] &=  \vX'_{\sL\sR} +\sum_{j=2} \alpha_{j} \vW_{\sL\sR}^{j}\otimes  \tBE^{(\mathsf{ABC})}_{\sA_\sR\sB\sC_\sL}[\vW_{\sB,\sA\sB\sC}^{j}].
    \end{align}

\noindent \textbf{Step 3: Leveraging Weak Clustering.} One readily identifies that the second term above implicitly contains KMS-correlation functions. Indeed, its norm can be written as
\begin{align}
    \bigg\|\sum_{j=2} \alpha_{j} \vW_{\sL\sR}^{j}\otimes  \tBE^{(\mathsf{ABC})}_{\sA_\sR\sB\sC_\sL}[\vW_{\sB,\sA\sB\sC}^{j}]\bigg\|_{\vrho_\mathsf{out}\otimes \vrho_{\sA\sB\sC}}^2  &= \sum_j |\alpha_j|^2 \cdot \big\|\tBE^{(\mathsf{ABC})}_{\sA_\sR\sB\sC_\sL}[\vW_{\sB,\sA\sB\sC}^{j}]\big\|_{\vrho_{\sA\sB\sC}}^2 \tag*{(By~\eqref{eq:orthonormalW})}\\
    & \leq c_1 \sum_j |\alpha_j|^2 \cdot \big\|\tBE^{(\mathsf{ABC})}_{\sA_\sR\sB\sC_\sL}[\vW_{\sB,\sA\sB\sC}^{j}]\big\|_{\vrho}^2 \tag*{(\autoref{lem:compare_measure})}\\
    &\approx c_1 \sum_j |\alpha_j|^2 \cdot \big\|\tBE_{\sA_\sR\sB\sC_\sL}[\vW_{\sB,\sA\sB\sC}^{j}]\big\|_{\vrho}^2 \tag*{(from \eqref{eq:quasi-local-gluing})}\\
    &= c_1 \sum_j |\alpha_j|^2 \cdot \braket{\vW_{\sB,\sA\sB\sC}^{j}, \tBE_{\sA_\sR\sB\sC_\sL}[\vW_{\sB,\sA\sB\sC}^{j}]}_{\vrho} \tag*{(KMS-DB of $\tBE$)} \\
    &\leq c_2\cdot \norm{\vX_{\sL\sB\sR}}_{\vrho}^2 \cdot \epsilon_{\mathsf{KMS}, 2q, 2\beta}(\min(|\sA_\sR|, |\sC_\sL|)),
\end{align} 
where, in sequence, we used the orthonormality of the Schmidt decomposition, the comparison of measures \autoref{lem:compare_measure}, the quasi-locality of $\tBE$, and detailed balance. Finally, in the last line above, we invoked the assumption of \textsf{KMS Weak Clustering}:
\begin{enumerate}
    \item For each $j$, the operator $\vW_{\sB,\sA\sB\sC}^{j}$ is supported on $\sB$, has $\vrho_{\sA\sB\sC}-$KMS norm 1, and thus $\vrho-$norm $c_{\beta, q}$ (\autoref{lem:compare_measure})
    \item $\tBE^{(\mathsf{ABC})}_{\sA_\sR\sB\sC_\sL}[\vW_{\sB,\sA\sB\sC}^{j}]$ is supported on $\sA_\sL, \sC_\sR$ and similarly has $\vrho-$norm $c_{\beta, q}$ (\autoref{lem:compare_measure}).
\end{enumerate}
Thereby, the distance between the operator supports is $\min(|\sA_\sR|, |\sC_\sL|)$. We note that the support of the second operator above does not lie on an interval of the 1D chain, but instead on a separated pair of intervals with pattern $\mathsf{A}_1\mathsf{B}_1\mathsf{C}_2\mathsf{B}_2\mathsf{A}_3.$. \\

\noindent \textbf{Altogether.} We apply the comparison of measures (\autoref{lem:compare_measure}), to relate the error above under the $\vrho_{\sA\sB\sC}-$norm back to the $\vrho-$ norm up to a $\beta, q$ dependent constant. 
\begin{align}
    \bigg\|\tBE_{\sA_\sR\sB\sC_\sL}\vX_{\sL\sB\sR} - \tr_{\sA\sB\sC}[\vrho_{\sA\sB\sC}\vX_{\sL\sB\sR}]\bigg\|_{\vrho}  \leq  c_2'\cdot \epsilon_{\mathsf{KMS}, 2q, 2\beta}(\min(|\sA_\sR|, |\sC_\sL|))^{1/2} \quad +  \text{(truncation error \eqref{eq:quasi-local-gluing})},
\end{align}
which gives the advertised bound. 
\end{proof}

\subsubsection{Obtaining the 3-stack gluing lemma}

Now, we show that \autoref{lem:island} implies the following 3-stack Gluing Lemma:

\begin{lem}
    [A 3-stack Gluing Lemma]\label{lem:gluing_2} Under the assumptions of \autoref{thm:strong_from_weak}, the conditional expectation $\tBE_{\sA\sB\sC}$ admits the following decomposition:
    \begin{align}
        \| \tBE_{\sA\sB\sC}  - \tBE_{\sA_\sR\sB\sC_\sL} \tBE_{\sA} \tBE_{\sC}\|_{\vrho} \lesssim c_1\cdot \epsilon_{\mathsf{KMS}}(\min(|\sA_\sR|, |\sC_\sL|))^{1/2} + \exp\big( c_2\cdot |\sA\sB\sC| - c_3\cdot \ell\log \ell \big),
    \end{align}
    for constants $c_1, c_2, c_3$ that depend only on $\beta, q$, and $\ell = \min(|\sA_{\sL}|, |\sC_{\sR}|, |\sB|)$.
\end{lem}

\begin{figure}[ht]
\includegraphics[width=0.7\textwidth]{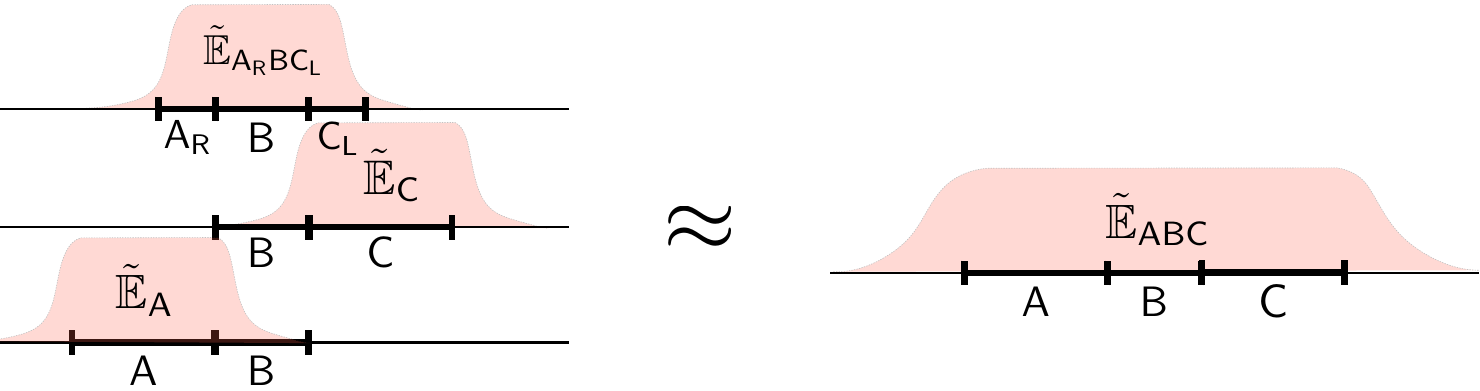}
\caption{\textit{The 3-stack Gluing Lemma} (\autoref{lem:gluing_2}). The spectral conditional expectation on $\sA\sB\sC$ can be approximated by that on $\sA, \sC$ and subsequently a slightly fattened region $\sB$: $\tBE_{\sA\sB\sC}  \approx \tBE_{\sA_\sR\sB\sC_\sL} \tBE_{\sA} \tBE_{\sC}$.}
\label{fig:3stack}
\end{figure}

\begin{proof}

    [of \autoref{lem:gluing_2}] As a consequence of \autoref{lem:island}, we have
    \begin{align}
        \tBE_{\sA_\sR\sB\sC_\sL}\tBE_{\sA\sC} [\vX] = \tBE_{\sA_\sR\sB\sC_\sL}[\vX_{\sL\sB\sR}] \approx  \vX_{\sL\sR} = \tBE_{\sA\sB\sC}[\vX_{\sL\sR}] \approx \tBE_{\sA\sB\sC} \tBE_{\sA_\sR\sB\sC_\sL}\tBE_{\sA\sC}[\vX] =  \tBE_{\sA\sB\sC}[\vX].
    \end{align}
    By combining the above with the quasi-locality of $\tBE_{\sA\sC}$ in \autoref{lem:localized_E},
    \begin{equation}
        \|\tBE_{\sA\sC} - \tBE_\sA\tBE_\sC\|_{\vrho} \leq \exp\bigg( c_2\cdot |\sA\sC| - c_3\cdot |\sB|\log |\sB|\bigg),
    \end{equation}
    where again the constants depend only on $\beta, q$. 
\end{proof}

Now, we can obtain the 2-stack gluing as required for~\autoref{thm:strong_from_weak}, by applying the derived three-stack lemma twice, under a very careful choice of interval partition (\autoref{fig:2stack-3stack}).

\begin{figure}[ht]
\includegraphics[width=0.9\textwidth]{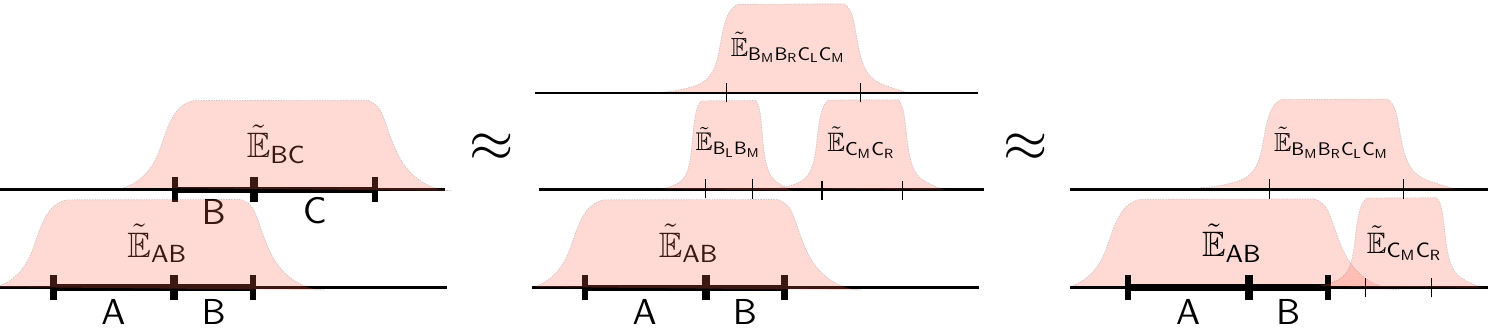}
\caption{\textit{Bipartite Gluing from 3-stack gluing}. The desired strong clustering statement in \eqref{eq:gluing-strong-clustering} follows by applying the derived 3-stack Gluing Lemma \eqref{eq:tripartite-gluing} twice, on a partition of the regions $\sB, \sC$, see \eqref{eq:bc-partition-gluing}. To compress the figure, we displayed the pairs of maps $\tBE_{\mathsf{B_LB_M}}$, $\tBE_{\mathsf{C_MC_R}}$, and $\tBE_{\mathsf{AB}}$, $\tBE_{\mathsf{C_MC_R}}$ in the same row. }
\label{fig:2stack-3stack}
\end{figure}

\begin{proof}

    [of the bipartite gluing lemma, \autoref{thm:strong_from_weak}] We begin by applying \autoref{lem:gluing_2} to the conditional expectation $\tBE_{\sB\sC}$. Consider the following partition of the interval $\sB\cup \sC$:
    \begin{equation}
       \under{ \sB_\sL\cup \sB_\mathsf{M}}{\sA'}\cup \under{\sB_\sR \cup \sC_\sL}{\sB'}\cup \under{\sC_\mathsf{M}\cup \sC_\sR}{\sC'}.\label{eq:bc-partition-gluing}
    \end{equation}
    The 3-stack Gluing Lemma in \autoref{lem:gluing_2} then gives the following approximate decomposition:
\begin{align}
      \big\|\tBE_{\sB\sC} -  
    \tBE_{\sB_\sL\sB_\mathsf{M}}
    \tBE_{ \sC_{\mathsf{M}}\sC_{\mathsf{R}}}\tBE_{\sB_\mathsf{M}\sB_\sR\sC_{\sL}\sC_\sM} \big\|_{\vrho} \lesssim &c_1\cdot \epsilon_{\mathsf{KMS}}(\min(|\sB_\sM|, |\sC_\sM|))^{1/2} + \exp\big( c_2\cdot |\sB\sC| - c_3\cdot |\ell|\log |\ell|\big), \label{eq:patchBC}
\end{align}
with $\ell = \min(|\sB_\sL|, |\sC_\sR|, |\sB_\sR\sC_\sL|)$.
We then have 
\begin{align}
    \|\tBE_{\sA\sB} \tBE_{\sB\sC} - \tBE_{\sA\sB\sC}\|_{\vrho} &\leq \lnorm{\tBE_{\sA\sB}\big(\tBE_{\sB_\sL\sB_\mathsf{M}}
    \tBE_{ \sC_{\mathsf{M}}\sC_{\mathsf{R}}}\tBE_{\sB_\mathsf{M}\sB_\sR\sC_{\sL}\sC_\sM}  \big) - \tBE_{\sA\sB\sC}}_{\vrho} + \big\|\tBE_{\sA\sB}\big(\tBE_{\sB\sC} -  
    \tBE_{\sB_\sL\sB_\mathsf{M}}
    \tBE_{ \sC_{\mathsf{M}}\sC_{\mathsf{R}}}\tBE_{\sB_\mathsf{M}\sB_\sR\sC_{\sL}\sC_\sM}\big) \big\|_{\vrho} 
    \\ &\leq\lnorm{\tBE_{\sA\sB}
    \tBE_{ \sC_{\mathsf{M}}\sC_{\mathsf{R}}}\tBE_{\sB_\mathsf{M}\sB_\sR\sC_{\sL}\sC_\sM}  - \tBE_{\sA\sB\sC}}_{\vrho}  + \big\|\tBE_{\sB\sC} -  
    \tBE_{\sB_\sL\sB_\mathsf{M}}
    \tBE_{ \sC_{\mathsf{M}}\sC_{\mathsf{R}}}\tBE_{\sB_\mathsf{M}\sB_\sR\sC_{\sL}\sC_\sM} \big\|_{\vrho} \\
    &\lesssim c_1\epsilon_\mathsf{KMS}(\min(|\sB_\sM|, |\sC_\sM|))^{1/2}  + \exp\big( c_2\cdot |\sA\sB\sC| - c_3\cdot \ell\log \ell\big) \label{eq:patchBC_triangle}
\end{align}
where, on the second term we used the fact that $\tBE_{\sA\sB}$ has spectra in $[0, 1]$ and then \eqref{eq:patchBC}; and to the first term we use $\tBE_{\sA\sB}\tBE_{\sB_\sL\sB_\mathsf{M}}  = \tBE_{\sA\sB}$ and then \autoref{lem:gluing_2} again onto regions $\sA\sB:\sB_\mathsf{M}\sB_\sR\sC_{\sL}\sC_\sM :\sC_{\mathsf{M}}\sC_{\mathsf{R}}$. The resulting error is dominated by that of the second term. 

To conclude the proof, we make the explicit choice that the regions $\sB, \sC$ are divided into thirds: $|\sB_\sL| , |\sB_\sM|, |\sB_R|\geq \floor{|\sB|/3}$, $|\sC_\sL| , |\sC_\sM| , |\sC_R|\geq \floor{|\sC|/3}$. The assumptions on the size of $\sB$ relative to $\sA, \sC$ and $\sA\sB\sC$ then give the claimed bound.

\end{proof}

\newpage
\section{Strong Clustering Implies a Constant Gap for Generator \texorpdfstring{$\CK$}{K}}
\label{sec:conditionalgapK}

A long-standing paradigm in the theory of Gibbs samplers is that \textit{static} correlations in the Gibbs measure should be intimately related to the \textit{dynamic} mixing time of the Markov chain. As discussed, our main strategy to prove the spectral gap is to exploit a suitable notion of strong clustering and recursive arguments~\cite{kastoryano2016commuting}. While our ultimate goal is the spectral gap of $\CL$, however, this section focuses solely on the spectral gap of the auxiliary generator $\CK$, as it appears to be a natural, mathematically self-contained intermediate step. Here, we show that the version of strong clustering (\autoref{defn:strong_clustering}) derived for 1D Hamiltonians at finite temperatures in~\autoref{thm:strong_from_weak}, implies that the generator $\CK$ \eqref{eq:CK} admits a constant spectral gap.
\begin{thm}
    [The Spectral Gap of the Generator $\CK$]\label{thm:K-is-gapped}    
    Consider a 1D Hamiltonian (\autoref{defn:1D}) on $n$ qudits with local dimension $2^q$ and inverse-temperature $\beta>0.$ Let $\CK$ be the generator of the spectral conditional expectation~\eqref{eq:CK} defined by the set $\CS^1_{[n]}$ of all single-site Pauli jumps, which is detailed-balanced for the Gibbs state $\vrho.$ Then, the spectral gap of $\CK$ is at least
    \begin{equation}
        \lambda(\CK):=\inf_{\vO} \frac{\braket{\vO, -\CK[\vO]}_{\vrho}}{\|\vO - \vI \cdot \tr[\vrho\vO]\|_{\vrho}^2} \geq \gamma_{\beta,q} >0.
    \end{equation}
    for a constant $\gamma_{\beta,q}$ depending only on $\beta,q.$
\end{thm}
While we will soon obtain spectral gaps for $\CL$ (\autoref{sec:CKGgap}), one may ask, what does the spectral gap of $\CK$ even physically mean? Even in the absence of complete positivity, $\CK$ is, for spectral purposes, a valid generator of conditional expectations. Therefore, by vectorizing $\CK$ into two overlapping copies of the quantum spin chain, one can regard $-\CK$ as a quasi-local Hamiltonian acting on $2n$ many qudits, whose gapped ground state is the purified Gibbs state $\ket{\sqrt{\vrho}}.$ By appealing to standard adiabatic arguments for gapped ground states, $\CK$ gives a totally valid adiabatic path for preparing $\ket{\sqrt{\vrho}}$ -- see \autoref{sec:results} and \autoref{sec:adiabatic} for details.

Our approach in the proof of \autoref{thm:K-is-gapped} follows closely that of \cite{kastoryano2016commuting}, who showed that strong clustering implies a gap in commuting systems. However, we need to make significant edits to address the caveats in our version of strong clustering, which only holds over ``well-separated'' tripartitions. Our proof proceeds in two stages, where we first establish that above a sufficiently large length-scale $\ell_0$ the generator $\CK$ admits an inverse polylogarithmic (in the subregion size) conditional spectral gap. This already improves on the a priori inverse-exponential bound computed back in \autoref{lem:local_gap_K}; and it enables us to derive improved locality properties for the associated spectral conditional expectation above the length-scale $\ell_0$. In turn, these improved locality properties give rise to a tighter version of strong-clustering, which allows us to bootstrap the spectral gap up to a fixed constant at a larger length-scale $\ell_1$.

\subsection{An Outline of the Proof}

We refer the reader back to \autoref{defn:variance} and \autoref{defn:conditional_covariance} for the definitions of conditional variance. The first result of this section is the following inverse-polylog lower bound on the conditional spectral gap of $\CK$; assuming strong clustering for suitably well-separated tripartitions. 

\begin{prop}
    [An Inverse $\mathsf{polylog}$ Conditional Spectral Gap for $\CK$]\label{prop:log-gap} Consider the family of spectral conditional expectations $\{\tBE_\sR\}_\sR$ w.r.t. the Gibbs state $\vrho$ of a 1D Hamiltonian (\autoref{lem:explicit_CE}). Assume $(\vrho, \{\tBE_\sR\}_\sR)$ admits strong clustering with sub-exponential error for all operators $\vO$
    \begin{equation}
        \mathsf{Cov}_{\sA\sB\sC}(\tBE_{\sA\sB}[\vO], \tBE_{\sB\sC}[\vO]) \leq  \|\vO\|^2_{\vrho}\cdot c_1 \cdot \exp\big(-|\sB|^{c_2}\big)
    \end{equation}
    over disjoint consecutive intervals $\sA,\sB,\sC$ such that $\frac{1}{2} |\sA|, \frac{1}{2}|\sC| \geq |\sB|\geq c_3\cdot |\sA\sB\sC| / \log |\sA\sB\sC|$ for some constants $c_1, c_2, c_3$ which depend only on $\beta, q$. 

    Then, there exists constants $\ell_0, \lambda_0, c>0$ dependent only on $\beta, q$ s.t. for all intervals $\sR\subseteq [n]$ s.t. $|\sR|\geq \ell_0$, the conditional spectral gap of $\CK_\sR$ satisfies:
    \begin{equation}
       \lambda_{\sR} =  \inf_{\vO}\frac{\braket{\vO, -\CK_\sR[\vO]}_{\vrho}}{\mathsf{Var}_\sR(\vO)} \geq \lambda_0\cdot \frac{1}{\log^c |\sR|}.
    \end{equation}
\end{prop}

This bound falls only slightly short of our target goal of a constant, system-size independent gap. The proof of the theorem above follows closely the recursive proof of \cite{kastoryano2016commuting} in the commuting setting, however, with a central modification to address the caveat that our version of strong clustering only holds for suitable tripartitions. As a result, under this proof technique, we are only able to establish an inverse-polylog gap.

\begin{rmk}
    The assumption in \autoref{prop:log-gap} is the version of strong clustering derived in \autoref{section:clustering}.
\end{rmk}

 The second result of this section shows how to ``bootstrap'' the improvement to the gap all the way to a constant. In particular, we show that this spectral gap lower bound allows us to further iterate on the locality properties of $\CK$ -- which gives us a stronger form of strong clustering. In turn, this strengthened strong clustering property will imply that the spectral gap of $\CK$, at a slightly larger length-scale, is in fact \textit{constant}. 

\begin{prop}
    [Spectral Gap Bootstrapping]\label{prop:bootstrapping} Consider the family of spectral conditional expectations $\{\tBE_\sR\}$ w.r.t. the Gibbs state $\vrho$ of a 1D Hamiltonian (\autoref{lem:explicit_CE}). Assume there exists constants $\ell_0\geq 1$ and $1>c_{\mathsf{gap}}\geq 0$, such that that for every interval $\sR\subseteq [n]$ of length $|\sR|\geq \ell_0$, the conditional spectral gap is at least
    \begin{equation}
        \lambda_\sR 
        \geq |\sR|^{-c_{\mathsf{gap}}}.
    \end{equation}
    Then, there exists constants $\ell_1(c_{\mathsf{gap}}, \ell_0, \beta, q)$ and $\gamma(c_{\mathsf{gap}}, \ell_0, \beta, q)$ such that $\forall \sR\subseteq [n]$ with $|\sR|\geq \ell_1$,
    \begin{equation}
        \lambda_\sR 
        \geq \gamma(c_{\mathsf{gap}}, \ell_0, \beta, q).
    \end{equation}
\end{prop}

As an immediate corollary of the propositions above \autoref{prop:log-gap} + \autoref{prop:bootstrapping}, put in conjunction with our proof of strong clustering for 1D systems over well-separated tripartitions \autoref{thm:strong_from_weak} + \autoref{thm:kms-clustering-uncon}, we conclude the generator $\CK = \sum_{a\in \CS_{[n]}^1}\CK_a$ admits a constant conditional spectral gap, for all sufficiently large intervals $\sR$. When applied to the full chain, we establish the main result of this section, \autoref{thm:K-is-gapped}.

\subsection{Spectral Gap Recursion I (Proof of \autoref{prop:log-gap})}

As mentioned, we follow the recursive proof strategy of \cite{kastoryano2016commuting}. The proof of \autoref{prop:log-gap} follows from the discussion in \cite[Theorem 23]{kastoryano2016commuting}, with modifications which we highlight as follows. To set up the recursion, we first require two lemmas. The first of which concerns a factorization of variance statement, which enables one (via strong clustering) to recursively relate the conditional variance with respect to intervals of the chain, to that of subintervals. 

\begin{lem}
    [Factorization of Variance, \cite{kastoryano2016commuting}]\label{lem:factorization1}
    In the context of \autoref{prop:log-gap}, consider any interval $\sI \subseteq [n]$ of the 1D spin chain, and any contiguous partition $\sI = \sA\cup\sB\cup \sC$ of the interval into disjoint intervals. Suppose there exists a positive constant $\epsilon>0$, s.t. for all observables $\vO:$
    \begin{equation}
        \mathsf{Cov}_{\sA\sB\sC}(\tBE_{\sA\sB}[\vO], \tBE_{\sB\sC}[\vO]) \leq \epsilon\cdot \mathsf{Var}_{\sA\sB\sC}(\vO).
    \end{equation}
    Then, the conditional variances approximately factorize:
    \begin{equation}
        \mathsf{Var}_{\sA\sB\sC}(\vO) \leq (1-2\epsilon)^{-1} \big( \mathsf{Var}_{\sA\sB}(\vO) + \mathsf{Var}_{\sB\sC}(\vO)\big).
    \end{equation}
\end{lem}

 The following proof is essentially verbatim from \cite[Proposition 20]{kastoryano2016commuting}, and presented for completeness.

\begin{proof}

    WLOG we may consider $\tBE_{\sA\sB\sC}[\vO] = \vO$. Consider
    \begin{align}
       0 &\leq \|(1-\tBE_{\sA\sB\sC})(1 - \tBE_{\sA\sB} - \tBE_{\sB\sC})[\vO]\|_{\vrho}^2   \\ & =  \|(1-\tBE_{\sA\sB\sC})(1 - \tBE_{\sA\sB})[\vO]\|_{\vrho}^2 + \|(1-\tBE_{\sA\sB\sC})(1  - \tBE_{\sB\sC})[\vO]\|_{\vrho}^2 + \|(1-\tBE_{\sA\sB\sC})[\vO]\|_{\vrho}^2 \\
       & \quad + \quad  \text{All 6 Cross Terms} = \\
       &= \|(1-\tBE_{\sA\sB\sC})(1 - \tBE_{\sA\sB})[\vO]\|_{\vrho}^2 + \|(1-\tBE_{\sA\sB\sC})(1  - \tBE_{\sB\sC})[\vO]\|_{\vrho}^2 + \|(1-\tBE_{\sA\sB\sC})[\vO]\|_{\vrho}^2  \\
       & + 2 \braket{(1-\tBE_{\sA\sB\sC})\tBE_{\sA\sB}[\vO], (1-\tBE_{\sA\sB\sC})\tBE_{\sB\sC}[\vO]}_{\vrho} -  \|(1-\tBE_{\sA\sB\sC})[\vO]\|_{\vrho}^2.
    \end{align}
    Then, by rearranging the terms above, leveraging the fact that $\tBE$ is contractive:
    \begin{equation}
         \|(1-\tBE_{\sA\sB\sC})[\vO]\|_{\vrho}^2\leq  \|(1-\tBE_{\sA\sB})[\vO]\|_{\vrho}^2 +  \|(1-\tBE_{\sB\sC})[\vO]\|_{\vrho}^2 +  2 \mathsf{Cov}_{\sA\sB\sC}(\tBE_{\sA\sB}[\vO], \tBE_{\sB\sC}[\vO]),
    \end{equation}
    as desired.
\end{proof}

Next, we require a lemma on the partition of a generic interval of the 1D chain into a collection of ``balanced'' sub-intervals. This will later serve to inform the sizes of the sub-intervals during the recursive step. \autoref{lem:balanced_partitions} below is analogous to those used in \cite{martinelli1994approach, kastoryano2016commuting}, however, we require the middle interval to be large relative to the other two; due to the constraint on our version of strong clustering.

\begin{lem}
    [$b$ - Balanced Partitions of the Chain]\label{lem:balanced_partitions} Let $b:\BZ^+\rightarrow \BZ^+$ be an arbitrary function over the integers such that $\forall x\in \BZ^+: b(x)\leq x/100$. Then, there exists constants $\ell_b, c_b>0$ such that for every interval $\sR\subseteq [n]$ of the chain of length $|\sR|=r\geq \ell_b$, there exists a set $\{(\sA_i, \sB_i, \sC)_i\}_{i\in [s]}$ of $s$ tripartitions of $\sR$ satisfying
    \begin{enumerate}
        \item For each $i$, $\sA_i, \sB_i, \sC_i$ are disjoint and $\sA_i\cup \sB_i\cup \sC_i = \sR$.
        
        \item For each $i$, $\ceil{2r/3}\geq |\sA_i|, |\sC_i|\geq \floor{r/3}$, and  $2\cdot b(r)\geq |\sB_i|\geq b(r)$.

        \item The middle regions don't overlap, $|\sB_i\cap \sB_j| = 0$ for each $i\neq j$.
        
        \item The number of partitions $s\geq c_b\cdot r/b(r)$.
    \end{enumerate}
\end{lem}

\begin{proof}

    Let the interval $\sR$ be of size $r$, WLOG on sites $[1, r]$. We define a sequence of tripartitions, where the $i$th tripartition places the intervals $\sA_i, \sB_i, \sC_i$ at locations:
    \begin{align}
        &\sA_i = [1, \floor{r/3} + i \cdot b(r)], \\ &\sB_i = [\floor{r/3} + i \cdot b(r) + 1, \floor{r/3} + (i+1) \cdot b(r)], \\ &\sC_i = [\floor{r/3} + (i+1) \cdot b(r) + 1, r].
    \end{align}

    By design, for each partition we have (1) the covering constraint $\sA_i\cup\sB_i\cup\sC_i=\sR$ and (3) the overlap constraint $\sB_i\cap \sB_j = \emptyset$ for $i\neq j$. We also have the bound on the size of $|\sB_i|=b(r)$. So long as the total number of partitions is $s = \floor{r/(10\cdot b(r))}$, we also have that the sizes satisfy
    \begin{equation}
        |\sA_i|\geq \floor{r/3} \quad \text{and}\quad  |\sC_i| \geq \floor{\frac{2r}{3}} - \frac{r}{10} - b(r) - 1 \geq r/3
    \end{equation}
    where we recall $b(x)\leq x/10$, and we assumed $r\geq 30$. This finishes (2) on the sizes of the intervals, and guarantees the bound (4) on the number of tripartitions $s$. 
\end{proof}

We are now equipped with the relevant tools to prove \autoref{prop:log-gap}, on the inverse polylogarithmic conditional spectral gap. 

\begin{proof}

    [of \autoref{prop:log-gap}] Fix $|\sR|=r$ and a global operator $\vO$, where WLOG $\tBE_{\sA\sB\sC}[\vO]=0$. We begin by applying the strong clustering assumption in the proposition statement. For sufficiently large $r$, and any decomposition of the interval $\sR = \sA\sB\sC$ with $|\sB| = \ceil{c_1\cdot r / \log r}$ and $|\sA|, |\sC|\geq r/4$, we have the following factorization of variance statement (\autoref{lem:factorization1})
    \begin{align}
        \mathsf{var}_\sR[\vO] & \leq \bigg(1-2\cdot \epsilon(|\sB|)\bigg)^{-1}\cdot \bigg(\mathsf{var}_{\sA\sB}[\vO] +\mathsf{var}_{\sB\sC}[\vO]\bigg) \\
    `&\leq \bigg(1-2\cdot \epsilon(|\sB|)\bigg)^{-1}\cdot \bigg(\frac{\braket{\vO, -\CK_{\sA\sB}[\vO]}_{\vrho}}{\lambda_{\sA\sB}} + \frac{\braket{\vO, -\CK_{\sB\sC}[\vO]}_{\vrho}}{\lambda_{\sB\sC}}\bigg)\\
        & \leq  \bigg(1-2\cdot c_2\cdot \exp\bigg[-(c_1\cdot r/\log r)^{c_3}\bigg]\bigg)^{-1} \cdot \bigg(\frac{\braket{\vO, -\CK_{\sA\sB\sC}[\vO]}_{\vrho} + \braket{\vO, -\CK_{\sB}[\vO]}_{\vrho}}{\min (\lambda_{\sA\sB}, \lambda_{\sB\sC})} \bigg).\label{eq:factorization-recursion}
    \end{align}
    Where in the second line we leveraged the definition of the conditional spectral gap on the subintervals $\sA\sB$ and $\sB\sC$. The issue in the above, akin to \cite{martinelli1994approach, kastoryano2016commuting}, is the additive term dependent on $\CK_\sB$. To proceed, we consider the collection of balanced tripartitions $\{(\sA_i, \sB_i, \sC)_i\}_{i\in [s]}$  with parameter $b(x) = \ceil{c_1\cdot x / \log x}$ as constructed in \autoref{lem:balanced_partitions}; where furthermore $s\geq d\cdot \log r$ for an appropriate constant $d$ and sufficiently large $r$. We amortize (or average) the bound in \eqref{eq:factorization-recursion} over the said collection:
    \begin{align}
         \mathsf{var}_\sR[\vO] &\leq  \frac{\bigg(1-2\cdot c_2\cdot e^{-(c_1\cdot r/\log r)^{c_3}}\bigg)^{-1}}{\min_i(\lambda_{\sA_i\sB_i}, \lambda_{\sB_i\sC_i})} \cdot \bigg(\braket{\vO, -\CK_{\sA\sB\sC}[\vO]}_{\vrho} + \frac{1}{s}\sum_i^s \braket{\vO, -\CK_{\sB_i}[\vO]}_{\vrho}\bigg) \\
         &\leq \bigg(1+\frac{2}{d\cdot \log r}\bigg) \cdot \frac{1}{\min_i(\lambda_{\sA_i\sB_i}, \lambda_{\sB_i\sC_i})} \cdot \braket{\vO, -\CK_{\sR}[\vO]}_{\vrho} ,
    \end{align}
    (again) for $r$ larger than an explicit integer constant length-scale $z$, and where we denote $\lambda_\ell:=\min_{|\sR|\leq \ell}\lambda_\sR$. With $r_i = (4/3)^i\cdot z$, we then have the following recursion for $i\geq 1$, for a sufficiently large constant (and WLOG integer) $K>1$:
    \begin{equation}
        \lambda_{r_i} \geq \lambda_{r_{i-1}}\cdot  \bigg(\frac{i}{K+i} \bigg) \geq \frac{i!\cdot K!}{(K+i)!}\cdot \lambda_{z} \geq \lambda_{z} \cdot \frac{1}{i^K}.
    \end{equation}
    To conclude the proof, for any $r\geq z$ one can pick the smallest $j\in \BZ^+$ such that $r_j \geq r\Rightarrow j\leq a\log (r/z)$ for an appropriate constant $a$. For $r\geq \ell_0=z^2$, we then have the bound:
    \begin{equation}
        \lambda_{r} \geq \lambda_{z}\cdot a^{-K} \cdot \frac{1}{\log^K (r/\ell_1)} \geq \lambda_{z}\cdot (2a)^{-K} \cdot \frac{1}{\log^K r} ,
    \end{equation}
    appropriately labeling the prefactor $\lambda_0$ then gives the advertised bound. 
\end{proof}

\subsection{Strong Clustering from a Better-Than-Linear Spectral Gap}
\label{sec:clustering-from-gap}
To establish \autoref{prop:bootstrapping}, we first need to derive an improved form of strong clustering for the conditional expectations $\{\tBE_\sR\}_\sR$, defined by the infinite-time limit of the evolution of $\CK$. This improvement is presented in \autoref{lem:strong-from-gap}, proved at the bottom of this subsection. To do so, we first prove two lemmas on the finite-time truncations of the conditional expectations. The first lemma below is a simple statement on the error of the finite time truncation as a function of the a priori spectral gap. 

\begin{lem}
    [Finite-Time Truncation of Conditional Expectations]\label{lem:truncation_expectations} In the context of \autoref{prop:bootstrapping}, for any subset $\sR\subseteq[n]$, if the conditional spectral gap of $\CK_\sR$ is at least $\lambda_\sR$, then the conditional expectation on $\sR$ can be truncated to time $t>0$ with error:
    \begin{equation}
        \bigg\|\big(\tBE_{\sR}-\tBE_{\sR, t}\big)[\vO]\bigg\|_{\vrho} \leq e^{-\lambda_\sR\cdot t}\cdot \|\vO\|_{\vrho}
    \end{equation}
\end{lem}

\begin{proof}
    Follows from a straightforward calculation, following the definition of the conditional spectral gap:
    \begin{align}
        \frac{\rd }{\rd t}\bigg\|\big(\tBE_{\sR}-\tBE_{\sR, t}\big)[\vO]\bigg\|_{\vrho}^2 = 2\braket{\tBE_{\sR, t}[\vO], \CK_\sR\circ \tBE_{\sR, t}[\vO]} \leq -2\lambda_{\mathsf{R}}\cdot \bigg\|\big(\tBE_{\sR}-\tBE_{\sR, t}\big)[\vO]\bigg\|_{\vrho}^2.
    \end{align}
\end{proof}

The second lemma below quantifies the speed in which information travels under the evolution of $\CK$. 

\begin{lem}
    [A Lieb-Robinson Bound for $\tBE$]\label{lem:LRforCE} In the context of \autoref{prop:bootstrapping}, let $\sA\cup\sB\cup\sC$ be three consecutive disjoint intervals in the 1D chain. Then, for any operator $\vO$ which is trivial on $\sB\sC$, and any time $t\geq 0$, the finite-time conditional expectations on $\sA\sB$ and $\sA\sB\sC$ differ by 
    \begin{equation}
        \bigg\|\bigg(\tBE_{\sA\sB,t}-\tBE_{\sA\sB\sC,t}\bigg)\tBE_{\sB\sC}[\vO]\bigg\|_{\vrho} \leq \|\vO\|_{\vrho}\cdot a_{\mathsf{LB}}\cdot \exp\bigg[ - c_\mathsf{LB} \bigg(|\sB| - v_{\mathsf{LB}}\cdot t\bigg)\bigg].
    \end{equation}
    where $a_{\mathsf{LB}}, c_{\mathsf{LB}}, v_{\mathsf{LB}}>0$ are constants that depend only on $\beta, q$.
\end{lem}
Since the proof is somewhat tangent to the present discussion, it is deferred to~\autoref{sec:proofLRCE}. Intuitively, if the spectral gap of $\CK_{\sA\sB\sC}$ is any better than the inverse of the size of $\sB$, then \autoref{lem:LRforCE} suggests that the conditional expectations on $\sA\sB\sC$ and $\sA\sB$ for operators trivial on $\sB\sC$ should match. This implies a ``gluing'' statement for conditional expectations akin to those developed in \autoref{section:clustering}. This gluing statement, in turn, implies an improvement to strong clustering, at least for suitable well-separated tripartitions. 

\begin{lem}
    [Strong Clustering from an a priori Spectral Gap]\label{lem:strong-from-gap} 
    In the context of \autoref{prop:bootstrapping}, assume there exists constants $\ell_0>1$, $1>c_{\mathsf{gap}}\geq 0$, s.t. the conditional spectral gap of $\CK_\sR$ on every interval $\sR\subseteq [n]$ of length $|\sR|\geq \ell_0$ satisfies the lower bound $\lambda_\sR \geq |\sR|^{-c_{\mathsf{gap}}}.$
    Then, there exists constants $\ell_1, c_{\mathsf{decay}}, 1>c_{\mathsf{size}}\geq 0$ as a function only of $(\ell_0, v_{\mathsf{LB}}, c_\mathsf{LB}, \beta, q)$, such that for every consecutive disjoint intervals $\sA, \sB, \sC$ where $|\sB|\geq |\sA\sB\sC|^{c_{\mathsf{size}}} \geq \ell_1$ and every operator $\vO$:
    \begin{equation}
        \mathsf{Cov}_{\sA\sB\sC}(\tBE_{\sA\sB}[\vO], \tBE_{\sB\sC}[\vO])  \leq \|\vO\|_{\vrho}^2\cdot \exp\big[- |\sB|^{c_{\mathsf{decay}}}\big].
    \end{equation}
\end{lem}

In words, any ``better-than-inverse-linear'' conditional spectral gap implies a version of strong clustering, but only over suitable tripartitions where the middle region is at least polynomially large, relative to the total size. One should compare this to \autoref{thm:strong_from_weak}, where the middle region must be of size $|\sA\sB\sC| / \log |\sA\sB\sC|$. Although this may seem like a modest improvement, it is precisely this improvement that will give rise to the constant spectral gap. 

\begin{proof}

    [of \autoref{lem:strong-from-gap}] Let us begin by recollecting that $\tBE_{\sA\sB\sC}\tBE_{\sB\sC}=\tBE_{\sA\sB\sC}$. We proceed by applying the triangle inequality:
    \begin{align}
         \bigg\|\bigg(\tBE_{\sA\sB} - \tBE_{\sA\sB\sC}\bigg)\tBE_{\sB\sC}[\vO]\bigg\|_{\vrho} \leq & \bigg\|\bigg(\tBE_{\sA\sB} - \tBE_{\sA\sB, t}\bigg)\tBE_{\sB\sC}[\vO]\bigg\|_{\vrho} + \\
           & \bigg\|\bigg(\tBE_{\sA\sB, t} - \tBE_{\sA\sB\sC, t}\bigg) \tBE_{\sB\sC}[\vO]\bigg\|_{\vrho} + \\
          & \bigg\|\bigg(\tBE_{\sA\sB\sC, t} - \tBE_{\sA\sB\sC}\bigg)\tBE_{\sB\sC}[\vO]\bigg\|_{\vrho} ,\label{eq:trotterization_error}
    \end{align}
    \noindent and analyze the resulting contributions separately. The first and the third term above are both special cases of \autoref{lem:truncation_expectations}; as a result:
    \begin{align}
        \bigg\|\bigg(\tBE_{\sA\sB} - \tBE_{\sA\sB, t}\bigg)\tBE_{\sB\sC}[\vO]\bigg\|_{\vrho}, \bigg\|\bigg(\tBE_{\sA\sB\sC, t} - \tBE_{\sA\sB\sC}\bigg)\tBE_{\sB\sC}[\vO]\bigg\|_{\vrho} \leq \|\vO\|_{\vrho}\cdot \exp\bigg[ -\min(\lambda_{\mathsf{ABC}}, \lambda_{\mathsf{AB}})\cdot t\bigg].
    \end{align}
    In turn, we note that the operator $\tBE_{\sB\sC}[\vO]$ is trivial on $\sB\sC$. By the Lieb-Robinson bounds for the spectral conditional expectation in \autoref{lem:LRforCE}:
    \begin{align}
        \bigg\|\bigg(\tBE_{\sA\sB, t} - \tBE_{\sA\sB\sC, t}\bigg) \tBE_{\sB\sC}[\vO]\bigg\|_{\vrho} \leq \|\vO\|_{\vrho}\cdot a_{\mathsf{LB}} \exp\bigg[ - c_\mathsf{LB} \bigg(|\sB| - v_{\mathsf{LB}}\cdot t\bigg)\bigg].
    \end{align}
    Now, assume that $|\sA\sB| \geq \ell_0$, such that we can leverage the assumption in the lemma statement on the conditional spectral gap of $\CK$.  For any constant $\delta$, let us fix
    \begin{equation}
        t = |\sA\sB\sC|^{c_{\mathsf{gap}}}\cdot|\sB|^{\delta}.
    \end{equation}
    Then, so long as the size of $\sB$ satisfies the lower bound
    \begin{equation}\label{eq:size-of-b}
        |\sB| \geq  2\cdot v_{\mathsf{LB}}\cdot |\sA\sB\sC|^{\frac{c_{\mathsf{gap}}}{1-\delta}}.
    \end{equation}
    We have that the conditional expectations on $\sA,\sB,\sC$ admit the following gluing statement.
    \begin{equation}
        \bigg\|\bigg(\tBE_{\sA\sB}\tBE_{\sB\sC} - \tBE_{\sA\sB\sC}\bigg)[\vO]\bigg\|_{\vrho} \lesssim \|\vO\|_{\vrho}\cdot \max(a_{\mathsf{LB}}, 1) \exp\bigg[-\min(c_{\mathsf{LB}}, 1)\cdot |\sB|^\delta\bigg].
    \end{equation}
    Since the conditional expectations are projective, we then get strong clustering as in \autoref{defn:strong_clustering}:
    \begin{align}
        \mathsf{Cov}_{\sA\sB\sC}(\tBE_{\sA\sB}[\vO], \tBE_{\sB\sC}[\vO])  &= \braket{\vO, \bigg(\tBE_{\sA\sB}\tBE_{\sB\sC} - \tBE_{\sA\sB\sC}\bigg)[\vO]}_{\vrho} \\&\leq \|\vO\|_{\vrho}\cdot \bigg\|\bigg(\tBE_{\sA\sB}\tBE_{\sB\sC} - \tBE_{\sA\sB\sC}\bigg)[\vO]\bigg\|_{\vrho}  \\ &\lesssim \|\vO\|_{\vrho}^2\cdot \max(a_{\mathsf{LB}}, 1) \exp\bigg[-\min(c_{\mathsf{LB}}, 1)\cdot |\sB|^\delta\bigg].
    \end{align}
    To conclude, since $c_{\mathsf{gap}}<1$, one can pick $\delta = \frac{1}{2}(1-c_{\mathsf{gap}})>0$ and $|\sB|$ sufficiently large relative as a function of $c_{\mathsf{gap}}, c_{\mathsf{LB}}, v_{\mathsf{LB}}, a_{\mathsf{LB}}$, such that if
    \begin{equation}
        |\sB| \geq 2\cdot v_{\mathsf{LB}}\cdot|\sA\sB\sC|^{\frac{2c_{\mathsf{gap}}}{1+c_{\mathsf{gap}}}} \Rightarrow \mathsf{Cov}_{\sA\sB\sC}(\tBE_{\sA\sB}[\vO], \tBE_{\sB\sC}[\vO])  \leq \|\vO\|_{\vrho}^2\cdot \exp\bigg[- |\sB|^{\frac{1-c_{\mathsf{gap}}}{4}}\bigg].
    \end{equation}
    Appropriately adjusting the constants gives the advertised bound.  
    
\end{proof}

\subsection{Spectral Gap Recursion II (Proof of \autoref{prop:bootstrapping})}

We are now equipped with the relevant tools to prove \autoref{prop:bootstrapping}, on the constant spectral gap of $\CK$. The proof is also based on a recursion over conditional spectral gaps akin to \cite{martinelli1994approach, kastoryano2016commuting} and is very similar to that of \autoref{prop:log-gap}. The only distinction lies in the partitioning scheme: now that strong clustering holds over ``well-separated" tripartitions with improved bounds on the separation, we can leverage a more refined partitioning scheme, which will yield better parameters in the recursion.

\begin{proof}

    [of \autoref{prop:bootstrapping}] Fix $|\sR|=r$, and consider the balancing function $b(x) = \ceil{x^{c_{\mathsf{size}}}}$ with $c_{\mathsf{size}}$ as defined in \autoref{lem:strong-from-gap}. Consider the collection of $s \geq r^{c_{\mathsf{count}}}$ decompositions of the interval $\sR = \sA_i\sB_i\sC_i$ with $|\sB_i| \geq r^{c_{\mathsf{size}}}$ guaranteed by \autoref{lem:balanced_partitions}, so long as $r$ is sufficiently large and $1>c_{\mathsf{count}}> 0$ is an appropriately chosen constant. For any global operator $\vO$ with $\tBE_{\sA\sB\sC}[\vO]=0$, we again have from \autoref{lem:factorization1} and strong clustering \autoref{lem:strong-from-gap} (under the gap assumption) that
    \begin{align}
        \mathsf{var}_\sR[\vO] & \leq \bigg(1+\frac{1}{s}\bigg) \cdot  \bigg(1- 2e^{-r^{c_{\mathsf{decay}}\cdot c_{\mathsf{size}}}}\bigg)^{-1}\cdot \frac{1}{\min_i(\lambda_{\sA_i\sB_i}, \lambda_{\sB_i\sC_i})} \cdot \braket{\vO, -\CK_{\sR}[\vO]}_{\vrho} \\
        &\leq \bigg(1+\frac{2}{r^{c_{\mathsf{count}}}}\bigg) \cdot \frac{1}{\min_i(\lambda_{\sA_i\sB_i}, \lambda_{\sB_i\sC_i})} \cdot \braket{\vO, -\CK_{\sR}[\vO]}_{\vrho}.
    \end{align}
    So long as $r$ is at least a sufficiently large constant $r_0$ as a function of $c_{\mathsf{count}}, c_{\mathsf{size}}, c_{\mathsf{decay}}$, as well as the length-scales of strong clustering \autoref{lem:strong-from-gap}, and the partitioning \autoref{lem:balanced_partitions}. We denote $\lambda_\ell:=\min_{|\sR|\leq \ell}\lambda_\sR$. With $r_i = (4/3)^i\cdot r_0$, we then have the following recursion for $i\geq 1$:
    \begin{align}
         \lambda_{r_i} \geq \lambda_{r_{i-1}} \cdot  \bigg(1+\bigg(\frac{3}{4}\bigg)^{i\cdot c_{\mathsf{count}}}\bigg)^{-1}.
    \end{align}
    To solve this recursion, we observe that for any fixed constant $0\leq a<1$ (in this case, $a =  (3/4)^{c_{\mathsf{count}}}$), we have the geometric sum:
    \begin{equation}
        \prod_{i=1}(1+a^i)\leq \exp\bigg(\sum_{i=1}a^i \bigg) = \exp\bigg(\frac{a}{1-a}\bigg)
    \end{equation}
    and therefore, one can solve the recursion by computing:
    \begin{align}
        \lambda_{r_i} \geq \lambda_{r_0} \cdot  \prod_{j=1}^\infty\bigg(1+\bigg(\frac{3}{4}\bigg)^{i\cdot c_{\mathsf{count}}}\bigg)^{-1} \geq \lambda_{r_0}\cdot \exp\bigg(-\frac{a}{1-a}\bigg)
    \end{align}
    which is always lower-bounded by a constant, as desired. 
\end{proof}

\subsection{Lieb-Robinson Bounds for the Spectral Conditional Expectation (\autoref{lem:LRforCE})}
\label{sec:proofLRCE}
Since the $\norm{\cdot}_{\vrho}$ can be treated as a spectral norm under similarity transformation, the proof of \autoref{lem:LRforCE} reduces to a standard Lieb-Robinson bound for rapidly decaying interactions~\cite{hastings2006spectral} but in the KMS-norm. However, we must be somewhat careful with how we decompose the generator $\CK$ into smaller pieces, as they may not individually preserve the KMS norm.

\begin{proof}

    [of~\autoref{lem:LRforCE}] WLOG we arrange the intervals $\sA, \sB, \sC$ from left to right, where the leftmost site of $\sA$ is site $1$. For any single-site jump $a\in \CS_{\sA\sB\sC}^1$, let its location be the site $x(a)\in \sA\sB\sC \subseteq [n].$ We begin by considering an annulus decomposition for $\CK_a$ akin to that computed in the proof of \autoref{lem:localized_E}.
    \begin{align}
        \CK_a = \sum_{\ell} \delta\CK_a^{\ell}\quad \text{where}\quad \delta\CK_a^{\ell}:= \CK_a^{(\ell)} - \CK_a^{(\ell-1)} \quad \text{and}\quad \|\delta\CK_a^\ell\|_{\vrho}  \leq d_1\cdot \frac{d_2^{\ell}}{\ell!},\label{eq:K-annulus-lb}
    \end{align}
    for appropriate constants $d_1, d_2$. We note $\delta\CK_a^{\ell}$ is supported on the interval $[x(a)-\ell,x(a)+\ell]$. These decomposed superoperators still respect some locality:
    \begin{align}
      \delta\CK_a^{\ell}[\vY] = 0\quad \text{if}\quad   \vY \quad \text{is trivial on $[x(a)-\ell,x(a)+\ell]$}\label{eq:commute_with_Aa},
    \end{align}
    since $\delta\CK_a^{\ell}[\vY]$ vanishes if $[\vA,\vY]=0.$ This can almost be plugged into the Lieb-Robinson bounds for rapidly decaying interactions~\cite{hastings2006spectral}, except that individual terms are not Hermitian, so exponentiating a subset of $\delta\CK_a^{\ell}$ could increase the KMS norm $\norm{\cdot}_{\vrho}$, thereby spoiling the Lieb-Robinson bound. The remedy is that if we carefully expand the difference of exponentials $\tBE_{\sA\sB,t}-\tBE_{\sA\sB\sC,t}$, many of the terms still regroup together, nearly recovering the original $\CK_a$ (which monotonically decreases the KMS-norm).\\

    \noindent \textbf{The Growth of the Superoperator Norm.} To proceed, for integer $w \in [n],$ consider the subset of terms in the expansion~\eqref{eq:commute_with_Aa} labeled by $a\in \CS_{\sA\sB\sC}^1, \ell\in [n]$, whose support does not exceed $z$:
    \begin{align}
        \CV_{<w} &:= \{ a, \ell: a \in \CS_{\sA\sB\sC}^1 \quad \text{and}\quad  x(a)+\ell < w \} \subseteq \CS_{\sA\sB\sC}^1\times [n],
    \end{align}
    and let us further introduce the truncation of the generator $\CK$ to such terms: 
    \begin{align}
        \CK_{< w} &:= \sum_{(a, \ell)\in \CV_{<w}} \delta\CK_a^{\ell} = \sum_{\substack{a\in \CS_{\sA\sB\sC}^1 \text{ s.t.} \\ x(a) \in [1, w-1]}} \CK^{(\mathsf{dist}(a,w)-1)}_a.
    \end{align}
    
    Indeed, for a jump $a$ far away from the boundary index $w,$ almost all the components $\delta\CK_a^{\ell}$ in the decomposition of $\CK_{a}$ are contained in $\CK_{< w}$, and the ``strength'' of non-Hermiticity can be shown to be actually constant~\eqref{eq:nonHermiticity_bound}. Indeed, for any $x\in [n]$ and $s\in \BR^+$, the KMS induced superoperator norm is controlled by the norm of the non-self-adjoint part. By elementary linear algebra (\autoref{lem:trotter}), we then have:
\begin{align}
  \sum_{\substack{a\in \CS_{\sA\sB\sC}^1 \text{ s.t.} \\ x(a) \in [1, w-1]}}  \lnorm{\CK_a - \CK^{(\mathsf{dist}(a,w)-1)}_a}_{\vrho} \le c_1 \quad \Rightarrow \quad  \norm{e^{\CK_{<x} s}}_{\vrho} \le e^{c_{1} \labs{s}}, \label{eq:nonHermiticity_bound}
\end{align}
for an appropriate constant $c_1$ as a function of $d_1, d_2$ \eqref{eq:K-annulus-lb}. This significantly improves over the naive bound $e^{s\norm{\CK_{<w}}_{\vrho}}$ that grows with $w.$\\

\noindent \textbf{Duhamel's Principle.} The rest of the proof proceeds via a careful ``self-avoiding paths'' decomposition~\cite[Theorem 3]{chen2021operator}, which explicitly expose the intermediate exponentials so that~\eqref{eq:nonHermiticity_bound} applies. For integer $k\in \BN^+$, we denote as $\CP_k\subseteq (\CS_{\sA\sB\sC}^1\times [n])^{\times k}$ the set of ``paths'' or sequences of length $k$ comprised of components $(a', \ell')\in  \CS_{\sA\sB\sC}^1\times [n]$, satisfying:

\begin{align}
        \CP_{k} := \bigg\{\{(a_i,\ell_i)\}_{i\in [k]} \quad \text{s.t.} \quad i\in [k]:\undersetbrace{\text{consecutive terms overlap}}{ x(a_i) - \ell_i \le x(a_{i-1}) + \ell_{i-1}}\quad \text{and}\quad \undersetbrace{\text{next term must make progress}}{x(a_i) + \ell_i \ge x(a_{i-1}) + \ell_{i-1}}\bigg\}.\label{eq:pathsk}
    \end{align}
    For $z\in \CS_{\sA\sB\sC}^1$, we denote as $\CP_{k}^z\subseteq \CP_{k}$ the subset of paths where the $k$th step is the first such that $x(z)\leq x(a_k) + \ell_k$. We denote the edge cases $x(a_{k+1}) =x(z), \ell_{k+1}=\ell_0=0$, and $x(a_0):=y_0$, where $y_0\in [n]$ is the right-most site of $\sA$. Now, by Duhamel's principle, 
    \begin{align}
        \tBE_{\sA\sB,t}-\tBE_{\sA\sB\sC,t} =  \int_{0}^t \tBE_{\sA\sB\sC,t-s} \CK_{\sC} \tBE_{\sA\sB,s} \rd s\label{eq:DuhamelEABC},
    \end{align}
    with $\tBE_{\sB\sC}[\vO] = \vX$, one can bound each term $z\in \CS_\sC^1$ in the integral above by applying the sum over self-avoiding paths~\cite[Theorem 3]{chen2021operator} for $\CK = \sum_{a,\ell} \delta{\CK}_{a}^\ell$. We have $ \CK_{z} \cdot  \tBE_{\sA\sB,s} [\vX] = $ 
    \begin{align} 
        &=  \sum_{k=1}^{\infty} \sum_{\{(a_i, \ell_i)\} \in \CP_k^z}\CK_{z}\int_{s>s_k > \cdots s_1 >0}  \tBE_{\sA\sB,(s-s_k)} \prod_{i\in [k]} \bigg(\delta\CK_{a_i}^{\ell_i}\cdot e^{ (s_i-s_{i-1})\cdot \CK_{<x(a_{i+1})-\ell_{i+1}}} \bigg)[\vX] \prod_{i\in [k]}\rd s_i,
    \end{align}
   with the edge case $s_0 = 0$. As a consequence of the bounds \eqref{eq:nonHermiticity_bound} on the super-operator norm, we then have:    \begin{align}
        \|\CK_{z} \tBE_{\sA\sB,s}[\vX] \|_{\vrho} \le \|\vX\|_{\vrho}\cdot \lnorm{ \CK_{z} }_{\vrho}\cdot  e^{c_1s}\cdot \sum_{k=1}^{\infty}  \frac{s^k}{k!}\cdot  \sum_{\{(a_i, \ell_i)\} \in \CP_k^z} \prod_{i\in [k]}\|\delta\CK^{\ell_i}_{a_i}\|_{\vrho}.\label{eq:sumoverpath}
    \end{align}

    The path-counting expression above (without $e^{c_1s}$) is essentially what one encounters in analyzing Lieb-Robinson bounds with exponentially decaying strength~\cite{hastings2006spectral}; we give a self-contained argument for completeness. \\

\noindent \textbf{Path Counting.} To organize the sum over paths, we proceed by pinning the right-most site $y_i = x(a_i)+\ell_i$ in the support of each component $(a_i, \ell_i)\in \CS^1_{\sA\sB\sC}\times [n]$ on the path. Recall $y_0$ is the rightmost site in $\sA$. Note $y_0<y_1\cdots <y_{k-1}<x(z)\leq y_k$. Given fixed consecutive right-most sites $y_{j+1} > y_j$, we proceed by summing over all possible consistent terms $(a_{j+1}, \ell_{j+1})$: 
\begin{align}
\sum_{\substack{x(a)+ \ell =y_{j+1} \\ x(a) - \ell \leq y_{j}}} \lnorm{\delta\CK_{a}^\ell}  &= \sum_{\ell \geq (y_{j+1}-y_j)/2}\sum_{x(a) = y_{j+1}-\ell}\lnorm{\delta\CK_{a}^\ell} \\
&\le g'_1 \sum_{\ell \geq (y_{j+1}-y_j)/2} \frac{d_2^\ell}{\ell!} \leq g_1\cdot e^{- g_2(y_{j+1}-y_j)},
\end{align}
for appropriately chosen constants $g_1, g'_1, g_2$. For each fixed sequence, $y_0 < y_1 < \cdots < y_{k-1} < x(z)\leq y_k,$ the contribution decays exponentially with the total distance but grows exponentially with $k$. 
\begin{align}
    \sum_{\{(a_i, \ell_i)\} \in \CP_k^z} \prod_{i\in [k]}\|\delta\CK^{\ell_i}_{a_i}\|_{\vrho} &\leq g_1^k\cdot  \sum_{y_k\geq x(z)>y_{k-1} >\cdots > y_1 > y_0} \prod_i e^{-g_2(y_{i+1}-y_{i})} \\
    &\leq g_1^k \cdot \sum_{y_k\geq x(z)} \binom{y_k-y_0-1}{k-1} \cdot e^{-g_2(y_k-y_0)} \tag{Counting Subsequences}\\
    &\leq  g_1^k \cdot (2/g_2)^k \cdot \sum_{y_k\geq x(z)} e^{-g_2(y_k-y_0)/2} \tag{By $\binom{w}{k}\leq e^{w/r}\cdot r^k$, and $r = 2/g_2$.} \\
    &\leq g_3^k\cdot e^{-g_2(x(z)-y_0)/2},
\end{align}
for an appropriately chosen constant $g_3$. We can now return to \eqref{eq:sumoverpath}, to bound
\begin{align}
   \sum_{z\in \CS_\sC^1} \|\CK_{z} \tBE_{\sA\sB,s}[\vX] \|_{\vrho} &\le \|\vX\|_{\vrho}\cdot g_4\cdot  e^{c_1s} \cdot \bigg(\sum_{x(z)\geq y_0+|\sB|}  e^{-g_2(x(z)-y_0)/2}\bigg)\cdot  \sum_k \frac{(sg_3)^k}{k!} \\
   &\leq \|\vX\|_{\vrho}\cdot g_5\cdot \exp\bigg(-g_2\cdot \mathsf{dist}(\sA, \sC)/2 + g_6\cdot s\bigg).
\end{align}
Integrating over $s\in [0, t]$ in \eqref{eq:DuhamelEABC} and relabeling the constants then gives the desired bound. 
\end{proof}

\begin{lem}\label{lem:trotter}
Suppose a Hermitian matrix has no positive eigenvalues $\vH \prec 0$. Then,
    \begin{align}
        \norm{e^{\vH + \vY}} \le e^{\norm{\vY}}.
    \end{align}
\end{lem}
\begin{proof}
By the Trotterization limit
\begin{align}
    e^{\vH + \vY} = \lim_{p\rightarrow \infty} (e^{\vH/p} e^{\vY/p})^{p} 
\end{align}
and use that $\norm{e^{\vH/p}} \le 1$ and $\norm{e^{\vY/p}} \le e^{\norm{\vY}/p}$ to conclude the proof. 
\end{proof}

\section{Transferring the gap of \texorpdfstring{$\CK$}{K} to that of the Lindbladian}
\label{sec:CKGgap}
So far, our arguments have been largely built around the auxiliary generator $\CK,$ due to its nice interplay with the spectral conditional expectation. Although we do not know how to do the recursion on $\CL$ directly, fortunately, we are able to transfer the hard-earned constant spectral gap for the generator $\CK$ in \eqref{eq:CK}, to that of the Lindbladian $\CL$ of \cite{chen2023efficient}. We refer the reader to \eqref{eq:exact_DB_L} to recollect the explicit form of $\CL$.

\begin{thm}
    [Comparing the Gap of $\CK$ to that of $\CL$]\label{thm:ckg_gap} For any inverse-temperature $\beta\in \BR^+$, and local dimension $2^q$, fix a family of 1D Hamiltonians (\autoref{defn:1D}), and consider the generator $\CK$ of the spectral conditional expectation~\eqref{eq:CK}, and the Lindbladian $\CL$~\eqref{eq:exact_DB_L} defined by the set of single-site Pauli jumps $\CS^1_{[n]}$ and Gaussian weight function with width $\sigma = \beta^{-1}$, which are both detailed-balanced under the Gibbs state $\vrho$.
    Then, there exists a constant $\infty >\alpha>0$ dependent only on $\beta, q$ such that, for any operator $\vO,$
    \begin{equation}
        \alpha\cdot \braket{\vO, -\CL^\dagger[\vO]}_{\vrho} \ge \braket{\vO, -\CK[\vO]}_{\vrho} = \sum_{\vS\in \CS^1_{[n]}} \|[\vS, \vO]\|_{\vrho}^2. \label{eq:compare_dirichlet}
        \end{equation}
    Consequently, the spectral gap of $\CL$ is bounded by the spectral gap of $\CK$:
    \begin{equation}
        \inf_{\vO}  \frac{\braket{\vO, -\CL^\dagger[\vO]}_{\vrho}}{\|\vO - \vI\cdot \tr[\vrho\vO]\|^2_{\vrho}}  \geq  \frac{1}{\alpha}\cdot  \inf_{\vO} \frac{\braket{\vO, -\CK[\vO]}_{\vrho}}{\|\vO - \vI\cdot \tr[\vrho\vO]\|^2_{\vrho}}. 
    \end{equation}
\end{thm}

Here, we focus on single-site Pauli jump operators for simplicity. Since this comparison is the only place in our proof where the choice of Lindbladian explicitly comes into play, our spectral gap results are by no means restricted to~\cite{chen2023efficient}.

\begin{rmk}[Other Lindbladians] \label{rmk:other_L}
    The comparison above holds for both the Metropolis and Gaussian transition weights, see \eqref{eq:Metropolis}, with different $\beta$-dependences in $\alpha$. In principle, we believe the comparison holds for KMS-detailed-balanced Lindbladians with (1) a locally ergodic collection of jump operators and (2) reasonable choices of weights and filters that may come from~\cite{ding2024efficient,SA24}. 
\end{rmk}

\subsection{The Dirichlet Form Comparison}

The starting point for~\eqref{eq:compare_dirichlet} is the following explicit integral expression for the Dirichlet form of the Lindbladian $\CL$ \eqref{eq:exact_DB_L}.

\begin{lem}[The Dirichlet Form of $\CL$ {~\cite[Lemma X.2]{chen2025quantumMarkov}}]\label{lem:Dirichlet}
    Fix a single Hermitian jump operator $\vA^a = \vA^{a, \dagger}$. The Dirichlet form for the local Lindbladian $\CL_a$~\eqref{eq:exact_DB_L} associated to the jump operator $\vA^a$ with frequency width $\sigma\in \BR^+$, can be written as
\begin{align}
    \CE_{a}(\vO,\vO) &:= -\braket{\vO,\CL_a^{\dagger}[\vO]}_{\vrho}=  \iint_{-\infty}^{\infty} g(t)h(\omega) \cdot \norm{[\hat{\vA}^a(\omega,t),\vO]}^2_{\vrho} \cdot  \rd t \rd \omega.\label{eq:dirichlet_L}
\end{align}
where $\hat{\vA}^a(\omega,t):=\e^{i\vH t}\hat{\vA}^a(\omega)\e^{-i\vH t}$,  and $g(t) = \frac{1}{\beta\cosh(2\pi t/ \beta)}\ge 0.$ Further, the frequency filter $h(\omega)$ depends on the choice of transition rate $\gamma(\omega)$ \eqref{eq:Metropolis} as:
\begin{align}
    \emph{(Metropolis)} \quad h_\mathsf{M}(\omega) = e^{-\sigma^2\beta^2/8}e^{-|\omega|\beta/2}, \quad \emph{(Gaussian)} \quad h_\mathsf{G}(\omega) = e^{-1/4}\exp\bigg(-\frac{\omega^2\beta^2}{2(2-\sigma^2\beta^2)}\bigg).
\end{align}
\end{lem}

\begin{rmk}
    Note that when $\sigma = 1/\beta$, we have $\forall\omega\in \BR: h_\mathsf{M}(\omega)\geq h_\mathsf{G}(\omega)$. Therefore, the spectral gap of the Lindbladian under the Gaussian weight always provides a lower bound for that under the Metropolis weight. In this section, we thus only prove the relation under the Gaussian weight function. 
\end{rmk}

Roughly, the reason why~\eqref{eq:compare_dirichlet} should be possible, is that the kernel of $\CK_a$ should contain the kernel of $\CL^\dagger_a$:
\begin{align}
    [\hat{\vA}^a(\omega,t),\vO] = 0 \quad \text{for all $\omega,t \in \BR$} \quad \text{implies that}\quad [\vA^a,\vO] =0 ,\label{eq:rough_comparison}
\end{align}
since $\vA^a$ can be decomposed as an integral over the operator Fourier transforms $\hat{\vA}^a(\omega,t).$ 
Unfortunately, to establish a robust, quantitative version of \eqref{eq:rough_comparison}, challenges arise in both (I) the choice of norm and (II) the non-locality of $\hat{\vA}(\omega, t)$.
To proceed, we heavily rely on fine-grained quasi-locality properties of operator Fourier transforms in both space and frequency in order to relate the Dirichlet forms of the full chain. To begin, we build up two tools on the locality of Operator Fourier Transforms in 1D, whose proofs are deferred to \autoref{section:deferred-dirichlet-comparison} for clarity. 

\begin{lem}
    [Locality of Operator Fourier Transforms in Space and Frequency]\label{lem:annulus_OFT}
    For any single-site operator $\vA$, $\|\vA\|\leq 1$, supported on site $i\in [n]$, frequency $\omega\in \BR$, frequency width $\sigma \in \BR^{+}$, the operator Fourier transform w.r.t. a 1D Hamiltonian \autoref{defn:1D} admits an annulus decomposition
    \begin{equation}
        \hat{\vA}_\sigma(\omega) = \sum_\ell \vY_{\ell, \omega, \sigma}, \quad \text{where}\quad  \| \vY_{\ell, \omega, \sigma}\| \leq c_1 \cdot e^{-c_2|\omega|-c_3\ell\log\ell},
    \end{equation}
    where $\vY_{\ell, \omega, \sigma}$ is supported on the neighbourhood interval $[i-\ell,i+\ell]$, and where the constants $c_1, c_2, c_3$ depend only on $\sigma,  q$.
\end{lem}

Roughly speaking, the operator Fourier transform of local operators satisfies locality in frequency: a local jump $\vA$ is unlikely to substantially change the energy of the state. Furthermore, in 1D, the real-time evolution of local operators admits a Lieb-Robinson bound and thereby should be localized in space.

Building on \autoref{lem:annulus_OFT}, we show that commutators with operator Fourier transforms of local jumps admit a similar form of locality, in that they are bounded by the commutators of ``nearby'' Pauli jumps. 

\begin{lem}
    [Locality of Commutators of Operator Fourier Transforms]\label{lem:norms-of-commutators} In the context of \autoref{lem:annulus_OFT}, for any observable $\vO$, the commutator of the operator Fourier transform of $\vA$ w.r.t. a 1D Hamiltonian, satisfies the following norm bound:
    \begin{align}
        \lnorm{[\hat{\vA}_{\sigma}(\omega),\vO]}_{\vrho} \le c_1\cdot e^{-c_2|\omega|}\cdot \sum_{j\in [n]}\sum_{\vB\in \CS_j^1}   \|[\vB, \vO]\|_{\vrho}\cdot e^{-c_3\cdot |i-j|\ln |i-j|},
    \end{align}
    where $\vB\in \CS_{j}^1$ denotes the sum over all single-site Pauli operators on site $j$, and $c_1, c_2, c_3>0$ are constants that depend only on $\sigma,\beta,  q$.
\end{lem}

We are now in a position to relate the Dirichlet forms~\eqref{eq:compare_dirichlet}.

\subsubsection{Proof of the Dirichlet Form Comparison \eqref{eq:compare_dirichlet}}

The high-level idea of the proof is to write $\vA$ in terms of a linear combination of the operator Fourier transforms at different frequencies. We proceed by decomposing the linear combination into a ``low-frequency'' component, whose commutators can be related directly to the Dirichlet form of $\CL$; and a ``high-frequency'' component, which we analyze via the previously derived locality properties of the operator Fourier transforms (\autoref{lem:norms-of-commutators}).

\begin{proof}

[of~\autoref{thm:ckg_gap}] To begin, consider a single-site Pauli jump operator $\vA \in \CS_{[n]}^1$, WLOG with support on some site $i\in [n]$. We first rewrite the operator $\vA$ in terms of a linear combination of its operator Fourier transform $\hat{\vA}_{\sigma_1}(\omega)$, for a suitably chosen uncertainty $\sigma_1 = 1/(\sqrt{2}\beta)$. Then, we split the integral over $\omega$ into ``high'' and ``low'' frequency components, tuned by a threshold $\Omega>0.$
\begin{align}
   \sqrt{2\sigma_1\sqrt{2\pi}} \vA &= \int_{-\infty}^{\infty} \hat{\vA}_{\sigma_1}(\omega) \rd \omega  = \int_{-\infty}^{\infty} \bigg(\indicator(\labs{\omega}\le \Omega) + \indicator(\labs{\omega}> \Omega) \bigg) \hat{\vA}_{\sigma_1}(\omega) \rd \omega.\\
    & = \undersetbrace{\text{low-frequency}}{\sqrt{\frac{\sigma_2\sigma_3\sqrt{2\pi}}{\sigma_1}}\cdot  \int_{\labs{\omega}\le \Omega} \int_{-\infty}^{\infty} \hat{\vA}_{\sigma_2}(\omega)(t) e^{-i \omega t}f_{\sigma_3}(t) \rd t \rd \omega} + \undersetbrace{\text{high-frequency}}{\int_{\labs{\omega}> \Omega} \hat{\vA}_{\sigma_1}(\omega) \rd \omega}.\label{eq:A_high_low}
\end{align}
In the second line, we re-wrote the low frequency component $\labs{\omega} \le \Omega$ in terms of the time-evolved operators $\hat{\vA}(\omega,t)$ using a simple convolution property of operator Fourier transforms (\autoref{lem:OFT_Gaussian_expansion}, below). The uncertainties are defined to satisfy $1/\sigma_1^2 = 1/\sigma_2^2 +1/\sigma_3^2$, and in particular, we pick $\sigma_2 = \sigma_3 = 1/\beta:=\sigma$. We proceed by taking the $\vrho$-weighted norm with the commutator $\norm{[\cdot,\vO]}_{\vrho}$ on both sides. \\

\noindent \textbf{The Low Frequency Component.} The low frequency $\labs{\omega} \le \Omega$ part can be bounded by the Dirichlet form of the local Lindbladian $\CL_a$ via the Cauchy-Schwarz inequality over $\rd \omega \rd t$:
\begin{align}
    \int_{\labs{\omega}\le \Omega} \int_{-\infty}^{\infty} \norm{[\hat{\vA}_{\sigma}(\omega)(t),\vO]}_{\vrho}\cdot  \labs{f_{\sigma}(t)} \rd t \rd \omega 
    &\leq \sqrt{\int_{-\infty}^{\infty}  \int_{\labs{\omega}\le \Omega} \frac{f^2_{\sigma}(t)}{g(t)h_\mathsf{G}(\omega)}\rd t \rd \omega} \sqrt{\CE(\vO,\vO)}\\
    & \lesssim e^{\beta^2\Omega^2/4}\cdot \sqrt{\CE(\vO,\vO)},
\end{align}
where the last line evaluates the $\rd t, \rd \omega$ integrals using the expression for the Dirichlet form \autoref{lem:Dirichlet} under the Gaussian transition weight $h_\mathsf{G}(\omega): = \e^{-\frac{1}{4}} \e^{-\omega^2\beta^2/2}\ge 0$, $g(t) = \frac{1}{\beta\cosh(2\pi t/ \beta)}\ge 0,$ and $f_{\sigma}(t)=e^{-\sigma^2t^2}\sqrt{\sigma\sqrt{2/\pi}}.$ \\

\noindent \textbf{The High Frequency Component.} For the high frequency $\labs{\omega} > \Omega$ part, we first invoke \autoref{lem:norms-of-commutators} for the following bound on the norms of commutators of operator Fourier transforms of single-site jump operators in 1D:
\begin{equation}
    \lnorm{[\hat{\vA}_{\sigma_1}(\omega),\vO]}_{\vrho} \le c_1\cdot e^{-c_2|\omega|}\cdot  \sum_{j\in [n]}\sum_{\vB\in \CS_j^1}   \|[\vB, \vO]\|_{\vrho}\cdot e^{-c_3\cdot |i-j|\ln |i-j|},
\end{equation}
for suitable $\beta, q$ dependent constants $c_1, c_2, c_3$. To proceed, we integrate over $\omega$:

\begin{align}
    \int_{\labs{\omega}> \Omega}\norm{[\hat{\vA}_{\sigma_1}(\omega),\vO]}_{\vrho} \rd \omega 
    &\leq c_4\cdot e^{-c_2\Omega}\cdot \sum_{j\in [n]}\sum_{\vB\in \CS_j^1}   \|[\vB, \vO]\|_{\vrho}\cdot e^{-c_3\cdot |i-j|\ln |i-j|},
\end{align}
for a suitable choice of constant $c_4$. \\

\noindent \textbf{Combining the Regimes.} To proceed, we consider the Dirichlet form of global generator $\CK$, defined on the full chain, where the collection of jump operators is the set of all single-site Pauli operators $\CS^1_{[n]}$ on $[n]$. Applying the high and low frequency decompositions, 
\begin{align}
    \sum_{\vB\in \CS^1_{[n]}} \|[\vB, \vO]\|_{\vrho}^2 &\leq 
    c_5 \sum_{\substack{i\in [n] \\ \vB\in \CS^1_{i}}} \bigg( e^{\beta^2\Omega^2/4} \sqrt{\CE_{\vB}(\vO, \vO)} +e^{-c_2\Omega}\cdot \sum_{\substack{j\in [n] \\ \vB\in \CS_j^1}}   \|[\vB, \vO]\|_{\vrho}\cdot e^{-c_3\cdot |i-j|\ln |i-j|}\bigg)^2 \\
    &\leq 2\cdot c_5\cdot \bigg( e^{\beta^2\Omega^2/2} \CE(\vO, \vO) + e^{-2c_2\Omega}\cdot\sum_{i\in [n]} \bigg(\sum_{\substack{j\in [n] \\ \vB\in \CS_j^1}}   \|[\vB, \vO]\|_{\vrho}\cdot e^{-c_3\cdot |i-j|\ln |i-j|}\bigg)^2\bigg),\label{eq:put-together-frequencies}
\end{align}
using that $(a+b)^2 \le 2(a^2+b^2)$ pointwise for each $i, \vS$. Our goal is to expose a self-bounding inequality by bringing the second term into the same commutator-square form of $\CK$ on the LHS. Let $\alpha_{j} = \sum_{\vS\in \CS_j^1}\|[\vS, \vO]\|_{\vrho}$ and $f(k) = e^{-c_3\cdot \labs{k}\ln \labs{k}}$, then 
\begin{align}
    \sum_{i\in [n]} \bigg(\sum_{\substack{j\in [n]}}  \alpha_j \cdot f(i-j)\bigg)^2 &\le \sum_{i\in [n]} \bigg(\sum_{\substack{j\in [n]}}  \alpha_j^2 \cdot f(i-j) \bigg) \bigg( \sum_{\substack{j\in [n]}} f(i-j) \bigg) \\
    &\le \bigg(\sum_{k} f(k)\bigg) \cdot \sum_{j\in [n]} \alpha_j^2 \sum_{i\in [n]} f(i-j)\\
    &\le \bigg(\sum_{k} f(k)\bigg)^2 \sum_{j\in [n]} \alpha_j^2.
\end{align}

Therefore, we establish a self-bounding inequality for the Dirichlet form of $\CK$, by choosing a sufficiently large truncation frequency $\Omega$ and a function of $c_2, c_6, \beta$ to suppress the coefficients of the second term:
\begin{align}
    &\sum_{\vA\in \CS^1_{[n]}} \|[\vA, \vO]\|_{\vrho}^2 \leq c_6\cdot \bigg( e^{\beta^2\Omega^2/2} \CE(\vO, \vO) + e^{-2c_2\Omega}\cdot \sum_{\vA\in \CS^1_{[n]}}  \|[\vA, \vO]\|_{\vrho}^2\bigg) \\ \Rightarrow &\sum_{\vA\in \CS^1_{[n]}} \|[\vA, \vO]\|_{\vrho}^2 \leq \CE(\vO, \vO)\cdot \frac{c_6\cdot e^{\beta^2\Omega^2/2}}{1 - c_6 \cdot e^{-2c_2\Omega}},
\end{align}
which concludes the proof.
\end{proof}

To conclude this subsection, we present the deferred claim on a convolution property of operator Fourier transforms. 
\begin{lem}
    [Gaussian Convolution of Operator Fourier Transforms]
    \label{lem:OFT_Gaussian_expansion} For any operator $\vA$ and uncertainties $\sigma_1, \sigma_2, \sigma_3\in \BR^+$ satisfying $1/\sigma_1^2 = 1/\sigma_2^2+1/\sigma_3^2$, we have:
    \begin{equation}
        \hat{\vA}_{\sigma_1}(\omega) =  \sqrt{\frac{\sigma_2\sigma_3}{\sigma_1\sqrt{2\pi}}} \int_{-\infty}^{\infty} \hat{\vA}_{\sigma_3}(\omega, t) e^{-i \omega t}f_{\sigma_2}(t) \rd t.
    \end{equation}
\end{lem}

\begin{proof}
    By definition,

    \begin{align}
    \hat{\vA}_{\sigma_1}(\omega) = \sum_{\nu} \vA_{\nu} \hat{f}_{\sigma_1}(\omega-\nu)
    &= \sqrt{\frac{\sigma_2\sigma_3\sqrt{2\pi}}{\sigma_1}} \sum_{\nu} \vA_{\nu} \hat{f}_{\sigma_3}(\omega-\nu) \hat{f}_{\sigma_2}(\omega-\nu)\\
        &= \sqrt{\frac{\sigma_2\sigma_3\sqrt{2\pi}}{\sigma_1}} \sum_{\nu}(\hat{\vA}_{\sigma_3}(\omega))_{\nu} \hat{f}_{\sigma_2}(\omega-\nu)\\
    & = \sqrt{\frac{\sigma_2\sigma_3\sqrt{2\pi}}{\sigma_1}} \int_{-\infty}^{\infty} \hat{\vA}_{\sigma_3}(\omega)(t) e^{-i \omega t}f_{\sigma_2}(t) \rd t.
\end{align}
The first line is a variance property of Gaussian distributions $\hat{f}_{\sigma_1}(x)\propto \hat{f}_{\sigma_2}(x)\hat{f}_{\sigma_3}(x)$ assuming $1/\sigma_1^2 = 1/\sigma_2^2 +1/\sigma_3^2.$ The last two lines rewrite the sum over the Bohr frequencies as a convolution in the time domain to expose $\vA(\omega,t)$.
\end{proof}

\subsection{Deferred Proofs of the Locality of Operator Fourier Transforms}
\label{section:deferred-dirichlet-comparison}

Here we present the proofs of \autoref{lem:annulus_OFT} and \autoref{lem:norms-of-commutators} (stated above), on the locality in space and frequency of operator Fourier transforms of single-site Pauli operators. Let us begin with the proof of the annulus decomposition of the operator Fourier transform.

\begin{proof}

    [of \autoref{lem:annulus_OFT}]
    By applying \autoref{lem:real_time_locality}, we first write the annulus decomposition for the real-time evolution $ \vA(t) = \sum_\ell \vA_{\ell}(t) - \vA_{\ell-1}(t), $ with 
    \begin{equation}
        \|\vA_{\ell}(t) - \vA_{\ell-1}(t) \| \leq \min \bigg( 2,\frac{(4\labs{t})^\ell}{\ell!} \bigg).
    \end{equation}
Next, we define the Fourier transform of each annulus: 
    \begin{equation}
        \vY_{\ell, \omega, \sigma} := \int \bigg(\vA_{\ell}(t) - \vA_{\ell-1}(t)\bigg) e^{i\omega t}f_\sigma(t) \rd t.
    \end{equation}
    We will prove two bounds on the norm of $ \vY_{\ell, \omega, \sigma}$. The first bound lies in the regime of large $\ell$:
    \begin{align}
        \|\vY_{\ell, \omega, \sigma}\| &\leq \int \bigg\|\vA_{\ell}(t) - \vA_{\ell-1}(t)\bigg\| f_\sigma(t) \rd t \\& \leq \frac{a_1^\ell}{\ell!} \int |t|^\ell\cdot  f_\sigma(t) \cdot \rd t  \leq (a_2 \sigma^{-1})^\ell \cdot \frac{\ell^{\ell/2} }{\ell!} \leq a_3 e^{-a_4\ell\log \ell},
    \end{align}
    where we evaluated the moments of Gaussians. The constants $a_2, a_3, a_4$ depend on $a_1$ and the width $\sigma$. In turn, for the frequency dependence, one can leverage the bounds on the norms of operator Fourier transforms in 1D, as guaranteed by \autoref{lem:bounds_imaginary_conjugation} + \autoref{lem:locality_complextime} applied to $\beta = a_5\cdot \textsf{sign}(\omega)$
    \begin{equation}
        \|\vY_{\ell, \omega, \sigma}\| \leq \bigg\|\int \vA_{\ell}(t)  e^{i\omega t}f_\sigma(t) \rd t\bigg\| + \bigg\|\int \vA_{\ell-1}(t)  e^{i\omega t}f_\sigma(t) \rd t\bigg\| \leq a_4\cdot \exp\bigg(-a_5|\omega| \bigg),
    \end{equation}
    where $a_4, a_5$ depend only on $\sigma, q$. Since $x, y\geq 0:\min(x, y)\leq \sqrt{xy}$ we arrive at the desired bound. 
\end{proof}

Now, we can prove \autoref{lem:norms-of-commutators} on the decay of commutators of operator Fourier transforms. 

\begin{proof}

    [of \autoref{lem:norms-of-commutators}] WLOG suppose that $\vA$ is supported at site $i\in [n]$. In a slight abuse of notation, we let $\CS_{i, \ell}:=\CS_{[i-\ell, i+\ell]}$ be the collection of all multi-qubit Pauli operators with support at distance $\leq \ell$ of $i$, and $\CS^1_{i, \ell} = \CS_{[i-\ell, i+\ell]}^1$ all the single-site $q$-ary Pauli operators at distance $\ell$ of $i$. We begin with the annulus decomposition for the operator Fourier transform $\hat{\vA}_\sigma(\omega) = \sum_\ell \vY_{\ell, \omega}$, as computed in \autoref{lem:annulus_OFT}. We proceed by expanding $\vY_{\ell, \omega, \sigma}$ in the Pauli basis $\CS_{i, \ell}$: 
\begin{equation}
    \vY_{\ell, \omega} = \sum_{\vS\in \CS_{i, \ell}} \alpha_{\vS, \ell, \omega} \vS, \quad\text{where}\quad |\alpha_{\vS, \ell, \omega}|  = 2^{-2q\ell} \cdot |\tr[\vS \vY_{\ell, \omega}]|\leq \|\vY_{\ell, \omega}\|.
\end{equation}
To relate this sum over multi-qubit Pauli's $\CS_{i, \ell}$, to single-qubit Pauli's $\CS_{i, \ell}^1$, one can expand the commutator as follows. Note that each multi-qubit Pauli is a tensor product of single-qubit Pauli's: $\vS = \otimes_j^{u}\vB_j$, with $\vB_j\in \CS_{i, \ell}^1$ and $u = 2\ell+1$.
\begin{align}
    \|[\vS, \vO]\|_{\vrho}&\leq \sum_j \|\prod_{k}^{j-1}\vB_k[\vB_j, \vO]\prod_{k=j+1}\vB_k\|_{\vrho} \tag*{(Expanding the Commutator)}  \\
    &\leq \max_{\vX\in \CS^1_{[n]}} \bigg(\|\vrho^{-1/4}\vX\vrho^{1/4}\|, \|\vrho^{1/4}\vX\vrho^{-1/4}\|\bigg)^{u-1}\cdot \sum_j \|[\vB_j, \vO]\|_{\vrho} \tag*{(KMS Holder's Inequality)} \\
    &\leq a_{\beta, q}^{u}\cdot  \sum_j \|[\vB_j, \vO]\|_{\vrho},  \tag*{(\autoref{lem:locality_complextime})}
\end{align}
where at last we applied the locality of complex time evolution \autoref{lem:locality_complextime} to single-qubit Pauli operators $\vX$, resulting in a constant $a_{\beta, q}$ which depends only on $\beta, q$. Returning to the annulus decomposition, we can now proceed via the triangle inequality to bound the norms of commutators with each component in the annulus decomposition:
\begin{equation}\label{lem:annulus-commutator}
    \|[ \vY_{\ell, \omega}, \vO]\|_{\vrho} \leq  a_{\beta, q}^{2\ell+1}\cdot 4^{2q\ell} \cdot \bigg(\sum_{\vS\in \CS_{i, \ell}^{1}} \|[\vB, \vO]\|_{\vrho}\bigg) \cdot \|\vY_{\ell, \omega}\|.
\end{equation}
To conclude, we proceed again via the triangle inequality to sum over the components in the decomposition. Recall we defined the original operator $\vA$ to lie at site $i\in [n]$. 
\begin{align}
    \|[ \hat{\vA}_\sigma(\omega), \vO]\|_{\vrho}  &\leq   \sum_\ell \|[ \vY_{\ell, \omega}, \vO]\|_{\vrho} \\& \leq e^{-c_3|\omega|}\cdot  \sum_\ell \undersetbrace{\le c_4\cdot e^{-c_2\ell\log \ell /2}}{b^{\ell}\cdot e^{-c_2\ell\log \ell}}   \cdot \bigg(\sum_{\vB\in \CS_{i, \ell+1}^1} \|[\vB, \vO]\|_{\vrho}\bigg) 
    \\ &\leq c_4\cdot e^{-c_3|\omega|}\cdot \sum_{j\in [n]}\sum_{\vB\in \CS_{j}^1} \|[\vB, \vO]\|_{\vrho} \cdot \sum_{\ell \geq d(i, j)} e^{-c_2\ell\log \ell/2} \tag{Exchanging the sums}\\
    &\leq c_5\cdot e^{-c_3|\omega|}\cdot \sum_{j\in [n]}\sum_{\vB\in \CS_{j}^1} \|[\vB, \vO]\|_{\vrho} \cdot e^{-c_2d(i, j)\log d(i, j)/2},
\end{align}
where, in the second inequality we leveraged \eqref{lem:annulus-commutator} and the bound on $\|\vY_{\ell, \omega}\|$ from \autoref{lem:annulus_OFT}; $b$ is the smallest constant s.t. $\forall\ell>1:b^\ell\geq c_1\cdot  a_{\beta, q}^{2(\ell+1)+1}4^{2q(\ell+1)}$, which can be further absorbed by the $\ell \log(\ell)$ decay. In the third inequality, we appropriately pick a constant $c_4$ (as a function of $c_2, b$) to upper bound the exponential factor in $\ell$ by the exponential decay in $\ell\log \ell$. Next, we re-ordered the sum over $\ell, \vB$ in terms of the support $j\in [n]$ of $\vS$ and the distance $d(i, j)$ between the sites. To conclude, we again leverage the superexponential decay in $d(i, j)$. adjusting the constants gives the advertised bound. 
\end{proof}

\section{Proofs of the Main Results}
\label{sec:results}

We dedicate this section to the proofs of the main results by collecting all the ingredients derived in the previous sections. We begin with the proof of the spectral gap for $\CL$ at all constant temperatures.

\begin{proof}

    [of \autoref{thm:main}] We put together the sequence of theorems developed in the previous sections. Fix $\beta\in \BR^+, q\in \BN^+$, and consider any family of 1D local Hamiltonians $\{\vH_n\}_{n\in \BN^+}$.\\

    \noindent \textbf{Decay-of-Correlations.} The starting point of our proof is the statement in \autoref{thm:weak-kimura} \cite{Kimura_2025}, which establishes $\infty$-Weak Clustering for 1D quantum spin chains at all temperatures. By the result in \autoref{thm:kms_from_weak}, $\infty$-Weak Clustering implies KMS-Weak Clustering at all temperatures in 1D. Finally,  KMS-Weak Clustering implies strong clustering for suitable tripartitions under the conditional expectation map $\{\BE_\sR\}_\sR$, \autoref{thm:strong_from_weak}. \\

    \noindent \textbf{The Generator $\CK$ is gapped.} Following the results of \autoref{sec:conditionalgapK}, \autoref{thm:K-is-gapped}, strong-clustering for suitable tripartitions implies that the generator $\CK$, defined by all single-site Pauli jumps, admits a constant spectral gap at all temperatures in 1D. \\

    \noindent \textbf{The Lindbladian $\CL$ is gapped.} By the comparison between the Dirichlet forms of $\CK$ and $\CL$, we conclude the Lindbladian $\CL$ \eqref{eq:exact_DB_L} defined by all single-site Pauli jumps similarly admits a constant spectral gap \autoref{thm:ckg_gap}. \\

    \noindent Finally, we conclude by remarking that all the constants in this paper are bounded for finite $\beta$; and thereby over any compact set $s\in [0, \beta]$ the spectral gap of $\CL_s$ admits a lower bound as a function only of $\beta$ (and $q$).
\end{proof}

\noindent Next, we show that via adiabatic evolution, one can design a low-depth circuit to prepare the purified Gibbs state.

\begin{proof}

    [of \autoref{cor:adiabatic}] We follow the description of the adiabatic algorithm present in  \autoref{thm:adiabatic-algorithm}. The \textit{discriminant} $\vcH_\beta$ resulting from the vectorization of the Lindbladian $\CL_\beta$. 
    \begin{equation}
      - \vrho^{-1/4} \CL_\beta[\vrho^{1/4}\cdot\vrho^{1/4}]\vrho^{-1/4} = \sum_i \vX_i\cdot \vY_i\quad \rightarrow \quad \vcH_\beta  = \sum_i \vX_i\otimes \vY_i^T.
    \end{equation}
    The discriminant is Hermitian, frustration-free, and has as ground state the purified Gibbs state \cite[Proposition I.1]{chen2023efficient}. What is more, the spectral gap of $\vcH_\beta$ is the same as that of $\CL_\beta$. By a consequence of \autoref{thm:main}, we then have that $\vcH_s$ admits a spectral gap over each $s\in [0, \beta]$. Since all the constants as a function of $\beta$ in this paper are bounded, the infimum of the gap along the compact set $[0, \beta]$ is also bounded away from 0 by a function of $\beta$.

    We can thereby prepare the ground state of $\vcH_\beta$ via the adiabatic algorithm along a path of decreasing temperature, from $\beta_0 = 0$ to $\beta$ \cite{Bachmann_2011}. A self-contained description of the algorithm is present in \autoref{sec:adiabatic}. As an immediate corollary of \autoref{thm:adiabatic-algorithm}, there exists a quantum circuit of depth $O(\mathsf{polylog}(n/\epsilon))$ and size $O(n\cdot \mathsf{polylog}(n/\epsilon))$ to prepare the purified Gibbs state, where the implicit constants depend only on $\beta, q$. 
\end{proof}

To conclude, we prove the exponential clustering for Gibbs states of 1D Hamiltonians at all temperatures. 

\begin{proof}
    
    [of \autoref{cor:clustering}] We recall that the existence of a constant spectral gap in a Hamiltonian with exponentially-decaying interactions implies an exponential decay of correlations \cite{hastings2006spectral}.    
    
    In \autoref{lem:lindblad-decaying-interactions},  we prove the discriminant $\vcH_\beta$ of $\CL_\beta$ admits an exponential decay of interactions. From the discussion in the proof of \autoref{cor:adiabatic} above, $\vcH_\beta$ is Hermitian, frustration-free, gapped (by \autoref{thm:main}), and admits the purified Gibbs state as unique ground state. This implies $\vcH_\beta$ admits \cite[Assumption 2.3]{hastings2006spectral}, and consequently, the purified Gibbs state satisfies exponential clustering for an appropriately chosen constant decay rate \cite[Theorem 2.8]{hastings2006spectral}. Finally, we can relate the correlation in the purified Gibbs state to that in the Gibbs state:
    \begin{align}
    \bra{\sqrt{\vrho}} \vX_{\sA}\otimes \vI \cdot \vY_{\sC}\otimes \vI \ket{\sqrt{\vrho}} - \bra{\sqrt{\vrho}} \vX_{\sA}\otimes \vI \ket{\sqrt{\vrho}} \bra{\sqrt{\vrho}} \vY_{\sC}\otimes \vI \ket{\sqrt{\vrho}} = \tr[\vrho\vX_{\sA}\vY_{\sC}] - \tr[\vrho\vX_{\sA}]\tr[\vrho\vY_{\sC}].
\end{align}

\end{proof}

\bibliographystyle{alphaUrlePrint.bst}
\bibliography{ref}

\appendix

\section{Additional Background on the Locality of One-Dimensional Quantum Systems}
\label{section:1D-locality}

We dedicate this appendix to a series of known facts and properties on the locality of the Gibbs states of one-dimensional, short-range, local Hamiltonians, as defined in \autoref{defn:1D}. These bounds are very well understood, and are among the earliest results in mathematical physics \cite{araki1969gibbs}, but we give a minimal and self-contained summary as follows. 

\subsection{Locality of One-Dimensional Complex Time Evolution}

Let $\vO$ be a generic operator defined over the spin chain Hilbert space $\otimes_i^n (\BC^{2^q})$; with support on the interval $\sA:=[a, b]\subseteq [n]$. For any integer $\ell$, recall the $\ell$-radius of $\sA$ is the interval $\sA_\ell:=[a-\ell, b+\ell]$. Now, fix a 1D Hamiltonian $\vH = \sum_\gamma \vH_\gamma$ as in \autoref{defn:1D}. We introduce the restricted Hamiltonian $\vH_{\sA_\ell} = \vH_{\ell}$: 
\begin{align}
    \vH_{\ell} = \vH_{\sA_\ell} := \sum_{\gamma \subset [a-\ell,b+\ell]} \vH_{\gamma}.
\end{align}
We denote complex-time operator dynamics under the truncated Hamiltonian by
\begin{align}
    \vO_{\ell}(z): =e^{i \vH_{\ell}z} \vO e^{-i \vH_{\ell}z} \quad \text{for any}\quad z \in \BC.
\end{align}
For the edge case $\ell = -1$, we define 
\begin{align}
\vO_{-1}(z):= 0.
\end{align}

The well-known Lieb-Robinson bounds quantify the locality of $\vO_\ell(t)$ for real time $t\in\BR$.

\begin{lem}[Locality of Real-Time Evolution]\label{lem:real_time_locality}
For any operator $\vO$ supported on $[a,b]$, $\norm{\vO}\le 1$, the real time evolution for time $t\in \BR$ under a 1D Hamiltonian (\autoref{defn:1D}) satisfies an annulus decomposition:
    \begin{align}
    \vO(t) = \sum_{\ell=0}^{\infty} (\vO_{\ell}(t)-\vO_{\ell-1}(t)) \quad \text{where}\quad
   \big\|\vO_{\ell}(t)-\vO_{\ell-1}(t)\big\| \le \min\bigg(2,\frac{(4\labs{t})^{\ell}}{\ell!}\bigg),
\end{align}
and similarly:
\begin{align}
    \norm{\vO(t)-\vO_{\ell-1}(t)} \le \min \bigg( 2, \frac{(4\labs{t})^{\ell}}{\ell!}\bigg).
\end{align}
\end{lem}
\begin{proof}
    See~\cite[(3.31), Example 3.9]{chen2023speed} and \cite[Theorem 3]{chen2021operator}.
\end{proof}

Next, we present bounds on the locality of complex time evolution in 1D. The original proof is due to Araki~\cite{araki1969gibbs} and later generalized to exponentially decaying interactions~\cite[Theorem 2.3]{perez2023locality}. See also~\cite[Proposition 3.2]{bluhm2022exponential} for a recent usage in studying decay of correlation.

\begin{lem}[Locality of Complex-Time Evolution in 1D]\label{lem:locality_complextime}
For any operator $\vO$ supported on $[a,b]$, $\norm{\vO}\le 1$, the complex-time dynamics with a 1D Hamiltonian (\autoref{defn:1D}) satisfies an annulus decomposition
\begin{align}
    \vO(z) = \sum_{\ell=0}^{\infty} (\vO_{\ell}(z)-\vO_{\ell-1}(z)) \quad \text{where}\quad
   \norm{\vO_{\ell}(z)-\vO_{\ell-1}(z)} \le \frac{(8\labs{z}e^{8\labs{z}})^{\ell}}{\ell!} e^{4\labs{z}(b-a)}.
\end{align}
Therefore, for any $\ell,$ the truncation error is bounded by
\begin{align}
    \norm{\vO(z) - \vO_{\ell-1}(z)} \le e^{4\labs{z}(b-a)}\frac{(8\labs{z}e^{8\labs{z}})^{\ell}}{\ell!} e^{8\labs{z}e^{8\labs{z}}}\label{eq:imaginary-time-truncate}.
\end{align}
The norm of the operator is independent of the system size $n$
\begin{align}
    \norm{\vO(z)} \le e^{8\labs{z}e^{8\labs{z}}}e^{4\labs{z}(b-a)}. \label{eq:imaginary-time-conjugation}
\end{align}
\end{lem}

The above bound only depends on the complex times $z\in \BC$ and is independent of the system size. However, the bound grows doubly-exponentially with the magnitude $z$, and for real-time evolution, the Lieb-Robinson bounds are significantly tighter.

\begin{proof}
We begin with a Taylor expansion of the difference. For $\ell \ge 2,$ and abbreviating the commutator as a superoperator $\CC_{\vH}[\vO]=[\vH,\vO],$
\begin{align}
    \vO_{\ell}(z)-\vO_{\ell-1}(z) &= \sum_{k=0}^{\infty} \frac{(iz)^{k}}{k!}(\CC_{\vH_{\ell}}^k[\vO] - \CC_{\vH_{\ell-1}}^k[\vO])\\
    &=\sum_{k\ge \ell} \frac{(iz)^k}{k!} \sum_{k_{\ell}+\cdots+k_0 = k-\ell} \CC_{\vH_{\ell}}^{k_{\ell}}\CC_{\vH_{\ell}-\vH_{\ell-1}} \cdots \CC_{\vH_1-\vH_0}\CC_{\vH_0}^{k_0}[\vO]. 
\end{align}
Indeed, the locality is monotonically increasing from inner to outer commutators. Now, each commutator can be bounded by 
\begin{align}
\norm{\CC_{\vH_{i}}^{k_{i}}[\vO']} &\le (2\norm{\vH_{i}})^{k_i}\norm{\vO'} \le (2(2i+b-a))^{k_i}\norm{\vO'},\\
\norm{\CC_{\vH_{i}-\vH_{i-1}}[\vO']} &\le 2\cdot \norm{\vH_{i}-\vH_{i-1}}\norm{\vO'},
\end{align}
for each $\vO'.$ Since the number of patterns of the advancing moves $\CC_{\vH_{i}-\vH_{i-1}}$ is equal to the binomial $\binom{k}{\ell},$ 
\begin{align}
    \norm{\vO_{\ell}(z)-\vO_{\ell-1}(z)} &\le \sum_{k\ge \ell} \frac{(2\labs{z})^k}{k!} \binom{k}{\ell} 2^{\ell} (2\ell+b-a)^{k-\ell}\\
    &\le \sum_{k\ge \ell} \frac{(8\labs{z})^k}{k!}  (\ell+(b-a)/2)^{k-\ell}\tag*{(Since $\binom{k}{\ell}\le 2^k$)}\\
    &\le \frac{(8\labs{z})^{\ell}}{\ell!} \sum_{k\ge \ell}  \frac{(8\labs{z})^{k-\ell}}{(k-\ell)!} (\ell+(b-a)/2)^{k-\ell} \tag*{(Since $k! \ge \ell! (k-\ell)!$)}\\
    &= \frac{(8\labs{z})^{\ell}}{\ell!} e^{8\labs{z}(\ell+(b-a)/2)} = \frac{(8\labs{z}e^{8\labs{z}})^{\ell}}{\ell!} e^{4\labs{z}(b-a)}.\label{eq:Oell_Oell}
\end{align}

For the edge cases $\ell=0,1$, we also have 
\begin{align}
\norm{\vO_0(z)}&\le e^{2\labs{z}(b-a)}, \quad 
\norm{\vO_{1}(z)-\vO_{0}(z)} \le 4\labs{z} e^{4\labs{z}(\ell+(b-a)/2)},
\end{align}
which means that the above bound~\eqref{eq:Oell_Oell} remain valid for all $\ell \ge 0.$ Thus, 
\begin{align}
    \norm{\vO(z)-\vO_{\ell-1}(z)} \le \sum_{\ell'\ge \ell}^{\infty} \norm{\vO_{\ell'}(z)-\vO_{\ell'-1}(z)} &\le \sum_{\ell=\ell'}^{\infty}\frac{(8\labs{z}e^{8\labs{z}})^{\ell}}{\ell!} e^{4\labs{z}(b-a)}\\
    &\le e^{4\labs{z}(b-a)}\frac{(8\labs{z}e^{8\labs{z}})^{\ell}}{\ell!} \sum_{\ell'\ge \ell}^{\infty}\frac{(8\labs{z}e^{8\labs{z}})^{\ell-\ell'}}{(\ell'-\ell)!} \\
    &\le e^{4\labs{z}(b-a)}\frac{(8\labs{z}e^{8\labs{z}})^{\ell}}{\ell!} e^{8\labs{z}e^{8\labs{z}}}. 
\end{align}
In particular,
\begin{align}
    \norm{\vO(z)} \le \sum_{\ell=0}^{\infty} \norm{\vO_{\ell}(z)-\vO_{\ell-1}(z)} \le \sum_{\ell=0}^{\infty}\frac{(8\labs{z}e^{8\labs{z}})^{\ell}}{\ell!} e^{4\labs{z}(b-a)} = e^{8\labs{z}e^{8\labs{z}}}e^{4\labs{z}(b-a)},
\end{align}
which concludes the proof.
\end{proof}

\subsection{Araki's Expansionals}

Araki leveraged the locality of complex-time evolution in 1D to relate the Gibbs state of the full 1D lattice, to that with a link in the lattice removed -- decoupling the intervals to the left and right of said link. We review some of these relations, oftentimes referred to as Araki's expansionals, here. 

\begin{lem}
    [Norms of Expansionals, \cite{araki1969gibbs}]\label{lem:expansionals} Fix $z\in \BC$ and a 1D Hamiltonian $\vH$ (\autoref{defn:1D}). Fix two adjacent intervals of the 1D chain $\sA, \sB\subseteq [n]$. Then, 
    \begin{align}
        \|e^{-z\vH_{\sA\sB}} e^{z(\vH_{\sA}+\vH_\sB)}\|\leq \exp( 8\labs{z} e^{\labs{z} e^{8\labs{z}}}e^{4\labs{z}}).
    \end{align}
\end{lem}
Which is triply-exponential in $|z|$, but technically a fixed constant independent of system size. Perhaps unsurprisingly, we will also require an annulus decomposition for the expansional, which quantifies its quasi-locality around the link $\vH^{\sA:\sB}$. We slowly build up to the proof of said quasi-locality. 

\begin{proof}
Consider the following differential equation for any matrix $\vX,\vY$ and $s\in \BR$
\begin{align}
    \frac{\rd }{\rd s} e^{-s\vX}e^{s(\vX+\vY)}  = e^{-s\vX}\vY e^{s(\vX+\vY)} = (e^{-s\vX}\vY e^{s\vX}) e^{-s\vX}e^{s(\vX+\vY)}.
    \end{align}
    Therefore,
    \begin{align}
        \frac{\rd }{\rd s}\lnorm{e^{-s\vX}e^{s(\vX+\vY)}} \le \lnorm{\frac{\rd }{\rd s} e^{-s\vX}e^{s(\vX+\vY)}} \le \lnorm{e^{-s\vX}\vY e^{s\vX}}\cdot \norm{e^{-s\vX}e^{s(\vX+\vY)}}.
    \end{align}
    By Gronwell's lemma, 
    \begin{align}
        \lnorm{e^{-s\vX}e^{s(\vX+\vY)}}:=f(s) \le e^{\CT \int_{0}^s \lnorm{e^{-t\vX}\vY e^{t\vX}} \rd t} \undersetbrace{=1}{\labs{f(0)}}.
    \end{align}
    Integrate from $s=0$ to $s=1$, set $\vX= z (\vH_{\mathsf{A}} + \vH_{\mathsf{B}})$ and $\vY= z \vH_{\mathsf{A:B}}$, and apply ~\autoref{lem:locality_complextime}
\begin{align}
   \lnorm{e^{\beta(\vH_{\sA}+\vH_\sB)} \vH_{\sA:\sB} e^{-\beta(\vH_{\sA}+\vH_\sB)}} \le e^{8\beta e^{8\beta}}e^{4\beta}.
\end{align}
to conclude the proof.
\end{proof}

We further require the following tripartite variant of the above, which in some sense establishes a ``gluing'' lemma for the Gibbs states of consecutive intervals in the lattice. The proof is largely discussed in \cite[Corollary 3.4 (i), Remark 3.5 (ii)']{bluhm2022exponential}. We include a self-contained statement for completeness.

\begin{lem}[Imaginary Time Gluing]\label{lem:gluing}
Fix $z\in \BC$ and a 1D Hamiltonian $\vH$ (\autoref{defn:1D}). Fix three consecutive disjoint intervals $\sA,\sB,\sC\subseteq [n] $. Then, there exists explicit constants $a_{\labs{z}},b_{\labs{z}}>0$ depending only on $\labs{z}$ such that
    \begin{align}
    \lnorm{e^{-z \vH_{\sA\sB\sC}}e^{z \vH_{\sA\sB}}e^{-z\vH_{\sB}}e^{z\vH_{\sB\sC}}-\vI}\le a\cdot \frac{b^{\ell}}{\ell!}, 
    \end{align}
    where the length-scale $\ell = |\sB|$. 
\end{lem}

The above is naturally exactly zero if the Hamiltonian is commuting; in non-commuting 1D Hamiltonians, however, it only remains approximately true provided that the separation $\sB$ is large enough. Furthermore, the above gluing lemma is special to 1D. In higher dimensions, only the case of real-time is known~\cite{haah2020quantum}.

\begin{proof}

[of \autoref{lem:gluing}] Take the derivative w.r.t. $z$, and collect the appropriate Hamiltonian terms on convenient sides of their exponentials -- $\vH_{\sA\sB}$ with $\vH_{\sA\sB\sC}$, and $\vH_{\sB}$ with $\vH_{\sB\sC}$:
    \begin{align}
        \frac{\rd }{\rd z} e^{-z \vH_{\sA\sB\sC}}e^{z \vH_{\sA\sB}}e^{-z\vH_{\sB}}e^{z\vH_{\sB\sC}} &= e^{-z \vH_{\sA\sB\sC}}(\vH_{\sA\sB}-\vH_{\sA\sB\sC})e^{z \vH_{\sA\sB}}e^{-z\vH_{\sB}}e^{z\vH_{\sB\sC}} \\
        &- e^{-z \vH_{\sA\sB\sC}}e^{z \vH_{\sA\sB}}e^{-z\vH_{\sB}}(\vH_{\sB}-\vH_{\sB\sC})e^{z\vH_{\sB\sC}}\\
        &= - e^{-z \vH_{\sA\sB\sC}}e^{z \vH_{\sA\sB}}  \L( e^{-z \vH_{\sA\sB}}( \vH_{\sB:\sC}+\vH_{\sC})e^{z \vH_{\sA\sB}} \R)e^{-z\vH_{\sB}}e^{z\vH_{\sB\sC}}\\
        &+ e^{-z \vH_{\sA\sB\sC}}e^{z \vH_{\sA\sB}}  \L( e^{-z\vH_{\sB}}(\vH_{\sB:\sC}+\vH_{\sC})e^{z\vH_{\sB}}  \R)  e^{-z\vH_{\sB}}e^{z\vH_{\sB\sC}}\\
        &=e^{-z \vH_{\sA\sB\sC}}e^{z \vH_{\sA\sB}}  \L( e^{-z\vH_{\sB}}\vH_{\sB:\sC}e^{z\vH_{\sB}}- e^{-z\vH_{\sA\sB}}\vH_{\sB:\sC}e^{z\vH_{\sA\sB}}  \R) e^{-z\vH_{\sB}}e^{z\vH_{\sB\sC}},
    \end{align}
    where $H_{\sB:\sC}$ denotes the link acting on both $\sB$ and $\sC.$ Introduce factors of $e^{-z \vH_\sC}$ and take the norm
\begin{align}
    &\lnorm{e^{-z \vH_{\sA\sB\sC}}e^{z \vH_{\sA\sB}}e^{z \vH_\sC} e^{-z \vH_\sC} \L( e^{-z\vH_{\sB}}\vH_{\sB:\sC}e^{z\vH_{\sB}}- e^{-z\vH_{\sA\sB}}\vH_{\sB:\sC}e^{z\vH_{\sA\sB}} \R) e^{z \vH_\sC} e^{-z \vH_\sC} e^{-z\vH_{\sB}}e^{z\vH_{\sB\sC}}} \\
    &\le \norm{e^{-z \vH_{\sA\sB\sC}}e^{z \vH_{\sA\sB}}e^{z \vH_\sC}} \cdot \lnorm{e^{-z \vH_\sC} e^{-z\vH_{\sB}}e^{z\vH_{\sB\sC}}} \cdot \lnorm{e^{-z \vH_\sC} \undersetbrace{=:\vO}{\L( e^{-z\vH_{\sB}}\vH_{\sB:\sC}e^{z\vH_{\sB}}- e^{-z\vH_{\sA\sB}}\vH_{\sB:\sC}e^{z\vH_{\sA\sB}}  \R)} e^{z \vH_\sC}}.
\end{align}
The first two terms are controlled precisely by the norms of expansionals~\autoref{lem:expansionals}. In turn, apply \eqref{eq:imaginary-time-truncate} to note that $\vO$ has norm bounded as a function of $|\sB|$. Furthemore, $\vO$ acts only on a single site of $\sC$, thus \eqref{eq:imaginary-time-conjugation} entails a bound on the norm of the last term above.  Integrate the complex variable and use the fact that the expression equals the identity at $z=0$ to conclude the proof.
\end{proof}

We will also require the following minor variant of the gluing identity above in our proof of KMS Weak Clustering from Weak Clustering (\ref{app:Weakclustering}); where the exponential $e^{-\beta \vH}$ is instead placed with square roots on either side of the identity. The proof strategy is the same as the above. 

\begin{lem}[Imaginary Time Gluing II]\label{lem:variant_gluing} In the setting of ~\autoref{lem:gluing}, there exists explicit constants $a_{|z|}, b_{|z|}$ satisfying:
    \begin{align}
        \lnorm{e^{-z \vH_{\sA\sB\sC}/2}e^{z \vH_{\sA\sB}}e^{-z\vH_{\sB}}e^{z\vH_{\sB\sC}}e^{-z \vH_{\sA\sB\sC}/2}-\vI}\le a\cdot \frac{b^{\ell}}{\ell!},
    \end{align}
    where once again $\ell = |\sB|$.
\end{lem}
\begin{proof} Pull down the Hamiltonian terms at suitable sides to cancel out terms with the neighbouring exponentials
    \begin{align}
        \frac{\rd }{\rd z} e^{-z \vH_{\sA\sB\sC}/2}e^{z \vH_{\sA\sB}}e^{-z\vH_{\sB}}e^{z\vH_{\sB\sC}}e^{-z \vH_{\sA\sB\sC}/2} &= \frac{1}{2}e^{-z \vH_{\sA\sB\sC}/2}\undersetbrace{-\vH_{\sB:\sC}-\vH_\sC}{(\vH_{\sA\sB}-\vH_{\sA\sB\sC})}e^{z \vH_{\sA\sB}}e^{-z\vH_{\sB}}e^{z\vH_{\sB\sC}}e^{-z \vH_{\sA\sB\sC}/2} \\
        &- \frac{1}{2}e^{-z \vH_{\sA\sB\sC}/2}e^{z \vH_{\sA\sB}}e^{-z\vH_{\sB}}\undersetbrace{-\vH_{\sB:\sC}-\vH_\sC}{(\vH_{\sB}-\vH_{\sB\sC})}e^{z\vH_{\sB\sC}}e^{-z \vH_{\sA\sB\sC}/2}\\
        &+\frac{1}{2}e^{-z \vH_{\sA\sB\sC}/2}e^{z \vH_{\sA\sB}}\undersetbrace{\vH_{\sA}+\vH_{\sA:\sB}}{(\vH_{\sA\sB}-\vH_{\sB})}e^{-z\vH_{\sB}}e^{z\vH_{\sB\sC}}e^{-z \vH_{\sA\sB\sC}/2} \\
        &-\frac{1}{2}e^{-z \vH_{\sA\sB\sC}/2}e^{z \vH_{\sA\sB}}e^{-z\vH_{\sB}}e^{z\vH_{\sB\sC}}\undersetbrace{\vH_{\sA}+\vH_{\sA:\sB}}{( \vH_{\sA\sB\sC}-\vH_{\sB\sC})}e^{-z \vH_{\sA\sB\sC}/2}.
    \end{align}
Note that the $\vH_{C}$ in the first two terms cancel, and $\vH_{A}$ in the last two terms cancel. Treat the two cases as in~\autoref{lem:gluing} to conclude the proof.
\end{proof}

To conclude this section, we are now able to derive the desired annulus decomposition for Araki's expansionals. 

\begin{cor}[An Annulus Decomposition for Araki's Expansionals] Fix $\beta\in \BR$ and a 1D Hamiltonian $\vH$ (\autoref{defn:1D}). Fix two adjacent intervals of the 1D chain $\sA, \sB\subseteq [n]$. Then, the expansional on the link $\sA:\sB$ admits the annulus decomposition:
    \begin{align}
        e^{\beta\vH_{\sA\sB}} e^{-\beta(\vH_{\sA}+\vH_\sB)} = \sum_{\ell\geq 0} \vO_{\ell} \quad \text{where}\quad \norm{\vO_{\ell}} \le a_{\beta} \cdot \frac{b_{\beta}^{\ell}}{\ell!},
    \end{align}
    $a_\beta, b_\beta>0$ are constants that depend only on $\beta$, and the support of $\vO_\ell$ lies within distance $\ell$ from the link $\sA:\sB$.
\end{cor}
\begin{proof} 
We begin by considering a telescoping sum over increasing $(\ell+1)$ around the link $\sA:\sB$:
\begin{align}
    e^{\beta\vH_{\sA\sB}} e^{-\beta(\vH_{\sA}+\vH_\sB)} = \sum e^{\beta\vH_{\sA_{\ell}\sB_{\ell}}} e^{-\beta(\vH_{\sA_{\ell}}+\vH_{\sB_{\ell}})} - e^{\beta\vH_{\sA_{\ell-1}\sB_{\ell-1}}} e^{-\beta(\vH_{\sA_{\ell-1}}+\vH_{\sB_{\ell-1}})},
\end{align}
where $\sA_\ell, \sB_\ell$ are the subintervals at distance $\leq \ell$ to the left/right of the link $\sA:\sB$, and the edge cases $\sA_{-1}, \sB_{-1} :=\emptyset$. To proceed, we alter $\sA_{\ell}\rightarrow \sA_{\ell-1}$ using the imaginary-time gluing \autoref{lem:gluing} on the regions $\sA' = \sA_{\ell}\setminus \sA_{\ell-1}$, $\sB' = \sA_{\ell-1}$, $\sC' = \sB_\ell$:
\begin{align}
    e^{\beta\vH_{\sA_{\ell}\sB_{\ell}}} e^{-\beta(\vH_{\sA_{\ell}}+\vH_{\sB_{\ell}})} = (1+\vX)e^{\beta\vH_{\sA_{\ell-1}\sB_{\ell}}} e^{-\beta(\vH_{\sA_{\ell-1}}+\vH_{\sB_{\ell}})} 
\end{align}
    \noindent where the norm of the additive correction satisfies, by \autoref{lem:gluing} and \autoref{lem:expansionals}:
    \begin{align}
        \lnorm{\vX e^{\beta\vH_{\sA_{\ell-1}\sB_{\ell}}} e^{-\beta(\vH_{\sA_{\ell-1}}+\vH_{\sB_{\ell}})}} &\le \lnorm{\vX }\lnorm{e^{\beta\vH_{\sA_{\ell-1}\sB_{\ell}}} e^{-\beta(\vH_{\sA_{\ell-1}}+\vH_{\sB_{\ell}})}}\leq a'\cdot \frac{b'^{\ell-1}}{(\ell-1)!},
    \end{align}
    for suitable $\beta$-dependent constants $a', b'$. Repeat for $\sB\rightarrow \sB_{\ell}$, and appropriately define $\ell$ to conclude the proof. 
\end{proof}

\section{Weak Clustering in KMS Norm at all Temperatures}
\label{app:Weakclustering}
In this section, we show that the standard definition of Weak Clustering in the literature -- where the observables are normalized in operator norm (\autoref{defn:clustering}, below) -- implies KMS Weak Clustering (as in \autoref{defn:kms_clustering}) -- where the observables are normalized in KMS norm. Combined with the results of \cite{Kimura_2025}, who proved that arbitrary 1D Hamiltonians (\autoref{defn:1D}) admit Weak Clustering at all temperatures, we conclude that KMS Weak Clustering similarly holds at all temperatures. To state our results, we begin by presenting the relevant definitions. 

\begin{defn}
    [$\infty-$Weak Clustering]\label{defn:clustering} A Gibbs state $\vrho$ is said to satisfy $\infty-$\emph{\textsf{Weak Clustering}} with error $\epsilon_\mathsf{\infty}(\cdot)$, if for every pair of operators $\vX, \vY$ with support on intervals $\sA, \sC \subseteq [n]$:
 \begin{equation}\label{eq:inf-weak-clustering}
     \bigg|\tr[\vX^\dagger\vY\vrho] - \tr[\vrho \vX^\dagger] \cdot \tr[\vrho \vY]\bigg| \leq \|\vX\|\cdot  \|\vY\|\cdot \epsilon_\infty(\mathsf{dist}(\sA, \sC)).
\end{equation}
\end{defn}

To formulate our statements in a self-contained manner, we make the following assumption on the decay of correlations in 1D Hamiltonians. 

\begin{assumption}
    [$\mathsf{WC}_{\infty, q, \beta}$]\label{def:wc-assumption} For any $\beta\in \BR^+, q\in \BN^+$, we say $\mathsf{WC}_{\infty, q, \beta}$ holds with error $\epsilon_{\infty, q, \beta}(\cdot )$ if the Gibbs state at inverse-temperature $\beta$ of every family of 1D Hamiltonians (\autoref{defn:1D}) with local dimension $2^q$ satisfies $\infty-$\emph{\textsf{Weak Clustering}} with said error. 
\end{assumption}

These assumptions are ``downwards contained'' in that $\mathsf{WC}_{\infty, q_1, \beta_1}$ implies $\mathsf{WC}_{\infty, q_2, \beta_2}$, so long as $q_1\geq q_2, \beta_1\geq \beta_2$; simply by rescaling the interaction strengths of the Hamiltonian. Here, we only require the results of \cite{Kimura_2025} for 1D Hamiltonians with short-range interactions, as those are the ones covered by \autoref{defn:1D}.

\begin{thm}
    [Weak Clustering in 1D at all Finite Temperatures, \cite{Kimura_2025}]\label{thm:weak-kimura} For any $\beta\in \BR^+, q\in \BN^+$, there exists constants $\alpha, \gamma>0$ depending only on $\beta, q$ such that $\mathsf{WC}_{\infty, q, \beta}$ holds with sub-exponential error:
    \begin{equation}
        \epsilon_{\infty}(\ell) \leq \alpha\cdot \exp\big(-\ell^{\gamma}\big).
    \end{equation}
\end{thm}

We are now in a position to state the main result of this section.

\begin{thm}
    [\textsf{KMS Weak Clustering} from $\infty-$\textsf{Weak Clustering}]\label{thm:kms_from_weak} For any $\beta\in \BR^+, q\in \BN^+$, assume $\mathsf{WC}_{\infty, 2q, 2\beta}$ holds with error function $\epsilon_{\infty}(\cdot)$. Then, there exists explicit constants $c_1, c_2, c_3>0$ as a function of $\beta, q$ such that the Gibbs state of every 1D Hamiltonian admits \emph{\textsf{KMS Weak Clustering}} (\autoref{defn:kms_clustering}) with error function:
    \begin{equation}
        \epsilon_{\mathsf{KMS}}(\ell) \leq c_1\cdot \bigg( e^{-c_2\ell\log \ell} +  \epsilon_\infty(c_3\cdot \ell)\bigg).
    \end{equation}
\end{thm}

As an immediate corollary of \autoref{thm:kms_from_weak}, when combined with \autoref{thm:weak-kimura} \cite{Kimura_2025}, we prove the sub-exponential decay of correlations when measured in the KMS inner product, as originally presented in \autoref{thm:kms-clustering-uncon}. As previously discussed in Section~\ref{sec:conditionalgapK}, this is the ingredient we start with to prove the mixing time of these systems. Our proof of \autoref{thm:kms_from_weak} is based on an analogous decay-of-correlations property for Gibbs states, known as local indistinguishability. To begin, we dedicate the next subsection to the relevant definitions and their relationship to $\infty-$Weak Clustering.

\subsection{Local Indistinguishability}

We begin with the following definition of local indistinguishability, which roughly speaking captures the fact that ``local expectations can be evaluated locally''. In other words, when evaluating the expectation $\tr[\vO_\sA\vrho]$ of an observable supported on a region $\sA$, it suffices to consider the induced Gibbs state on the surrounding region $\sA\sB$.

\begin{defn}
    [Local Indistinguishability]\label{defn:LI} For any $\beta\in \BR^+, q\in \BN^+$, a 1D Hamiltonian is said to satisfy \emph{\textsf{Local Indistinguishability}} with error function $\epsilon_{\mathsf{LI}}(\cdot)$, if for every set of consecutive disjoint intervals $\sA, \sB, \sC \subseteq [n]$, the marginals of the induced Gibbs states at inverse-temperature $\beta$ on said regions satisfy:
    \begin{equation}
        \bigg\|\tr_{\sB\sC}[\vrho^{\sA\sB\sC}]  - \tr_{\sB}[\vrho^{\sA\sB}]\bigg\|_1\leq \epsilon_{\mathsf{LI}}(\mathsf{dist}(\sA, \sC)).
    \end{equation}
\end{defn}

Again, in the interest of self-containment, we formulate an assumption on the local indistinguishability of families of 1D Hamiltonians. 

\begin{assumption}\label{def:li-assumption}
    For any $\beta\in \BR^+, q\in \BN^+$, we say $\mathsf{LI}_{q, \beta}$ holds with error $\epsilon_{\mathsf{LI}, \beta, q}(\cdot )$ if every family of 
    1D Hamiltonians (\autoref{defn:1D}) satisfies \emph{\textsf{Local Indistinguishability}} (\autoref{defn:LI}) with said error function. 
\end{assumption}

Perhaps unsurprisingly, \textsf{Local Indistinguishability} is a known consequence of $\infty-$\textsf{Weak Clustering} in 1D Hamiltonians. Here, we simply invoke the result of \cite[Proposition 7.1]{bluhm2022exponential}.

\begin{lem}
    [\textsf{Local Indistinguishability} from $\infty-$\textsf{Weak Clustering}, \cite{bluhm2022exponential}]\label{lem:li-from-weak} For any $\beta\in \BR^+, q\in \BN^+$, assume $\mathsf{WC}_{\infty,q,  \beta}$ holds with error $\epsilon_{\infty, q, \beta}$ (\autoref{def:wc-assumption}). Then, there exists explicit constants $c_1, c_2, c_3>0$ as a function of $\beta, q$, such that $\mathsf{LI}_{q, \beta}$ holds with error function:
    \begin{equation}
        \epsilon_{\mathsf{LI}, q, \beta}(\ell) \leq c_1\cdot \bigg( e^{-c_2\ell\log \ell} +  \epsilon_{\infty, q, \beta}(c_3\cdot \ell)\bigg).
    \end{equation}
\end{lem}

The statement and proof presented in \cite[Proposition 7.1]{bluhm2022exponential} is defined over consecutive disjoint intervals $\sA,\sB,\sC$. Due to the folding argument present in \autoref{rmk:folding}, it can be extended to topologies $\sA_1\sB_1\sC\sB_2\sA_2$ with the distance measured by $\min(\sB_1, \sB_2)+1$ at the cost of invoking $\infty-$\textsf{Weak Clustering} on 1D Hamiltonians where the inverse-temperature and number of qubits per-site have doubled.

\subsection{KMS Weak Clustering from Local Indistinguishability}

The main result of this subsection is the following proposition, which entails that local indistinguishability of reduced density matrices (as in \autoref{defn:LI}), implies some form of local indistinguishability of KMS inner products, where furthermore the error is measured in the KMS norm of the observables.

\begin{prop}
    [Local Indistinguishability of KMS Inner Products] \label{prop:kms-from-LI} For any $\beta\in \BR^+, q\in \BN^+$, assume $\mathsf{LI}_{2q, 2\beta}$ holds with error function $\epsilon_{\mathsf{LI}, 2q, 2\beta}(\cdot)$ (\autoref{def:li-assumption}). Then, there exists constants $c_1, c_2, c_3>0$ dependent only on $\beta, q$ such that for every set of consecutive disjoint intervals $\sA, \sB, \sC \subseteq [n]$ and any pair of observables $\vO_\sA, \vO_\sC$:
    \begin{equation}
       \bigg|\tr\L[\vO_{\sA} \vrho^{1/2}\vO^{\dagger}_{\sC}\vrho^{1/2}\R] -  \tr\big[\vrho^{\sA\sB}\vO_{\sA}\big]\cdot \tr\big[\vrho^{\sB\sC}\vO_{\sC}^\dagger\big]\bigg|  \leq c_1\cdot \|\vO_\sA\|_{\vrho}\cdot \|\vO_\sC\|_{\vrho} \cdot \bigg(e^{-c_2\Delta\log \Delta}+ \epsilon_{\mathsf{LI}, 2q, 2\beta}\big(c_3\cdot\Delta\big)\bigg),
    \end{equation}
    where $\Delta = d(\sA, \sC)$.
\end{prop}
Before we prove the proposition above, we show why it implies KMS Weak Clustering. 
\begin{proof}

    [of \autoref{thm:kms_from_weak}] To begin, we invoke \autoref{lem:li-from-weak} to use the fact that for every constant $\beta', q'$, $\infty-$\textsf{Weak Clustering} $\mathsf{WC}_{\infty, q', \beta'}$ implies \textsf{Local Indistinguishability} $\mathsf{LI}_{q', \beta'}$ with related error. We proceed by assuming $\mathsf{LI}_{2q, 2\beta}$ holds with error function $\epsilon_{\mathsf{LI}, 2q, 2\beta}$. 
    
    Now, fix observables $\vO_{\sA}, \vO_{\sC}$ with support on intervals $\sA, \sC$. In light of \autoref{prop:kms-from-LI}, it only remains to relate the expectations under $\vrho$ to that under $\vrho^{\sA\sB}$ up to error in KMS norm. For this purpose, we apply \autoref{prop:kms-from-LI} on $\vO_\sA$ and $\vI$, resulting in 
    \begin{equation}
       \bigg| \tr\L[\vO_{\sA} \vrho\R] - \tr\L[\vO_{\sA} \vrho^{\sA\sB}\R]\bigg| \leq  c_1\cdot \|\vO_\sA\|_{\vrho}\cdot \bigg(e^{-c_2\Delta\log \Delta} + \epsilon_{\mathsf{LI}, 2q, 2\beta}\big(c_3\cdot\Delta\big)\bigg)
    \end{equation}
    and similarly for $\vO_\sC$; where $\Delta = d(\sA, \sC)$. To conclude, we observe by the comparison of measures \autoref{lem:compare_measure}: 
    \begin{equation}
        |\tr\big[\vrho^{\sB\sC}\vO_{\sC}\big]| \leq \|\vO_{\sC}\|_{\vrho^{\sB\sC}} \leq c_{\beta, q, k}\cdot \|\vO_{\sC}\|_{\vrho},
    \end{equation}
    and in this manner, 
    \begin{align}
        &\bigg|\tr\big[\vrho\vO_{\sA}\big]\cdot \tr\big[\vrho\vO_{\sC}^\dagger\big] -  \tr\big[\vrho^{\sA\sB}\vO_{\sA}\big]\cdot \tr\big[\vrho^{\sB\sC}\vO_{\sC}^\dagger\big]\bigg|   \\\leq &\bigg| \tr\L[\vO_{\sA} \vrho\R] - \tr\L[\vO_{\sA} \vrho^{\sA\sB}\R]\bigg| \cdot |\tr\big[\vrho^{\sB\sC}\vO_{\sC}^\dagger\big]| + \bigg| \tr\L[\vO_{\sC}^\dagger \vrho\R] - \tr\L[\vO_{\sC}^\dagger \vrho^{\sC\sB}\R]\bigg|\cdot |\tr\L[\vO_{\sA} \vrho\R]| \\
        \leq & c'\cdot  \bigg(e^{-c_2\Delta\log \Delta}+ \epsilon_{\mathsf{LI}, 2q, 2\beta}\big(c_3\cdot\Delta\big)\bigg)\cdot \|\vO_\sA\|_{\vrho}\cdot \|\vO_\sC\|_{\vrho},
    \end{align}
    adjusting constants and applying \autoref{prop:kms-from-LI} again, then gives the desired claim.
\end{proof}

To conclude this subsection, we present the proof of the Local Indistinguishabilty statement for KMS inner products. 

\subsubsection{Proof of \autoref{prop:kms-from-LI}}
We divide the proof of \autoref{prop:kms-from-LI} into a sequence of three lemmas, which together imply the desired statement by the triangle inequality. We first require a certain gluing statement for the partition functions of intervals of the chain, whose proof can be found in \cite[Step 2, Page 28]{bluhm2022exponential}.

\begin{lem}
    [Ratios of Partition Functions, \cite{bluhm2022exponential}] \label{lem:partition-func-ratio}
    Under the assumptions of \autoref{prop:kms-from-LI}, consider a partition of the 1D chain into consecutive disjoint intervals $[n] = \sA\cup\sB\cup\sC$. Then, there exists constants $c_1, c_2, c_3$ as a function of $\beta, q$ such that: 
     \begin{equation}
         \bigg|\frac{Z_{\sA\sB}\cdot Z_{\sB\sC}}{Z_{\sA\sB\sC}\cdot Z_\mathsf{B}} - 1\bigg| \leq c_1\cdot \bigg( e^{-c_2\ell\log \ell} +  \epsilon_{\mathsf{LI}, q, \beta}(c_3\cdot \ell)\bigg),
     \end{equation}
     where the length-scale $\ell = d(\sA, \sC)$.
\end{lem}

As also commented below \autoref{lem:li-from-weak}, the original proof presented in \cite[Step 2, Page 28]{bluhm2022exponential} consists of a single tripartition; however, the proof similarly extends to $k$-sequences of intervals \eqref{eq:many-intervals} by \autoref{rmk:folding}. 

\begin{lem}
    \label{lem:step-1-kms-li} In the context of \autoref{prop:kms-from-LI}, consider a decomposition of the 1D chain into consecutive disjoint intervals as $[n] := \sL \cup \sA \cup \sB \cup \sC\cup \sR$. Then, there exists constants $c_1, c_2, c_3$ such that for any pair of observables $\vO_\sL, \vO_\sR$ with support on $\sL, \sR$:
    \begin{equation}
       \bigg|\tr\L[\vO_{\sL} \sqrt{\vrho}\vO^{\dagger}_{\sR}\sqrt{\vrho}\R] -  \tr\bigg[\vrho_{\sA\sC}^{-1/2}\vrho_{\sL\sA\sC\sR}^{1/2}\vO_{\sL}\vrho_{\sL\sA\sC\sR}^{1/2}\vrho_{\sA\sC}^{-1/2}\vO_\sR^\dagger\cdot \vrho_{\sA\sB\sC}\bigg]\bigg|  \leq c_1\cdot \bigg(e^{-c_2\ell\log \ell} + \epsilon_{\mathsf{LI}, 2q, 2\beta}(c_3\cdot \ell)\bigg),
    \end{equation}
    with $\ell = \min(|\sA|, |\sC|)$.    
\end{lem}

\begin{proof} 

We first decompose the full Gibbs state on $[n]$ using the imaginary time gluing lemma \autoref{lem:gluing} twice, on $(\sR, \sA, \sB\sC\sL)$ and subsequently on $(\sA\sB, \sC,\sR)$:
\begin{align}
        e^{-\beta \vH/2} &= e^{-\beta \vH_{\sL\sA}/2} e^{\beta \vH_{\sA}/2} e^{-\beta \vH_{\sA\sB\sC\sR}/2}\cdot \vX_1\\
        &= e^{-\beta \vH_{\sL\sA}/2} e^{\beta \vH_{\sA}/2} \bigg(e^{-\beta \vH_{\sC\sR}/2} e^{\beta \vH_{\sC}/2} e^{-\beta \vH_{\sA\sB\sC}/2}\vX_2\bigg)\vX_1\\
        &= e^{-\beta (\vH_{\sL\sA} + \vH_{\sC\sR})/2}e^{\beta (\vH_{\sA} + \vH_{\sC})/2}\cdot e^{-\beta \vH_{\sA\sB\sC}/2}\vX_2 \vX_1.
\end{align}
Observe that for any $\vY$, by KMS Holder's inequality and with $\vY(i\beta/4) = \vrho \vY\vrho^{-1/4}$:
\begin{align}
    &\labs{\tr\L[\vO_{\sL} \sqrt{\vrho}\vY\vO^{\dagger}_{\sR}\vY^{\dagger}\sqrt{\vrho}\R] - \tr\L[\vO_{\sL} \sqrt{\vrho}\vO^{\dagger}_{\sR}\sqrt{\vrho}\R]}
    \le  \bigg( 2 \norm{\vY(i\beta/4) - \vI} + \norm{\vY(i\beta/4) - \vI}^2\bigg) \norm{\vO_\sL}_{\vrho} \norm{\vO_\sR}_{\vrho}.
\end{align}
to apply the above in the context of the gluing lemma, it remains then to compute $\norm{\vY(i\beta/4) - \vI}$ with $\vY = (\vX_2\vX_1)^{-1}$. This, in turn, follows directly from \autoref{lem:variant_gluing}:
\begin{align}
    \|(\vX_2\vX_1)^{-1}-\vI\| &\leq \|\vX_2^{-1}-\vI\|+\|\vX_1^{-1}-\vI\|+\|\vX_2^{-1}-\vI\|\cdot \|\vX_1^{-1}-\vI\|\\
    &\leq a_1\cdot \bigg( \frac{a_2^{|\sA|}}{|\sA|!} +  \frac{a_2^{|\sC|}}{|\sC|!}\bigg),
\end{align}
for appropriate constants $a_1, a_2, a_3$. We thereby have
\begin{align}
    \tr\L[\vO_{\sL} \sqrt{\vrho}\vY\vO^{\dagger}_{\sR}\vY^{\dagger}\sqrt{\vrho}\R] &= \tr\bigg[\vrho_{\sA\sC}^{-1/2}\vrho_{\sL\sA\sC\sR}^{1/2}\vO_{\sL}\vrho_{\sL\sA\sC\sR}^{1/2}\vrho_{\sA\sC}^{-1/2} \vrho_{\sA\sB\sC}^{1/2}\vO_\sR^\dagger\vrho_{\sA\sB\sC}^{1/2}\bigg] \cdot \frac{Z_{\sL\sA\sC\sR}\cdot Z_{\sA\sB\sC}}{Z\cdot Z_{\sA\sC}} \\
    & \approx \tr\bigg[\vrho_{\sA\sC}^{-1/2}\vrho_{\sL\sA\sC\sR}^{1/2}\vO_{\sL}\vrho_{\sL\sA\sC\sR}^{1/2}\vrho_{\sA\sC}^{-1/2} \vrho_{\sA\sB\sC}^{1/2}\vO_\sR^\dagger\vrho_{\sA\sB\sC}^{1/2}\bigg]\\
    & = \tr\bigg[\vrho_{\sA\sC}^{-1/2}\vrho_{\sL\sA\sC\sR}^{1/2}\vO_{\sL}\vrho_{\sL\sA\sC\sR}^{1/2}\vrho_{\sA\sC}^{-1/2}\vO_\sR^\dagger\cdot \vrho_{\sA\sB\sC}\bigg],
\end{align}
the approximation leveraged \autoref{lem:partition-func-ratio} on the ratio of partition functions, on the tripartition $\sA' = \sL\cup\sR, \sB'=\sA\cup\sC, \sC' = \sB.$ By \autoref{rmk:folding}, to address this topology of observables, it suffices to invoke $\mathsf{LI}_{2q, 2\beta}$ where the inverse-temperature and the number of qubits per-site have doubled. The last line recombines the two square-roots of $\vrho_{\sA\sB\sC}$ by commuting through $\vO_{\sR}^{\dagger}$. 
\end{proof}

\begin{lem}  \label{lem:step-2-kms-li}
    In the context of \autoref{lem:step-1-kms-li}, there exists constants $c_1, c_2$ such that for any pair of observables $\vO_\sL, \vO_\sR$ with support on $\sL, \sR$:
    \begin{align}
       \bigg|\tr\L[\vO_{\sL} \vrho_{\sL\sA\sC\sR}^{1/2}\vO^{\dagger}_{\sR}\vrho_{\sL\sA\sC\sR}^{1/2}\R] -  \tr\bigg[\vrho_{\sL\sA\sC\sR}^{1/2}\vO_{\sL}&\vrho_{\sL\sA\sC\sR}^{1/2}\vO_\sR^\dagger\cdot \vrho_{\sA\sC}^{-1/2}\vrho_{\sA\sB\sC}\vrho_{\sA\sC}^{-1/2}\bigg]\bigg|  \\ &\leq c_1\cdot \|\vO_\sL\|_{\vrho}\cdot \|\vO_\sR\|_{\vrho} \cdot \bigg(\frac{c_2^\Delta}{\Delta!} + \epsilon_{\mathsf{LI}, 2q, 2\beta}\big(\Delta\big)\bigg),
    \end{align}
    where $\Delta = \frac{1}{2} \min(|\sA|, |\sC|)$.
\end{lem}

\begin{proof}
    We proceed by expanding the Gibbs state on $\sL\sA, \sC\sR$ on the boundary between $\sL\sA$ and $\sC\sR$, using the expansionals:
    \begin{align}
       e^{-\beta \vH_{\sL\sA\sC\sR}/2}e^{\beta \vH_{\sA\sC}/2} &= e^{-\beta \vH_{\sL\sR}/2} (e^{\beta \vH_{\sL\sR}/2}e^{-\beta \vH_{\sL\sA\sC\sR}/2}e^{\beta \vH_{\sA\sC}/2})\\
       &= e^{-\beta \vH_{\sL\sR}/2} (e^{\beta \vH_{\sL}/2}e^{-\beta \vH_{\sL\sA}/2}e^{\beta \vH_{\sA}/2})  (e^{\beta \vH_{\sR}/2}e^{-\beta \vH_{\sC\sR}/2}e^{\beta \vH_{\sC}/2})\\
       &= e^{-\beta \vH_{\sL\sR}/2} \vY_{\sL\sA} \vY_{\sC\sR}:=  e^{-\beta \vH_{\sL\sR}/2} \vY,
   \end{align}
   where furthermore the expansionals satisfy the annulus decomposition:
   \begin{align}\label{eq:annulus-LI}
      \vY_{\sL\sA} = \sum_{\ell} \vY_{\sL\sA,\ell}\quad& \text{where}\quad \norm{\vY_{\sA\sL,\ell}} \le a_{\beta} \frac{b_{\beta}^\ell}{\ell!}, \quad
      \vY_{\sC\sR} = \sum_{\ell}\vY_{\sC\sR, \ell} \quad \text{where}\quad  \norm{\vY_{\sC\sR, \ell} } \le a_{\beta} \frac{b_{\beta}^\ell}{\ell!} ,\\
      \text{ as does their product: }& \vY:= \sum_\ell \vY_\ell \quad \text{where}\quad \vY_\ell = \vY_{\sL\sA,\ell}\cdot \sum_{k\leq \ell} \vY_{\sC\sR, k}, \quad  \|\vY_\ell\| \leq c_\beta \cdot \frac{d_\beta^\ell}{\ell!}.
  \end{align}
  the support of $\vY_\ell$ is within distance $\ell$ from both cuts $\sL:\sA$ and $\sC:\sR$. We proceed by expressing the desired quantities in terms of expectations over the expansionals:
  \begin{align}
&\tr\bigg[\vrho_{\sL\sA\sC\sR}^{1/2}\vO_{\sL}\vrho_{\sL\sA\sC\sR}^{1/2}\vO_\sR^\dagger\cdot \vrho_{\sA\sC}^{-1/2}\vrho_{\sA\sB\sC}\vrho_{\sA\sC}^{-1/2}\bigg] = \frac{Z_{\sA\sC}}{Z_{\sL\sA\sC\sR}} \cdot \tr\big[\vX\cdot  \vrho_{\sA\sB\sC}\big],  \\ &\tr\L[\vO_{\sL} \vrho_{\sL\sA\sC\sR}^{1/2}\vO^{\dagger}_{\sR}\vrho_{\sL\sA\sC\sR}^{1/2}\R] = \frac{Z_{\sA\sC}}{Z_{\sL\sA\sC\sR}}\tr\big[\vX \vrho_{\sA\sC}\big] \quad \text{ with } \quad \vX = \tr_{\sL\sR}\bigg[\vO_\sL e^{-\beta \vH_{\sL\sR}/2} \vY \vO_\sR^\dagger \vY^\dagger e^{-\beta \vH_{\sL\sR}/2}\bigg].
  \end{align}
  To relate the two, the observation is that the operator $\vX$ is localized on the boundaries $\sL:\sA, \sC:\sR$, and leaks into the regions $\sA, \sC$ under the annulus decomposition. Therefore, in principle, we should be able to leverage local indistinguishability, so long as the distance to the region $\sB$ is sufficiently large. To make this concrete, we first decompose $\vX$ into an annulus decomposition:
  \begin{equation}
      \vX = \sum_{\ell, k} \vX_{\ell, k}, \quad \vX_{\ell, k}:=\tr_{\sL\sR}\bigg[\vO_\sL e^{-\beta \vH_{\sL\sR}/2} \vY_\ell \vO_\sR^\dagger \vY_k^\dagger e^{-\beta \vH_{\sL\sR}/2}\bigg].
  \end{equation}
  We claim, and prove shortly (see below \eqref{eq:schmidt-X}), that the operator norm of each component is bounded by the KMS norm of $\vO_\sL, \vO_\sR$:
  \begin{equation}
      \|\vX_{\ell, k}\| \leq \|\vO_\sL\|_{\vrho}\cdot \|\vO_\sL\|_{\vrho} \cdot Z_{\sL\sR}\cdot \frac{g^{\ell+k}}{\ell!\cdot k!},\label{eq:norm-on-LI-truncation}
  \end{equation}
  for an appropriate choice of constant $g$. We can now proceed by invoking local-indistinguishability at small scales $k, \ell$; and at large scales we leverage the quasi-locality of $\vX_{\ell, k}$. For $\Delta = \frac{1}{2}\min(|\sA|, |\sC|)$, we have:\\

\noindent \textbf{Small $\ell, k$.} If both $\ell, k\leq \Delta$, then the distance of the support of $\vX_{\ell, k}$ to the region $\sB$ is at least $\Delta=\frac{1}{2}\min(|\sA|, |\sC|)$. Under the assumption that local-indistinguishability holds, i.e., $\mathsf{LI}_{2q, 2\beta}$ with error function $\epsilon_{\textsf{LI}, 2\beta, 2q}(\cdot)$, then we have the following bound on the contribution of the small length-scale truncation to $\vX$:
\begin{align}
    \sum_{k, \ell \leq \Delta} \bigg|\tr\big[\vX_{\ell, k}\cdot  \vrho_{\sA\sB\sC}\big] - \tr\big[\vX_{\ell, k}\cdot  \vrho_{\sA\sC}\big]\bigg| &\leq \epsilon_{\textsf{LI}, 2q, 2\beta}\big(\Delta\big)\cdot  \sum_{k, \ell} \|\vX_{k, \ell}\| \\
    &\leq \|\vO_\sL\|_{\vrho}\cdot \|\vO_\sL\|_{\vrho} \cdot Z_{\sL\sR}\cdot  \epsilon_{\textsf{LI}, 2q, 2\beta}\big(\Delta\big) \sum_{k,\ell } \frac{g^{\ell+k}}{\ell!k!} \\
    &\leq  \|\vO_\sL\|_{\vrho}\cdot \|\vO_\sL\|_{\vrho} \cdot Z_{\sL\sR}\cdot  \epsilon_{\textsf{LI}, 2q, 2\beta}\big(\Delta\big)\cdot K
\end{align}
for a suitable constant $K$. \\

\noindent \textbf{Large $\ell, k$.} In the regime that either $\ell\geq \Delta$ or $k\geq \Delta$, we leverage instead the bound \eqref{eq:norm-on-LI-truncation}.
 \begin{align}
     \sum_{\ell \geq \Delta} \sum_k \bigg|\tr\big[\vX_{\ell, k}\cdot  \vrho_{\sA\sB\sC}\big] - \tr\big[\vX_{\ell, k}\cdot  \vrho_{\sA\sC}\big]\bigg| &\leq 2\cdot  \sum_{\ell \geq \Delta} \sum_k\|\vX_{\ell, k}\|\\
     &\leq 2\cdot \|\vO_\sL\|_{\vrho}\cdot \|\vO_\sL\|_{\vrho} \cdot Z_{\sL\sR}\cdot  \sum_{\ell \geq \Delta} \sum_k \frac{g^{\ell+k}}{\ell!\cdot k!} \\
     &\leq 2\cdot \|\vO_\sL\|_{\vrho}\cdot \|\vO_\sL\|_{\vrho} \cdot Z_{\sL\sR}\cdot K\cdot \sum_{\ell\geq \Delta} \frac{g^{\ell}}{\ell!} \\
     &\leq 2K^2\cdot Z_{\sL\sR} \cdot \|\vO_\sL\|_{\vrho}\cdot \|\vO_\sL\|_{\vrho} \cdot \frac{g^\Delta}{\Delta!}.
  \end{align}
  Applying the triangle inequality and the Golden-Thompson inequality to the ratio of partition functions (as similarly done in \autoref{lem:compare_measure}), then gives the following result:
  \begin{align}
     \frac{Z_{\sA\sC}}{Z_{\sL\sA\sC\sR}}\cdot  \bigg|\tr\big[\vX\cdot  \vrho_{\sA\sB\sC}\big] - \tr\big[\vX\cdot  \vrho_{\sA\sC}\big]\bigg| \leq c'\cdot \|\vO_\sL\|_{\vrho}\cdot \|\vO_\sR\|_{\vrho} \cdot \bigg(\frac{g^\Delta}{\Delta!} + \epsilon_{\textsf{LI}, 2q, 2\beta}\big(\Delta\big)\bigg),
  \end{align}
  as advertised. It remains only to justify \eqref{eq:norm-on-LI-truncation}.\\

  \noindent\textbf{The norm of $\vX_{\ell, k}$.} For fixed $\ell$, $\vY_{\ell}$ has support only within distance $\leq \ell$ from the cuts $\sL:\sA$ and $\sC:\sR$. One can then express a Schmidt decomposition of $\vY_{\ell}$ across the bipartition $\sL_\ell\sR_\ell:\sA_\ell\sC_\ell$:
  \begin{align}
      \vY_{\ell} = \sum_{i}^{r_\ell} \vZ_{\sL\sR,i,\ell} \otimes \vZ_{\sA\sC,i,\ell},\quad \text{where}\quad \forall i\in [r_\ell]:\quad  \norm{\vZ_{\sL\sR,i,\ell} \otimes \vZ_{\sA\sC,i,\ell}} \le a_{\beta} \frac{b_{\beta}^\ell}{\ell!}, \label{eq:schmidt-X}
  \end{align}
  the rank of the decomposition $r_\ell\leq d^{\ell}$ for a suitable constant $d$, due to the bound on the dimension of the Hilbert space of the bipartition $\sA_\ell\sC_\ell$. We therefore have the expansion
  \begin{align}
      \vX_{\ell, k} = \sum_{i, j} \tr_{\sL\sR}\bigg[\vO_\sL e^{-\beta \vH_{\sL\sR}/2} \vW_{\sL\sR, i, \ell} \vO_\sR^\dagger \vW_{\sL\sR, j, k}^\dagger e^{-\beta \vH_{\sL\sR}/2}\bigg] \cdot \vW_{\sA\sC, i, \ell} \vW_{\sA\sC, j, k}^\dagger .
  \end{align}
  For each $i, j$ in the expansion, we can expose a KMS norm:
\begin{align}
  &\bigg|\tr_{\sL\sR}\bigg[\vO_\sL e^{-\beta \vH_{\sL\sR}/2} \vW_{\sL\sR, i, \ell} \vO_\sR^\dagger \vW_{\sL\sR, j, k}^\dagger e^{-\beta \vH_{\sL\sR}/2}\bigg]\bigg|\cdot   \| \vW_{\sA\sC, i, \ell} \vW_{\sA\sC, j, k}^\dagger\|  \\
  \leq& Z_{\sL\sR}\cdot  \|\vO_\sL\|_{\vrho_\sL}\cdot \|\vO_\sR\|_{\vrho_\sR}\cdot \bigg\|e^{-\beta \vH_{\sL\sR}/2} \vW_{\sL\sR, i, \ell}e^{\beta \vH_{\sL\sR}/2}\bigg\|\cdot \bigg\|e^{-\beta \vH_{\sL\sR}/2} \vW_{\sL\sR, j, k}e^{\beta \vH_{\sL\sR}/2}\bigg\|\cdot  \|\vW_{\sA\sC, i, \ell}\|\cdot \|\vW_{\sA\sC, j, k}\| 
  \\ \leq & Z_{\sL\sR} \cdot \|\vO_\sL\|_{\vrho_\sL}\cdot \|\vO_\sR\|_{\vrho_\sR}\cdot \bigg( c^{\ell}\cdot c^k\bigg) \cdot \norm{\vW_{\sL\sR,i,\ell} \otimes \vW_{\sA\sC,i,\ell}} \cdot \norm{\vW_{\sL\sR,j, k} \otimes \vW_{\sA\sC,j,k}}
  \\ \leq & Z_{\sL\sR}\cdot \|\vO_\sL\|_{\vrho_\sL}\cdot \|\vO_\sR\|_{\vrho_\sR}\cdot a_\beta \frac{(b_\beta c)^{k+\ell}}{k!\cdot \ell!}.
\end{align}
In sequence, we applied KMS Holder's inequality, then the bounds on norms of complex time evolution of $\ell, k$ local operators \autoref{lem:locality_complextime}, and finally the bound on the norm of the Schmidt components \eqref{eq:schmidt-X}. To conclude, we apply the comparison of measures \autoref{lem:compare_measure} to relate the KMS norm under $\vrho_{\sL\sR}$ to that under $\vrho$, up to a constant. The triangle inequality over the $r_\ell\cdot r_k\leq d^{\ell+k}$ Schmidt components concludes the proof, after appropriately adjusting constants. 

  \end{proof}

We are now in a position to prove \autoref{prop:kms-from-LI}.

\begin{proof}

    [of \autoref{prop:kms-from-LI}] Fixed three consecutive disjoint intervals $\sA, \sB, \sC$, we proceed by further decomposing the $\sB$ region into ``closer to $\sA$'' and ``closer to $\sC$'' segments:
    \begin{equation}
        \sB = \sB_\sA\cup\sB_\sM\cup\sB_\sC,
    \end{equation}
    where the middle piece $\sB_\sM$ consists of a single site, and $|\sB_{\sC}|, |\sB_{\sA}|\geq |\sB|/2 - 1$. 
    
    We can now invoke \autoref{lem:step-1-kms-li} and  \autoref{lem:step-2-kms-li}, where we define the segments $\sL' = \sA$, $\sR'=\sC$, $\sA'=\sB_{\sA}, \sB' = \sB_\sM, \sC'=\sB_\sC$. By design, the length-scale satisfies:
    \begin{equation}
        \ell = \min(|\sA'|, |\sC'|) = \min(|\sB_\sA|, |\sB_\sC|)\geq \frac{1}{2} d(\sA, \sC) -1.
    \end{equation}
    The fact that the theorem is mute if $d(\sA, \sC)\leq 3$ gives the desired bound. 
\end{proof}

\section{Low-Depth Quantum Circuits for Quasi-Adiabatic Evolution}
\label{sec:adiabatic}
For any fixed local Hamiltonian and inverse-temperature $\beta$, a lower bound on the spectral gap of the Lindbladian yields a Gibbs state preparation algorithm: simply evolve under $\CL$, and wait for convergence. Unfortunately, in the absence of sharper mixing-time estimates, the circuit depth to implement such an evolution can scale polynomially with $n$. However, if the family $\{\CL_\beta\}$ admits a spectral gap lower bound over a range of inverse temperatures, one can often obtain a much shorter circuit using techniques from adiabatic quantum computation.

We now give a self-contained exposition of adiabatic algorithms for the purified Gibbs state, assuming the corresponding family of Lindbladians is gapped. Since the argument is the same in any dimension, here we assume a generic Hamiltonian with exponentially decaying interactions defined on a $d$-dimensional lattice $[L]^d$~.

\begin{defn}
    [Exponentially Decaying Interactions]\label{defn:exp-decay} For any $q\in \BN^+$, an integer $L\in \BN^+$, and dimension $d\in \BN^+$, a Hamiltonian $\vH:\otimes ^n_i(\BC^{2^q})\rightarrow \otimes ^n_i(\BC^{2^q})$ on $n=L^d$ qudits is said to have \emph{exponentially-decaying interactions} on the $d$-dimensional lattice if it admits an annulus decomposition of the form:
    \begin{align}
        \vH = \sum_{a, r\in [n]}\big(\vH^{r}_a-\vH^{r-1}_a\big), \quad \text{ where }\quad  \|\vH^{r}_a-\vH^{r-1}_a\|\leq c_1\cdot e^{-c_2\cdot r},
    \end{align}
    for a choice of constants $c_1\geq 0, c_2>0$, and where $\vH^{r}_a$ has non-trivial support only on a ball of radius $\leq r$ around the ``center'' site $a\in [n]$.\footnote{The distance is measured under the Manhattan distance $\mathsf{dist}(x,y)  = \sum_{i=1}^d \labs{x_i-y_i}.$}
\end{defn}

We remark that other definitions of exponential decay of interactions are possible, although the above phrasing in terms of an annulus decomposition will be a recurring theme in this section. A simple consequence of \autoref{defn:exp-decay} is a Lieb-Robinson bound for local observables time-evolved under $\vH$, see e.g. \cite{nachtergaele2010LRB, nachtergaele2018lieb}. In what follows, the presentation in this section can be understood as a formalization of the discussion in \cite[Appendix C]{chen2023efficient}. 

\begin{thm}
    [Low-depth Adiabatic Algorithms]\label{thm:adiabatic-algorithm}
    For each inverse-temperature $\beta\in \BR^+$, constants $q, d\in \BN^+$, and a $d$-dimensional Hamiltonian $\vH$ with exponentially decaying interactions on $n$ qudits of local dimension $2^q$, consider the associated family of Lindbladians $\{\CL_{s\cdot \beta}\}_{s\in [0, 1]}$, defined in $\eqref{eq:exact_DB_L}$ by the set of single-site Pauli jumps $\CS^1_{[n]}$, under the Gaussian transition weight \eqref{eq:Metropolis}. Assume they admit a uniform lower bound on their spectral gap over a range of inverse-temperatures $\min_{s\in [0, 1]}\lambda_{\mathsf{gap}}(\CL_{s\beta})\geq \Delta$.

    Then, there exists a quantum circuit of size $\leq s_1\cdot n\cdot \mathsf{log}^{s_2}(n/\epsilon)$ and depth $\leq d_1\cdot \mathsf{log}^{d_2}(n/\epsilon)$ to prepare the purified Gibbs state at inverse temperature $\beta$, where $s_1, s_2, d_1, d_2$ are constants that depend only on $\beta, q, d, \Delta$.
\end{thm}

The big picture is to consider the ``parent Hamiltonian'' of the family of Lindbladians, which for each $\beta$ is frustration-free, has as its unique ground-state the purified Gibbs state, and by assumption admits a uniform lower bound on its spectral gap. We then follow the quasi-adiabatic framework for gapped ground states of \cite{Bachmann_2011}, which reduces the state preparation task to the Hamiltonian simulation of a time-dependent ``adiabatic operator'', built out of derivatives of the parent Hamiltonian. In order to do so, we prove that each of these Hamiltonians admits (sub-)exponentially decaying interactions and therefore Lieb-Robinson bounds, which in turn implies low-depth time-evolution algorithms using the framework of \cite{haah2020quantum}.

\begin{rmk}
    Although the discussion below is tailored to the Gaussian transition weight and single-site Pauli jumps, we highlight where the relevant modifications would be necessary to generalize the statement.
\end{rmk}

In the next section, we present the relevant ingredients in more detail.

\subsection{The Quasi-Adiabatic Framework}

We refer the reader to \cite[Appendix B.2]{rouze2024efficient} for an explicit calculation of the parent Hamiltonians of the \cite{chen2023efficient} family of Lindbladians, which we review in the next subsection. The first ingredient we require is the following lemma on the quasi-adiabatic framework for gapped ground states \cite{Bachmann_2011}, which states that one can exactly evolve between two ground states via the time-evolution of a time-dependent Hamiltonian, so long as there exists a continuous path of gapped Hamiltonians between the two:

\begin{lem}[{Quasi-Adiabatic Evolution~\cite[Proposition 2.4]{Bachmann_2011}}]\label{lem:quasi_adibatic_W}
    Consider a one-parameter family of Hamiltonians $\vcH_s$ for each $s\in [0, 1]$, and assume the family is uniformly gapped $\min_{s\in [0, 1]}\mathsf{gap}(\vcH_s)\geq \Delta$. Then, for each $s\in [0, 1]$, the associated ground state projection can be generated by the unitary time evolution $\vV(s)$ of a time-dependent Hamiltonian $\vW(s)$:
    \begin{align}
        \frac{\rd}{\rd s}\vV(s) = \ri \vW(s)\vV(s)\quad \text{where}\quad 
    \vW(s):= \int_{-\infty}^{\infty}\rd t\ w(t) \int_0^t \rd u\ \e^{\ri u \vcH_s}\bigg(\frac{\rd}{\rd s}\vcH_s\bigg)\e^{-\ri u \vcH_s},\label{eq:adiabatic-op}
\end{align}
where we adopt the specific choice of weight function \emph{\cite[Lemma 2.3]{Bachmann_2011}}\footnote{Any choice of function $w$ with bounded $L_1$ norm and with Fourier spectra contained in the interval $[-\Delta, \Delta]$ suffices. }
\begin{align}
    w(t) = c\cdot \prod_{k=1}^\infty \bigg(\frac{\sin a_k t}{a_kt}\bigg)^2, \quad \text{ where }\quad a_{k} = a (k \log^2 k)^{-1}\quad \text{for each}\quad k>1
\end{align}
for appropriate choice of constants $a, c$ as a function of $\Delta, q$.
\end{lem}

In the context of our work, we recale $\beta$ and let $\{\vcH_{s\beta}\}_{s\in [0, 1]}$ denote the family of parent Hamiltonians of the Lindbladians $\CL_{s\beta}$ \eqref{eq:exact_DB_L} for varying inverse-temperature, and we assume a uniform lower bound $\Delta$ on the spectral gap in said range of $s$. The ground state of $\vcH_0$ is a set of $n$ EPR pairs, and that of $\vcH_\beta$ is the purified Gibbs state at inverse-temperature $\beta$. In order to generate the adiabatic evolution, it suffices to perform the time-evolution of the time-dependent adiabatic operator $\vW$ guaranteed by \autoref{lem:quasi_adibatic_W}.

We follow the description present in \cite[Appendix C]{chen2023efficient}. In order to perform the Hamiltonian simulation above, we first need to establish basic locality properties for the parent Hamiltonian, and subsequently, the adiabatic operator $\vW$. Fortunately, by adapting a calculation of \cite[Appendix B.2]{rouze2024efficient}, one can easily extract an annulus decomposition for the parent Hamiltonian and its derivatives:

\begin{lem}[Exponentially Decaying Interactions in the Parent Hamiltonian]\label{lem:lindblad-decaying-interactions}
    For each $\beta\in \BR^+, q, d\in \BN^+$, the parent Hamiltonian $\vcH_\beta$ of the Lindbladian $\CL_\beta$ \eqref{eq:exact_DB_L} and its derivatives w.r.t. $\beta$ admit exponentially decaying interactions as in \autoref{defn:exp-decay}. That is, there exists an annulus decomposition $\vcH_\beta = \sum_{a, r\in [n]}\vcH^{a, r}_\beta-\vcH^{a, r-1}_\beta$, where $\vcH^{a, r}_\beta$ has support on a ball of radius $r$ around the site $a\in [n]$, and furthermore:
    \begin{equation}
        \|\vcH^{a, r}_\beta-\vcH^{a, r-1}_\beta\|\leq c_1\cdot e^{-c_2\cdot r}, \quad\text{and}\quad \bigg\|\frac{\rd}{\rd\beta}\big(\vcH^{a, r}_\beta-\vcH^{a, r-1}_\beta\big)\bigg\|\leq c_3\cdot e^{-c_4\cdot r},
    \end{equation}
    for constants $c_1 ,c_2, c_3, c_4$ that depend only on $\beta, q, d$.
\end{lem}
The proof of which is deferred to \autoref{sec:quasi-locality-pH}, and is mostly comprised of adapting components from \cite{rouze2024efficient}. A simple corollary of \autoref{lem:lindblad-decaying-interactions} is a Lieb-Robinson bound for the (time-independent) time evolution of $\vcH^{a, r}_\beta$. Together with the assumption on the spectral gap of the family of parent Hamiltonians $\{\vcH_{s\beta}\}_{s\in [0, 1]}$, the results of \cite{Bachmann_2011} imply that the associated adiabatic operator $\vW$ and its time-evolution $\vV$ satisfy similar locality properties:

\begin{lem}
    [Sub-Exponentially Decaying Interactions in the Adiabatic Operator]\label{lem:adiabatic_LB}
    For each $\beta\in \BR^+, q, d\in \BN^+$ and $s\in [0, 1]$, the adiabatic operator $\vW(s)$ associated to the family of parent Hamiltonians $\{\vcH_{s\beta}\}_{s\in [0, 1]}$ satisfies an annulus decomposition with sub-exponential decay. That is, there exists a decomposition $\vW(s) =\sum_{a, r\in [n]} \vW^{a, r}(s)-\vW^{a, r-1}(s)$, where $\vW^{a, r}(s)$ has support on a ball of radius $2r$ around $a\in [n]$, and furthermore
    \begin{align}
        \forall r>1:\quad\|\vW^{a, r}(s)-\vW^{a, r-1}(s)\|\leq c_1 \cdot \exp\bigg(-c_2\cdot \frac{r}{\log^2 r}\bigg),\label{eq:adiabatic-annulus}
    \end{align}
  for $c_1, c_2$ are constants that depend on $\beta, q, d$ and the gap $\Delta$.
\end{lem}

The proof of which is deferred to \autoref{sec:quasi-locality-adiabatic-evol}, and follows from the discussion in \cite[Section 4]{Bachmann_2011}. Again, as an immediate corollary of \cite[Theorem 4.5]{Bachmann_2011}, the time-evolution of $\vW$ similarly admits a Lieb-Robinson bound, however with sub-exponential decay in $r$. 

Equipped with the above locality properties for the quasi-adiabatic evolution, we are now able to present the main ``gluing'' lemma for the time-evolution of time-dependent Hamiltonians, with sub-exponentially decaying interactions. We follow closely the presentation of \cite[Lemma 6]{haah2020quantum}; the only distinction is that our interactions decay slightly slower. To set up some notation, for any time-dependent Hamiltonian $\vW(s)$, denote the time-ordered exponential as the unique solution to the differential equation:
\begin{align}
    \vV(t) =: \CT[e^{i\vW t}] \quad \text{such that} \quad \frac{\rd}{ \rd t}\vV(t) =  i \vW(t) \vV(t),\quad \vV(0) = \vI\label{eq:time-dep-evol}
\end{align}
in an abuse of notation, we also denote the inverse as $\vV(t)^{-1}=:\CT [e^{-i\vW t}]$.

\begin{lem}
    [{\cite[Lemma 6]{haah2020quantum}} Real-Time Gluing with Sub-Exponential Tails]\label{lem:gluing-time-dep} Consider the time-evolution of the time-dependent adiabatic operator $\vW$ with sub-exponential decay of interactions as in \eqref{eq:adiabatic-annulus}. Let $\sA, \sB, \sC\subseteq [L]^d$ be arbitrary disjoint regions of the lattice, and assume $\mathsf{d}(\sA, \sC):=\ell\geq \ell_0$ for a fixed, sufficiently large constant $\ell_0$. Then, the time-evolution of $\vW$ satisfies the following gluing lemma:
     \begin{align}
            \bigg\|\CT[e^{i\vW^{(\sA\sB\sC)} t}] - \CT[e^{i\vW^{(\sA\sB)} t}] \CT[e^{-i\vW^{(\sB)} t}] \CT[e^{i\vW^{(\sB\sC)} t}] \bigg\| \leq |\sA\sB\sC|\cdot c_1\cdot \exp\bigg(c_2\cdot t - c_3\cdot \frac{\ell}{\log^2 \ell}\bigg),
        \end{align}
        for an appropriate set of constants $c_1, c_2, c_3$ which depend only on $\beta, q, d, $ and $\Delta$.
\end{lem}

\begin{proof}
Following the original proof, we take the derivative w.r.t. $t$ and carefully group canceling terms:
\begin{align}
     &\bigg\|\frac{\rd }{\rd t}  \CT [e^{-i\vW^{(\sA\sB)}t}] \CT[e^{i\vW^{(\sB)}t}] \CT[e^{-i\vW^{(\sB\sC)}t} ]\CT [e^{i\vW^{(\sA\sB\sC)} t}]\bigg\| \\ =&\bigg\| \CT [e^{i\vW^{(\sB)}t} ]\CT [e^{-i\vW^{(\sB\sC)}t}] \big(\vW^{(\sA\sB\sC)} - \vW^{(\sB\sC)}\big)
    -\big(\vW^{(\sA\sB)}-\vW^{(\sB)}\big) \CT [e^{i\vW^{(\sB)}t}] \CT [e^{-i\vW^{(\sB\sC)}t}] \bigg\|\\
    =&\bigg\| \CT[e^{-i\vW^{(\sB\sC)}t}]\vW^{(\sA:\sB\sC)}\CT[e^{i\vW^{(\sB\sC)}t}]- \CT [e^{-i\vW^{(\sB)}t}]\vW^{(\sA:\sB)}\CT [e^{i\vW^{(\sB)}t}] \bigg\|,\label{eq:time-dep-duchamel}
\end{align}
where we denote as $\vW^{(\sX:\sY)}$ the interaction terms contained in $\sX\cup\sY$ and supported on both subsets $\sX, \sY$. 

To proceed, we first leverage the annulus decomposition for the adiabatic operator as given in \autoref{lem:adiabatic_LB}. Suppose we partition $\sB = \sB_\sA\cup\sB_\sC$ into the sites closest to $\sA$, and closest to $\sC$. Let $\ell = \mathsf{d}(\sA, \sC)$. Then, 
\begin{align}
     \|\vW^{(\sA:\sB\sC)}-\vW^{(\sA:\sB_\sA)}\| &\leq \sum_{a\in \sA\cup \sB\cup \sC} \sum_{r\in [n] }\|\vW^{a, r}-\vW^{a, r-1}\|\cdot \mathbb{I}[\mathsf{d}(a, \sA)\leq r, \mathsf{d}(a, \sB_\sC\sC)\leq r] \\
     &\leq\sum_{a\in \sA\cup \sB\cup \sC} \sum_{r\geq \frac{\ell}{4} }\|\vW^{a, r}-\vW^{a, r-1}\|\\
     &\leq  c_1\cdot  |\sA\sB\sC| \cdot e^{-c_2\cdot \ell/\log^2\ell},
\end{align}
for an appropriate choice of constants $c_1, c_2$. In the second inequality, we use the fact that $r\geq \mathsf{d}(a, \sA), \mathsf{d}(a, \sB_\sC\sC)$ and $d(\sA, \sB_\sC\sC)\geq \frac{1}{2}\mathsf{d}(\sA, \sC)$ implies $r\geq \frac{1}{4}\mathsf{d}(\sA, \sC)$. In the last inequality, the sum over $r$ of a (sub)-exponentially decaying function is convergent. We remark that an analogous bound holds for $\vW^{(\sA:\sB)}$.  

Next, we leverage the Lieb-Robinson bound for the time-dependent Hamiltonian $\vW^{(\sB\sC)}$, as a consequence of the sub-exponential decay of interactions \cite[Theorem 4.5]{Bachmann_2011}:
\begin{align}
    \bigg\| \CT[e^{-i\vW^{(\sB\sC)}t}]\vW^{(\sA:\sB_\sA)}\CT[e^{i\vW^{(\sB\sC)}t}]- \CT [e^{-i\vW^{(\sB)}t}]\vW^{(\sA:\sB_\sA)}\CT [e^{i\vW^{(\sB)}t}] \bigg\| &\leq \|\vW^{(\sA:\sB_\sA)}\|\cdot |\sA\sB|\cdot d_1\cdot e^{d_2\cdot t -  \frac{d_3\ell}{\log^2 \ell}} \\&\leq |\sA\sB|^2\cdot d'_1\cdot e^{d_2\cdot t -  \frac{d_3\ell}{\log^2 \ell}} ,
\end{align}
for an appropriate choice of constants. Finally, we integrate the expression for the time-derivative present in \eqref{eq:time-dep-duchamel}, and adequate constants to give the desired bound. 
\end{proof}

We are now finally in a position to prove the main theorem of this section. 

\begin{proof}

    [of \autoref{thm:adiabatic-algorithm}]
    We follow the argument in \cite[Proof of Theorem 1]{haah2020quantum}, which recursively partitions the time-evolution on the $d$ dimensional lattice using \autoref{lem:gluing-time-dep}. Under a sub-exponential Lieb-Robinson bound, the resulting simulation reduces to a depth $c_1\cdot d$ circuit, where each gate is comprised of the (time-dependent) Hamiltonian simulation of patches of $\leq \ell^d$ qubits, with $\ell = c_2\cdot \log \frac{L}{\epsilon}\cdot \mathsf{polyloglog}(\frac{L}{\epsilon})$ for appropriate constants $c_1, c_2$.
For each patch, we apply the high-precision time-dependent Hamiltonian simulation~\cite{chen2023efficient} in circuit depth $\mathsf{poly}(\ell^d, \log\frac{n}{\epsilon}) = \mathsf{polylog} \frac{n}{\epsilon}$, based on block encodings of the parent Hamiltonian and a black-box Hamiltonian simulation~\cite{low2018hamiltonian}, for the suitable subset choice of jumps and truncated Hamiltonian patches.
\end{proof}

\subsection{Quasi-Locality of the Parent Hamiltonian (Proof of \autoref{lem:lindblad-decaying-interactions})}
\label{sec:quasi-locality-pH}

We dedicate this section to the proof of \autoref{lem:lindblad-decaying-interactions}, on the annulus decomposition for the parent Hamiltonian. We recollect that we fix the Gaussian weight \eqref{eq:Metropolis} and the collection of single-site Pauli operators $\CS^1_{[n]}$ as jump operators. The explicit form of the parent Hamiltonian can be written as:
\begin{align}
    \vcH_\beta = \sum_{a \in \CS^1_{[n]}}\vec{\CS}_a + \frac{1}{2}(\vN_a\otimes \vI + \vI \otimes \vN^*_a),
\end{align}
where the only dependence on $\beta$ is isolated in the rescaled Heisenberg dynamics:
\begin{align}
    \vec{\CS}_a &=  \iint_{-\infty}^{\infty} h_1(t_1)h_2(t_2) \cdot  \vA^a\L(\beta(t_2-t_1)\R)\otimes\vA^{a}\L(-\beta(t_1+t_2)\R)^T \rd t_1\rd t_2, \\
    \vN_a &:=\iint_{-\infty}^{\infty}n_1(t_1)n_2(t_2)\cdot \vA^{a\dagger}(\beta (t_2-t_1))\vA^a(-\beta (t_2+t_1))\rd t_1 \rd t_2,
\end{align}
and finally the choice of weight functions $n_1,n_2, h_1,h_2$ are all decaying (at least) exponentially:
\begin{align}
h_1(t)&= \frac{2}{\pi} e^{ -2t^2}, \quad h_2(t)= e^{-1/4}\cdot  e^{-4t^2},\\ 
     n_1(t) &:=\frac{2\sqrt{\pi}}{3}\cdot  \L(\!\frac{1}{\cosh\L(2\pi t\R)\!}*_t  e^{-2 t^2}\!\R), \quad
 n_2(t) := \frac{3}{\pi\sqrt{\pi}}\exp\L(-4 t^2-2\ri t \R),
\end{align}
normalized such that $\norm{n_1}_1,\norm{n_2}_1,\norm{h_1}_1,\norm{h_2}_1 \le 1.$ 

We divide the proof of \autoref{lem:lindblad-decaying-interactions} into two claims, which establish the annulus decomposition for $\vcH_\beta$ and for its derivative w.r.t. $\beta$. For completeness, we begin with the annulus decomposition for $\vcH_\beta$, although it is all but present in \cite[Appendix B.2]{rouze2024efficient}. Henceforth, we restrict our attention to a fixed single-site jump operator, $\vA^a\equiv \vA$. Following \autoref{section:1D-locality} in the 1D setting, we denote as  
\begin{equation}
    \vA^\ell(t):= e^{i\vH_\ell t}\vA e^{-i\vH_\ell t}, \quad \vH_\ell = \sum_{\sX\subseteq \CB_a(\ell)} \vH_\sX,
\end{equation}
the real time evolution of $\vA$ under the truncation of the Hamiltonian to the ball $\CB_a(\ell)$ of radius $\ell$ around $a$. Then,
\begin{lem}
    In the context of \autoref{lem:lindblad-decaying-interactions}, the parent Hamiltonian $\vcH^a_\beta$ associated to the single-site jump operator $\vA$ admits the annulus decomposition:
    \begin{equation}
        \vcH^a_\beta = \sum_{ r\in [n]}\vcH^{a, r}_\beta-\vcH^{a, r-1}_\beta, \quad \|\vcH^{a, r}_\beta-\vcH^{a, r-1}_\beta\|\leq c_1\cdot e^{-c_2\cdot r},
    \end{equation}
    where $\vcH^{a, r}_\beta$ has support on a ball of radius $r$ around $a$,
    and $c_1, c_2$ are appropriate constants dependent only on $d, q, \beta$.
\end{lem}

\begin{proof}
    We present the proof for the $\vcS_a$ term, and the $\vN_a$ component is analogous. We begin by telescoping $\vcS_a = \sum_\ell \vS_a^{\ell}$ with 
    \begin{align}
    \vS^{\ell}_a=\iint_{-\infty}^{\infty} h_1(t_1)h_2(t_2) \cdot  \bigg(&\vA^{\ell}\L(\beta(t_2-t_1)\R)\otimes\vA^{\ell}(-\beta(t_1+t_2))^T\\&- \vA^{\ell-1}\L(\beta(t_2-t_1)\R)\otimes\vA^{\ell-1}\L(-\beta(t_1+t_2)\R)^T\bigg) \rd t_1\rd t_2.
\end{align}
We proceed via the triangle inequality applied to the tensor product, and by introducing a cutoff time $T$:
\begin{align}
     \norm{\vS_a^{\ell}}\leq \iint_{-\infty}^{\infty} h_1(t_1)h_2(t_2)& \bigg(\lnorm{\vA^{\ell}\L(\beta(t_1+t_2)\R)-\vA^{\ell-1}\L(\beta(t_2+t_1)\R)}\\&+\lnorm{\vA^{\ell}\L(\beta(t_2-t_1)\R)-\vA^{\ell-1}\L(\beta(t_2-t_1)\R)}\bigg)\rd t_1\rd t_2.\\
    &\leq c_1 \iint_{-\infty}^{\infty} h_1(t_1)h_2(t_2)\cdot \min\bigg(2, \exp\big(c_2\beta(|t_1|+|t_2|) - c_3\cdot \ell\big)\bigg)\rd t_1\rd t_2. \\
    &\leq c'_1\bigg(e^{-c_4\cdot T} + e^{c_5\beta\cdot T - c_3\cdot \ell}\bigg),
\end{align}
where in the second line, we applied a Lieb-Robinson bound for the time-independent time-evolution of the jump operator $\vA$, under the evolution of the original Hamiltonian $\vH$ with exponentially decaying interactions; in the third, we evaluate the integral over the regime $|t_1|+|t_2|\leq T$ and $\geq T$ separately. In the above, the constants depend only on the lattice dimension $d$ and the local dimension $2^q$. Finally, an appropriate choice of $T = c_3\cdot \ell/(c_4+c_5\beta)$ then gives the desired bound. 
\end{proof}

We are now in a position to compute the annulus decomposition for the derivative of the parent Hamiltonian.

\begin{lem}
    In the context of \autoref{lem:lindblad-decaying-interactions}, the derivatives of the parent Hamiltonian $\vcH^a_\beta$ associated to the jump operator $\vA$ admits the annulus decomposition:
    \begin{equation}
        \frac{\rd }{\rd\beta}\vcH^a_\beta = \sum_{ r\in [n]}\frac{\rd }{\rd\beta}\bigg(\vcH^{a, r}_\beta-\vcH^{a, r-1}_\beta\bigg), \quad \text{where}\quad\bigg\|\frac{\rd }{\rd\beta}\bigg(\vcH^{a, r}_\beta-\vcH^{a, r-1}_\beta\bigg)\bigg\|\leq c_1\cdot e^{-c_2\cdot r},
    \end{equation}
    where $\frac{\rd}{\rd\beta}\vcH^{a, r}_\beta$ has support on a ball of radius $r$ around $a$, and $c_1, c_2$ are appropriate constants dependent only on $d, q, \beta$.
\end{lem}

\begin{proof}
    Following the proof of the annulus decomposition for $\vcH^a_\beta$, let us analyze the derivative of the $\ell$th component of $\vcS_a$:
    \begin{align}
        \frac{\rd}{\rd \beta} \vcS_a^\ell = \iint_{-\infty}^{\infty} &h_1(t_1)h_2(t_2) \cdot  \bigg((t_2-t_1)[\vH^\ell, \vA^{\ell}\L(\beta(t_2-t_1)\R)]\otimes\vA^{\ell}(-\beta(t_1+t_2))^T\\
        &-(t_1+t_2)\vA^{\ell}\L(\beta(t_2-t_1)\R)\otimes[\vH^\ell, \vA^{\ell}(-\beta(t_1+t_2))^T]
        - \text{[The same, for $\ell-1$]} \bigg)\rd t_1\rd t_2.\label{eq:derivative-annulus-exp}
    \end{align}
    To proceed, we observe that the Lieb-Robinson bound similarly applies to the commutators, up to a factor of $\ell^d$ which can be absorbed into the exponential decay in $\ell$:
    \begin{align}
        &\bigg\|[\vH^\ell, \vA^{\ell}(\beta t)] - [\vH^{\ell-1}, \vA^{\ell-1}(\beta t)]\bigg\| = \bigg\|[\vH^{\ell}, \vA]^{\ell}(\beta t)-[\vH^{\ell-1}, \vA]^{\ell-1}(\beta t)\bigg\| \label{eq:commutator-lb-deriv} \\
        \leq& \bigg\|[\vH^{\ell}, \vA]-[\vH^{\ell/2}, \vA]\bigg\|+\bigg\|[\vH^{\ell/2}, \vA]-[\vH^{\ell-1}, \vA]\bigg\|+\bigg\|[\vH^{\ell/2}, \vA]^{\ell}(\beta t)-[\vH^{\ell/2}, \vA]^{\ell-1}(\beta t)\bigg\| \\
        \leq &\quad c_1'\cdot \ell^d \cdot e^{-c_2'\ell} + \|[\vH^{\ell/2}, \vA]\|\cdot \ell^d\cdot  e^{c_4'\cdot t-c_5'\cdot \ell} \leq c_1\cdot e^{c_2\cdot t-c_3\cdot \ell}.
    \end{align}
    where we first observe that the Hamiltonian terms centered away from $a$  contribute negligibly to the commutator, due to its exponentially decaying interactions, and the terms $\leq \ell/2$ close to $a$ are controlled via the Lieb-Robinson bound for $\vH$. $c_1, c_2, c_3$ are all constants that depend only on $d, q$. We remark the commutator in \eqref{eq:commutator-lb-deriv} also admits the trivial lower bound $\leq c_4$, as a function of $d, q$, again due to the exponential decay of interactions and the fact that $\vA$ is single-site.
 
 Applied to each component of \eqref{eq:derivative-annulus-exp} then gives
    \begin{align}
        \bigg\|\frac{\rd}{\rd \beta} \vcS_a^\ell\bigg\| &\leq c_1' \iint_{-\infty}^{\infty} h_1(t_1)h_2(t_2) \cdot \big(|t_1|+|t_2|\big)\cdot \min\bigg(c_4, \exp\big(c_2(|t_1|+|t_2|)-c_3\ell\big)\bigg)\rd t_1\rd t_2 \\
        \\ &\leq \quad g_1\cdot e^{-g_4\cdot T} + g_2\cdot \exp\big(g_3\cdot T-g_5\ell\big).
    \end{align}
    under a similar case division over $|t_1|+|t_2|\geq T$ and $\leq T$ as similarly done in the proof of the annulus decomposition. To conclude, we make a suitable choice of $T=\alpha\cdot \ell$. 
\end{proof}

The two lemmas above establish the two conditions in \autoref{lem:lindblad-decaying-interactions}.

\subsection{Quasi-Locality of the Adiabatic Evolution (Proof of \autoref{lem:adiabatic_LB})}
\label{sec:quasi-locality-adiabatic-evol}

We dedicate this section to the proof of \autoref{lem:adiabatic_LB}. Ultimately, we follow closely the discussion in \cite[Lemma 4.7]{Bachmann_2011}, which derives the annulus decomposition, and \cite[Theorem 4.5]{Bachmann_2011}, which derives the Lieb-Robinson bound; however, our choice of annulus decomposition is slightly different. This is since it is not immediately clear that the choice made in \cite[Equation 4.30]{Bachmann_2011} is explicitly computable, which we will later require to perform the Hamiltonian simulation.

\begin{proof}

    [of \autoref{lem:adiabatic_LB}]
    Fix a single-site jump operator $a\in \CS^1_{[n]}$; in a slight abuse of notation, we use $a$ to refer to the site/location $\in [n]$ as well. Let $\CB_a(r)\subseteq [L]^d$ denote the ball of radius $r$ around $a$. Recall the annulus decomposition to the parent Hamiltonian as guaranteed by \autoref{lem:lindblad-decaying-interactions}; we define the truncation to $\CB_a(r)$ to contain all the terms in the decomposition contained entirely within $\CB_a(r)$:
    \begin{equation}
        \vcH^{\CB_a(r)}_\beta = \sum_{i\in [n]} \sum_{\ell:\CB_i(\ell) \subseteq \CB_a(r)} \vcH^{i, \ell}_\beta-\vcH^{i, \ell-1}_\beta.
    \end{equation}
    We are now in a position to define the annulus decomposition to $\vW$. For $a\in\CS^1_{[n]}$, $r\geq 1$, we define 
    \begin{align}
\vW^{a, r}:=         \int_{-\infty}^{\infty}\ w(t) \int_0^t \bigg( &\e^{\ri u \vcH_{s\beta}^{\CB_a(2r)}}\bigg(\frac{\rd}{\rd s}\vcH_{s\beta}^{a, \CB_a(r)}\bigg)\e^{-\ri u \vcH_{s\beta}^{\CB_a(2r)}} \\&- \e^{\ri u \vcH_{s\beta}^{\CB_a(2r-2)}}\bigg(\frac{\rd}{\rd s}\vcH_{s\beta}^{a, \CB_a(r-1)}\bigg)\e^{-\ri u \vcH_{s\beta}^{\CB_a(2r-2)}}\bigg) \rd u \rd t.\label{eq:war-full}
    \end{align}
    In the above, we are truncating the derivative terms to radius $r$, and the time-evolution to radius $2r$. Later, this ``gap'' will allow us to leverage a Lieb-Robinson bound. To bound the norm of the above, we follow the approach \cite[Lemma 4.7]{Bachmann_2011}. Let $T>0$ denote a threshold time-scale. \\
    
   \noindent \textbf{In the regime of large} $t$, we can use an a priori sub-exponential bound on the (decaying) absolute value of $w(t)$ from \cite[Lemma 2.3]{Bachmann_2011}, and the fact that the derivatives of the parent Hamiltonian have norm bounded by a constant:
    \begin{align}
         \int_{|t|\geq T}\rd t |w(t)| \int_0^t \rd u\cdot \big\|\text{Integrand of }\eqref{eq:war-full}\big\| &\leq \max_r\bigg\|\frac{\rd}{\rd s}\vcH_{s\beta}^{a, \CB_a(r)}\bigg\|\cdot \bigg|  \int_{|t|\geq T}\rd t \cdot w(t) \cdot t\bigg| \\&\leq c_1\cdot \bigg|  \int_{|t|\geq T}\rd t \cdot w(t) \cdot t\bigg| \tag{\autoref{lem:lindblad-decaying-interactions}} \\
        &\leq c_2 \cdot \bigg| \int_{|t|\geq T}\rd t \cdot  t^2 \exp\bigg( - \frac{c_3t} {\log^2 c_4t}\bigg)\bigg| \\&\leq c_5 \exp\bigg( -\frac{c_6\cdot T} {\log^2 c_3T}\bigg),
    \end{align}
    under an appropriate set of constants $c_1\cdots c_7$.\\

    \noindent \textbf{In the regime of small} $t$, we can apply the annulus decomposition of the parent Hamiltonian \autoref{lem:lindblad-decaying-interactions}, and subsequently the associated exponential Lieb-Robinson bound for the time-evolution of $\vcH$ as guaranteed by e.g. \cite{nachtergaele2010LRB}, given that the support of $\frac{\rd}{\rd s}\vcH_{s\beta}^{a, \CB_a(r)}$ is at distance $\geq r$ from the boundary of $\CB_a(2r)$.
    \begin{align}
         &\int_{|t|\leq T}\rd t |w(t)| \int_0^t \rd u\cdot \big\|\text{Integrand of }\eqref{eq:war-full}\big\| \leq T^2\cdot \bigg\|\frac{\rd}{\rd s}\vcH_{s\beta}^{a, \CB_a(r)} - \frac{\rd}{\rd s}\vcH_{s\beta}^{a, \CB_a(r-1)}\bigg\| + \\
         & \int_{|t|\leq T}\rd t \int_0^t \rd u\bigg\|\e^{\ri u \vcH_{s\beta}^{\CB_a(2r)}}\bigg(\frac{\rd}{\rd s}\vcH_{s\beta}^{a, \CB_a(r)}\bigg)\e^{-\ri u \vcH_{s\beta}^{\CB_a(2r)}} - \e^{\ri u \vcH_{s\beta}^{\CB_a(2r-2)}}\bigg(\frac{\rd}{\rd s}\vcH_{s\beta}^{a, \CB_a(r)}\bigg)\e^{-\ri u \vcH_{s\beta}^{\CB_a(2r-2)}}\bigg\| \\
         \leq & \quad T^2 \cdot d_1\cdot e^{-d_2\cdot r} + d_3'\cdot r^d \exp\bigg( d_4'\cdot T - d_5'\cdot r\bigg) \leq d_3\cdot \exp\bigg( d_4\cdot T - d_5\cdot r\bigg) ,
    \end{align}

    under an appropriate choice of constants. \\

    \noindent \textbf{Putting together} with an appropriate choice of time-scale $T = \alpha\cdot r$ and appropriately adjusting constants gives the desired sub-exponential bound on the norms of terms in the annulus decomposition. 
\end{proof}

\end{document}